\documentclass[aps,10pt,showpacs,pra,amsmath,amssymb,superscriptaddress,longbibliography,twocolumn]{revtex4-2}
\usepackage{xcolor}
\usepackage{cancel}
\usepackage{xfrac}
\usepackage{enumerate}
\usepackage{dsfont}
\newcommand{\inlshort}{INL -- International Iberian Nanotechnology Laboratory, Braga, Portugal} 
\newcommand{\ulmshort}{Institut für Theoretische Physik, Universität Ulm, Ulm, Germany}
\newcommand{\southcut}{School of Mathematics, South China University of Technology, GuangZhou 510640, China} 
\newcommand{\cfumshort}{Centro de F\'{i}sica, Universidade do Minho, Braga, Portugal} 

 \newcommand{\hdu}{School of Science, Hangzhou Dianzi University, Hangzhou 310018, China} 
\usepackage{amsthm}
\usepackage{mathtools}
\usepackage{physics}
\usepackage{bbm}
\usepackage{amsfonts}
\usepackage{amssymb}
\usepackage{graphicx}

\usepackage[T1]{fontenc}
\usepackage{bbold}
\usepackage{float}
\usepackage[utf8]{inputenc}
\usepackage{csquotes}

\usepackage{tikz-cd}
\usepackage[normalem]{ulem}
\usetikzlibrary{shapes, shapes.geometric, shapes.symbols, shapes.arrows, shapes.multipart, shapes.callouts, shapes.misc,decorations.pathmorphing}
\tikzset{snake it/.style={decorate, decoration=snake}}
\usepackage{tcolorbox}

\usepackage[colorlinks=true,linktocpage=true]{hyperref}
\hypersetup{
    allcolors  = {blue},
}

\theoremstyle{plain}
\newtheorem*{main theorem}{Main Theorem}
\newtheorem{theorem}{Theorem}[section]
\newtheorem{corollary}[theorem]{Corollary}
\newtheorem{lemma}[theorem]{Lemma}
\newtheorem{proposition}[theorem]{Proposition}

\newtheorem{example*}[theorem]{Example*}
\newtheorem{examples*}[theorem]{Examples*}

\newtheorem{remark*}[theorem]{Remark*}

\newtheorem*{search problem}{Search Problem}

\newtheorem{definition}[theorem]{Definition}

\definecolor{selectiveyellow}{rgb}{0.8, 0.0, 0.8}

\usepackage{array}
\usepackage{multirow}
\usepackage{xcolor}

\usepackage{tikz}

\usepackage{hyperref}
\usepackage{graphicx}
\usepackage{mathrsfs}

\begin{document}

\title{Multi-state imaginarity and coherence in qubit systems}

\author{Mao-Sheng Li}
\affiliation{\southcut}

\author{Rafael Wagner}
\email{rafael.wagner@uni-ulm.de}
\affiliation{\ulmshort}
\affiliation{\inlshort}
\affiliation{\cfumshort}
\author{Lin Zhang}
\affiliation{\hdu}

\date{\today}

\begin{abstract}
    Traditionally, the characterization of quantum resources has focused on individual quantum states. Recent literature, however, has increasingly explored the characterization of resources in \emph{multi}-states---ordered collections of states indexed by a varying parameter. In this work, we provide a unitary‑invariant framework to pinpoint imaginarity and coherence in sets of qubit states: we prove that Bloch vectors must be coplanar to be imaginarity‑free and colinear to be incoherent, yielding exact rank‑based tests of coherence and imaginarity, and closed‑form bounds for existing robustness quantifiers, all based on two-state overlaps only. We also show that the set of imaginarity‑free multi‑states is not convex, and that third-order invariants completely characterize multi-state imaginarity of single-qubits but not of higher-dimensional systems. As our main technical result, we show that \emph{every} Bargmann invariant of single-qubit states is determined (up to conjugation) by two‑state overlaps. Beyond qubits, we give purity and system‑agnostic coherence witnesses from equality constraints on higher‑order invariants and connect our results to practical protocols---characterization of partial distinguishability, spin‑chirality detection, and subchannel discrimination. 
\end{abstract}

\maketitle

\section{Introduction}

Quantum resources~\cite{chitambar2019quantum,gour2024resourcesquantumworld,coecke2016mathematical} are typically defined and characterized with respect to individual quantum states or quantum channels. While this has been a successful approach, it arguably overlooks many well-known phenomena that call for a concrete understanding of the \emph{collective} properties of quantum states. Numerous examples in the literature, spanning a wide range of applications in quantum information science, support this perspective. 

One example is the collective behavior of a set of single photons entering a linear-optical multi-port interferometer~\cite{shchesnovich2015partial,shchesnovich2018collective,jones2020multiparticle,jones2023distinguishability}: In this case, the \emph{indistinguishability} of the states encoding each photon's internal degrees of freedom decisively impacts the output‑port statistics of the interferometer. Another example is provided by \emph{quantum pseudo-randomness}, which is a property defined relative to a parameterized set of quantum states $\{\rho_i\}_i$~\cite{ji2018pseudorandom,kretschmer2021quantumpseudo,bansal2025pseudorandom}. One can then investigate the resources present in said collections of states, and infer intriguing aspects about the relationship between pseudorandomness and the amount of quantum resources present in such sets~\cite{gu2024pseudomagic,haug2025pseudorandom,tanggara2025neartermpseudorandompseudoresourcequantum,aaronson2023quantumpseudoentanglement}. Yet another example, from the perspective of quantum channels, is the notion of \emph{non-Markovianity}~\cite{rivas2014quantum}. This is best understood as a property depending on a fine-grained, time-dependent evolution (i.e., of a collection of maps $\{\mathcal{E}_t\}_t$) rather than a coarse-grained description by a single quantum map.

This plurality of examples is motivating the study of quantum resources defined not by individual elements but by \emph{sets} or \emph{collections} of quantum objects~\cite{designolle2021set,miyazaki2022imaginarityfree,salazar2022resource,buscemi2020complete,uola2019quantifying,martins2020quantum,ducuara2020multiobject,selby2023contextuality,wagner2024inequalities,galvao2020quantum,zhang2025quantifierswitnessesnonclassicalitymeasurements,zhang2025reassessingboundaryclassicalnonclassical,gour2018Incompatibility}. Of particular relevance to our work are the collective quantum resources described by \emph{set coherence}, as introduced by Designolle et al.~\cite{designolle2021set} (previously investigated by Refs.~\cite{horodecki2007quantification,galvao2020quantum,piani2014quantumness,fuchs2003squeezing,horodecki2006quantumness} under different names) and \emph{set imaginarity}, as introduced by Miyazaki and Matsumoto~\cite{miyazaki2022imaginarityfree}. 

Quantum coherence~\cite{baumgratz2014quantifying,streltsov2017colloquium} is the paradigmatic quantum resource---essential (though not sufficient) for computational~\cite{ahnefeld2022coherence,naseri2022entanglement} or informational~\cite{ahnefeld2025coherenceresourcephaseestimation} advantage. A related---but less explored---resource is quantum imaginarity~\cite{gour2017quantum,hickey2018quantifying,wu2021operational,wu2021resource,wu2023resource}, which is a type of coherence requiring non-null off-diagonal imaginary terms. Imaginarity has been shown to yield communication advantages~\cite{elliott2025strictadvantagecomplexquantum}, and to be connected with hiding and masking quantum information~\cite{zhu2021hidingmasking}, quantum pseudorandom states and unitaries~\cite{haug2025pseudorandom}, enhanced sensing via weak measurements~\cite{kedem2012usingtechnical,dixon2009ultrasensitive,hosten2008observation,brunner2010measuringsmall}, work extraction through complex‑valued quasiprobabilities~\cite{gherardini2024quasiprobabilities,hernandezGomez2024interferometry}, and scrambling applied to quantum machine learning~\cite{sajjan2023imaginary}.

Moving to such a collective description naturally introduces a basis-independent perspective, since one considers the resource not of a single quantum state in a fixed basis, but of the overall structure and behavior of the ensemble. As multivariate traces are known to capture all unitary-invariant properties of a collection of operators, because they form a complete list of generating invariants~\cite{wigderson2019mathematics,popescu2009unitaryinvariants,chien2016characterization,oszmaniec2024measuring} one is naturally led---when considering density matrices as the relevant operators---to view \emph{Bargmann invariants}~\cite{bargmann1964note} as relevant tools. 

Bargmann invariants are multivariate traces of states  $\mathrm{Tr}(\rho_1\cdots \rho_n)$ that can be experimentally measured in various ways~\cite{wagner2024quantumcircuits,oszmaniec2024measuring,halpern2018quasiprobability,quek2024multivariatetrace,pont2022quantifying,simonov2025estimationmultivariatetracesstates}. These have recently been connected with 
Kirkwood--Dirac (KD) quasiprobability distributions~\cite{kirkwood1933quantum,dirac1945analogy,wagner2024quantumcircuits,arvidssonshukur2024properties,schmid2024kirkwood,liu2025boundarykirkwooddiracquasiprobability}, out-of-time-order correlators~\cite{yunger2018quasiprobability,gonzalez2019out,wagner2024quantumcircuits}, weak values~\cite{wagner2023anomalous,hofmann2012complex}, quantum speed limits~\cite{sagarsilvapratapsi2025quantumspeedlimits},  geometric phases~\cite{mukunda2001Bargmann,mukunda2003Bargmann,mukunda2003Wigner,avdoshkin2023extrinsic}, multi-photon indistinguishability~\cite{menssen2017distinguishability,jones2020multiparticle,minke2021characterizing,pont2022quantifying,rodari2024experimentalobservationcounterintuitivefeatures,rodari2024semideviceindependentcharacterizationmultiphoton,seron2023boson,giordani2021witnesses,giordani2020experimental,jones2023distinguishability,jones2023distinguishability,brod2019witnessing,annoni2025incoherentbehaviorpartiallydistinguishable}, overlap uncertainty relations~\cite{bong2018strong}, quantum thermodynamics~\cite{gherardini2024quasiprobabilities,lostaglio2022kirkwood,levy2020quasiprobability,hernandez2024projective,santini2023work,donati2024energetics}, and the certification of quantum resources~\cite{fernandes2024unitary,wagner2025unitary,wagner2024coherence,wagner2024inequalities,zhang2024local,giordani2023experimental}.

In this work, we exploit Gram matrix methods and Bargmann invariant theory to characterize the simplest nontrivial instances of basis‐independent imaginarity and coherence for collections of states: the single‐qubit case. Since we are interested in ordered collections of finite states, relative to a varying parameter, these are elegantly captured by the operational notion of a \emph{multi-state}~\cite{zhang2025reassessingboundaryclassicalnonclassical,gour2018Incompatibility} (see Fig.~\ref{fig: from_states_to_bargmanns}), i.e., a function $\varrho: I \to \mathcal{D}(\mathcal{H})$ which to every index $i \in I$ associates a quantum state $\varrho(i) \equiv \rho_i$ in some Hilbert space $\mathcal{H}$. The image of $\varrho$ is given by sets of states $\{\rho_i\}_{i\in I}$ and, moreover,  each such multi-state has an associated ordered set $\varrho =(\rho_1,\ldots,\rho_n)$ when $I = \{1,\ldots,n\} =: [n]$ for some fixed number $n \in \mathbb{N}$.

Our main contributions are the following:
\begin{itemize}
\item We derive necessary and sufficient conditions for single‐qubit multi-states $\varrho: [n]\to \mathcal{D}(\mathbb{C}^2)$ to be imaginarity‐free (see Theorem~\ref{theorem: Characterization_Imaginarity}) or incoherent (see Theorem~\ref{theorem: Characterization_Coherence}), expressed solely in terms of two‐state overlaps $\mathrm{Tr}(\rho_i\rho_i)$ of the states in the multi-states.  Our results formalize the following geometric intuition: a set of states is incoherent exactly when all Bloch vectors lie on a common line, and imaginarity‐free exactly when they lie in a common plane.

\item This simple criteria has a few immediate implications. To mention one, we show that a multi-state $\varrho:\{1,2,3\} \to \mathcal{D}(\mathbb{C}^2)$ has imaginarity if and only if (iff) the imaginary part of its associated third‐order Bargmann invariant is non-zero, i.e. $\mathrm{Im}[\mathrm{Tr}(\rho_1\rho_2\rho_3)] \neq 0$. This is, however, \emph{not} the case for multi-states relative to higher-dimensional systems. We construct a counterexample of a multi-state $\varrho: \{1,2,3\} \to \mathcal{D}(\mathbb{C}^4)$ for which $\mathrm{Im}[\mathrm{Tr}(\rho_1\rho_2\rho_3)] = 0$ but that, nevertheless, has imaginarity. This follows from the results in Refs.~\cite{fernandes2024unitary,wagner2024quantumcircuits} as the imaginarity of this multi-state is witnessed by the fact that $\mathrm{Im}[\mathrm{Tr}(\rho_1^2\rho_2\rho_3)] \neq 0$.

\item We furthermore use this characterization to show that the set of all imaginarity-free multi-states is not convex. The non-convexity of the incoherent multi-states had been previously showed by Designolle et al.~\cite{designolle2021set}.

\item Beyond providing a necessary and sufficient condition for single-qubit multi-state coherence, our main technical contribution is a more powerful result: Take any single-qubit multi-state $\varrho:[n] \to \mathcal{D}(\mathbb{C}^2)$ and write its associated Bargmann invariant as $$\Delta(\varrho) := \mathrm{Tr}(\rho_1 \cdots \rho_n) = \mathrm{Re}[\Delta(\varrho)]+ \mathrm{i} \mathrm{Im}[\Delta(\varrho)].$$ Then, both $\mathrm{Re}[\Delta(\varrho)]$ and $\vert \mathrm{Im}[\Delta(\varrho)]\vert$ are completely characterized by polynomials over the complete list of two-state overlaps $(\mathrm{Tr}(\rho_i\rho_j))_{i,j}$ (see Theorem~\ref{theorem:Bargmanns_dependence_overlaps}). In other words, the only non-trivial unitary-invariant information that is not already encoded in the overlaps is the \emph{sign} of the imaginary part of the invariant. This result is a generalization to every $n \in \mathbb{N}$ of the cases $n=3,4,$ and $5$ explored in Ref.~\cite{zhang2025geometrysets}. The polynomials mentioned above are constructed inductively. 

\item We express set‐coherence and set‐imaginarity \emph{quantifiers}---restricted to the single-qubit case---as functions of two‐state overlaps and show that imaginarity quantifiers in particular can be written as a semidefinite optimization problem. We then relate these quantifiers to our rank conditions, deriving explicit upper and lower bounds (see Theorem~\ref{theorem: Quantification_imaginarity} for the case of set imaginarity and Theorem~\ref{theorem: Quantification_Coherence} for the case of set coherence).

\item We also discuss simple extensions of our tools beyond the single-qubit case. Following Refs.~\cite{fernandes2024unitary,wagner2024inequalities}, we propose simple \emph{equality constraints} on Bargmann invariants that serve as operational witnesses for multi‐state coherence and imaginarity in systems of arbitrary dimension. 

\item We illustrate how our quantifiers yield operational advantages in sub-channel discrimination and discuss implications of our results for multi‑photon indistinguishability and spin chirality. These examples demonstrate the physical relevance of the resources we investigate.
\end{itemize}

\begin{figure}[t]
    \centering
    \includegraphics[width=0.85\linewidth]{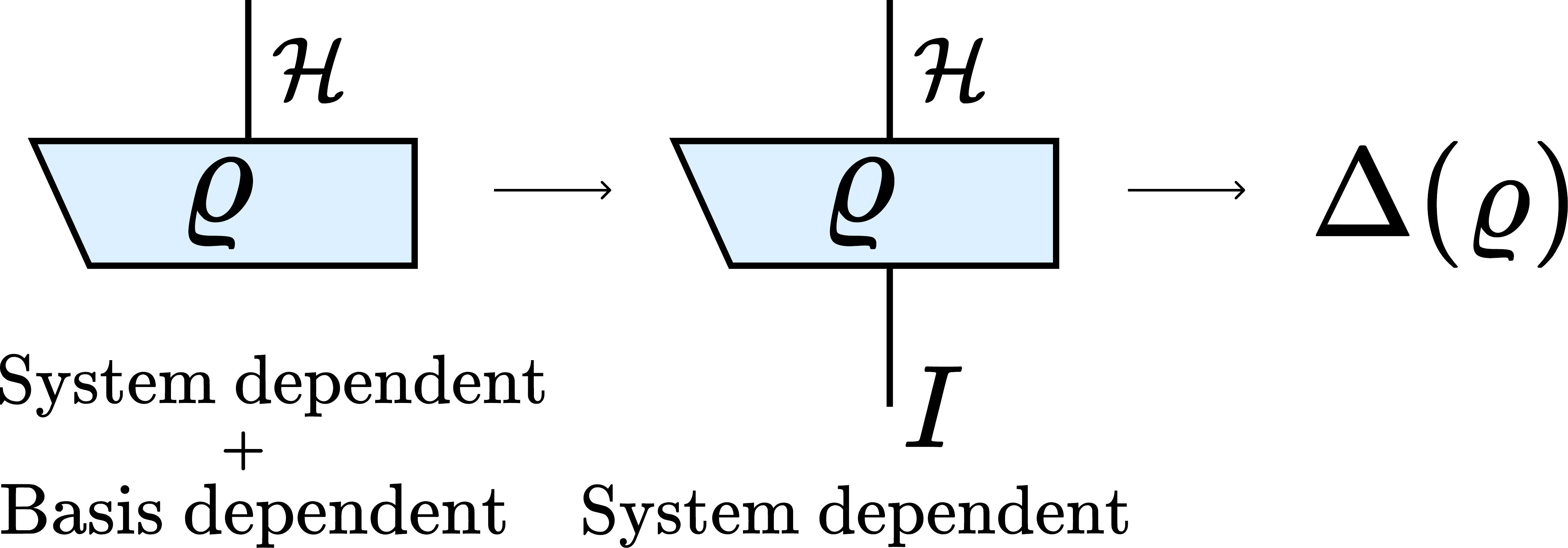}
    \caption{\textbf{States, multi-states, and Bargmann invariants.} (Left) A single quantum state $\varrho$, whose coherence is defined relative to a specific basis and system. (Middle) A multi-state $\varrho: I \to \mathcal{D}(\mathcal{H})$, which can be viewed operationally as a state preparator assigning a quantum state $\varrho(i)$ to each index $i \in I$. Its coherence is a collective, basis-independent property of the set $\{\varrho(i)\}_{i \in I}$. (Right) A Bargmann invariant $\Delta(\varrho) = \mathrm{Tr}\left[\prod_{i \in I} \varrho(i)\right]$ whose value is invariant under changes of basis. Its value can reflect the coherence present in \emph{every} multi-state that realizes it.}
    \label{fig: from_states_to_bargmanns}
\end{figure}

\textbf{Outline.} The remainder of this document is structured as follows. In Sec.~\ref{sect:Pre} we review the relevant background. Section~\ref{sec:imaginarity_results} characterizes single-qubit multi-state imaginarity, with Theorem~\ref{theorem: Characterization_Imaginarity} establishing a simple and intuitive rank-based criteria. Section~\ref{sec:singlequbit_coherence} characterizes single-qubit multi-state coherence via Theorem~\ref{theorem: Characterization_Coherence}. Section~\ref{sec:structure_Bargmann_invariants}  has our main technical contribution, where we provide a characterization of the real and imaginary part of \textit{all} Bargmann invariants of single-qubit states. Section~\ref{sec:applications} discusses the physical relevance of our results and possible applications. Section~\ref{sec:discussion} provides discussion and future directions.

\section{Background }\label{sect:Pre}

\subsection{Set coherence and set imaginarity}\label{sec:set_def}

We denote by $\mathcal{H}$ the Hilbert space associated to a quantum system of finite dimension $d = \dim(\mathcal{H})$. Moreover, we also denote by $\mathcal{D}(\mathcal{H})$ the set of all positive trace-one operators which we refer to as quantum states, and by $\mathcal{P}_1(\mathcal{H}) = \{\vert \psi \rangle \langle \psi \vert  \mid \langle \psi | \psi \rangle = 1 \}$ the extremal elements of $\mathcal{D}(\mathcal{H})$, i.e. the set of all pure quantum states. 

A \emph{multi-state} is a function $\varrho: I \to \mathcal{D}(\mathcal{H})$ from an index set $I$ to the set of density matrices of some Hilbert space~\cite{selby2023accessible,zhang2025quantifierswitnessesnonclassicalitymeasurements,zhang2025reassessingboundaryclassicalnonclassical,gour2018Incompatibility}. Operationally, we can view a multi-state as a `state preparator' or `preparation box' having  classical inputs chosen from an index set $I$ and to each $i$ it outputs a quantum state $\varrho(i) = \rho_i$. The image of a multi-state is a set of states $\varrho(I) = \{\rho_i\}_i$. If the index set is discrete (and finite) the multi-state corresponds to a (finite) sequence $\varrho = (\varrho(i))_{i \in I}$. For example, $\varrho: \{1,2,3\} \to \mathcal{D}(\mathcal{H})$ corresponds to  $\varrho = (\varrho(1),\varrho(2),\varrho(3))$. Note that since the tuple and the function are bijectively related, we use the same notation for both.

We start by introducing the notions of set coherence and set imaginarity, which are concepts defined relative to the image of a multi-state. Recall that a state $\rho \in \mathcal{H}$, with $d=\dim(\mathcal{H})<\infty$, is said to be incoherent with respect to an orthonormal basis $\mathbbm{A} = \{\ket{a_i}\}_{i=1}^d$ iff for every $i \neq j$ we have $\langle a_i \vert \rho \vert a_j \rangle =0$. Set coherence can then be defined as follows:

\begin{definition}[Set coherence]\label{def: set coherence}
   A \emph{set} of states $\{\rho_i\}_i \subseteq \mathcal{D}(\mathcal{H})$ is \emph{incoherent} if there exists a unitary $U:\mathcal{H} \to \mathcal{H}$ such that, for all $i$, $U \rho_i U^\dagger$ is a diagonal density matrix with respect to some orthonormal basis $\mathbbm{A}$. If this does not hold, we say that $\{\rho_i\}_i$ is \emph{coherent}.
\end{definition}

Alternatively, if we denote the set of all incoherent states as $$\mathcal{I}(\mathcal{H},\mathbbm{A}) := \{\rho \in \mathcal{D}(\mathcal{H}) \mid \forall i, j. \,\,i\neq j\Rightarrow  \,\langle a_i \vert \rho \vert a_j \rangle = 0\}$$ we can define set coherence as follows: a set $\{\rho_i\}_i$ is incoherent if there exists \emph{some} basis $\mathbbm{A}$ with respect to which $\{\rho_i\}_i \subseteq \mathcal{I}(\mathcal{H},\mathbbm{A})$ and coherent otherwise. 

While this definition is usually credited to Ref.~\cite{designolle2021set}, it is worth mentioning that there the authors consider set coherence as a notion defined with respect to finite \textit{tuples} (i.e., to multi-states whose index set is finite) as opposed to \textit{sets} (i.e., as opposed to the image of a multi-state). Moreover, it is obvious that a set $\{\rho_i\}_i$ is incoherent iff $[\rho_i,\rho_j] = 0$ for all $i,j$.

Similarly, we define the notion of set imaginarity lifting the definition of basis-dependent imaginarity of a single state~\cite{hickey2018quantifying,wu2021operational} to the  basis-independent definition of imaginarity for a set of states~\cite{miyazaki2022imaginarityfree}. Recall that a state $\rho \in \mathcal{D}(\mathcal{H})$, with $d=\mathrm{dim}(\mathcal{H})<\infty$, is imaginarity-free with respect to an orthonormal basis $\mathbbm{A} = \{\ket{a_i}\}_{i=1}^d$ iff for every $i,j$ we have that $\langle a_i \vert \rho \vert a_j \rangle \in \mathbb{R}$. Otherwise we say that $\rho$ has non-zero imaginarity, or that it is not imaginarity-free. With that, we can define set imaginarity as follows. 

\begin{definition}[Set imaginarity]\label{def: set imaginarity}
    A \emph{set} of states $\{\rho_i\}_i \subseteq \mathcal{D}(\mathcal{H})$ is  \emph{imaginarity-free} if there exists a unitary $U:\mathcal{H} \to \mathcal{H}$ such that, for all $i$, $U \rho_i U^\dagger$ is a real density matrix with respect to some basis $\mathbbm{A}$. If this does not hold, we then say that $\{\rho_i\}_i$ is \emph{not imaginarity-free}. 
\end{definition}

We say that a set of states that is not imaginarity-free \emph{has} imaginarity. If we denote the set of all imaginarity-free states with respect to the basis $\mathbbm{A}$ as $$\mathcal{R}(\mathcal{H},\mathbbm{A}):= \{\rho \in \mathcal{D}(\mathcal{H}) \mid \forall i,j, \langle a_i\vert \rho \vert a_j \rangle \in \mathbb{R}\}$$ we can equivalently define set imaginarity as follows: a set $\{\rho_i\}_{i\in I}$ is imaginarity-free if there exists \emph{some} basis $\mathbbm{A}$ with respect to which $\{\rho_i\}_i \subseteq \mathcal{R}(\mathcal{H},\mathbbm{A})$ and not imaginarity-free otherwise. 

In what follows, we will use the same terminology of set imaginarity and set coherence applied to the multi-states $\varrho$. We will then say that a multi-state $\varrho:I \to \mathcal{D}(\mathcal{H})$ has coherence (imaginarity) iff the image $\varrho(I)$ is set coherent (not imaginarity-free). 

The set of all incoherent multi-states $\varrho:I\to \mathcal{D}(\mathcal{H})$ will be denoted as
\begin{equation}
    \mathcal{I}_I(\mathcal{H}) = \{\varrho:I\to\mathcal{D}(\mathcal{H}) \mid \exists \mathbbm{A} \text{ s.t. }\varrho(I)\subseteq \mathcal{I}(\mathcal{H},\mathbbm{A})\}
\end{equation}
and similarly the set of all imaginarity-free multi-states 
\begin{equation}
    \mathcal{R}_I(\mathcal{H}) = \{\varrho:I\to\mathcal{D}(\mathcal{H}) \mid \exists \mathbbm{A} \text{ s.t. }\varrho(I)\subseteq \mathcal{R}(\mathcal{H},\mathbbm{A})\}.
\end{equation}
When $I = [n] \equiv \{1,\ldots,n\}$ we write $\mathcal{I}_n(\mathcal{H})$ and $\mathcal{R}_n(\mathcal{H})$ instead. 

We remark that this is a straightforward manner to lift basis-dependent definitions relative to a single state (in our case, coherence and imaginarity) to basis-independent definitions relative to multi-states. As another example, Refs.~\cite{wagner2025unitary,zamora2025prepareandmagicsemideviceindependentmagic} considered \emph{set magic}---whenever the image of a multi-state cannot be unitarily mapped to elements inside the stabilizer polytope---and characterized two-state overlap inequality witnesses of this property for generic multi-states. 

Changing from a single-state to a multi-state description allows us to straightforwardly `lift' basis-dependent notions into basis-independent ones. However, this approach still retains an inherent \emph{system dependence}---that is, a reference to the specific Hilbert space in which the states are defined. To go further, we can also `lift' this system dependence to a system (and basis) independent notion. We do so by drawing on the theory of Bargmann invariants and by framing the discussion in terms of quantum realization problems~\cite{fraser2023estimationtheoreticapproachquantum}, an approach which was also considered by Ref.~\cite{fernandes2024unitary}.

\subsection{Sets of Bargmann invariants}\label{sec:bargmann_background}

As mentioned in the introduction, Bargmann invariants play a crucial role in the characterization of basis-independent resources of multi-states, as they generate all unitary-invariant polynomial functions of tuples of states~\cite{oszmaniec2024measuring,chien2016characterization}. More formally, estimating sufficiently many Bargmann invariants provides a complete solution to a unitary equivalence decision problem of the following form: 

\begin{quote}
Given two multi-states $\varrho, \varsigma: [n] \to \mathcal{D}(\mathcal{H})$ is there a unitary $U:\mathcal{H} \to \mathcal{H}$ such that
$U \varrho U^\dagger = \varsigma$?
\end{quote}

Above, we have used the notation $U \varrho U^\dagger = (U\rho_1 U^\dagger,\ldots,U\rho_n U^\dagger)$. In finite dimensions ($d=\dim\mathcal H<\infty$), the existence of a unitary $U$ with $U\varrho U^\dagger=\varsigma$ is equivalent to the requirement that all Bargmann invariants agree.  Concretely, such a unitary $U$ exists iff for every $1 \leq m \leq d^2$ and every sequence $i_1, \ldots, i_m$ from $\{1,\ldots, n\}$ the corresponding Bargmann invariants \emph{coincide}:
\begin{equation}
    \mathrm{Tr}(\rho_{i_1}\ldots \rho_{i_m}) = \mathrm{Tr}(\sigma_{i_1} \ldots \sigma_{i_m}).
\end{equation}

It turns out, based on the findings of Refs.~\cite{oszmaniec2024measuring,wigderson2019mathematics}, that considering all the exponentially many Bargmann invariants mentioned above is often \emph{unnecessary}. Identifying minimal (or at least reduced) sets of invariants that suffice to solve these decision problems remains an active area of research. In fact, one of our results shows that in the case of single-qubit multi-state imaginarity, estimating a \emph{single} Bargmann invariant is sufficient (see Corollary~\ref{corollary:single_invariant_sufficient}). Moreover, Ref.~\cite{oszmaniec2024measuring} presents an algorithm for identifying a minimal set of necessary and sufficient Bargmann invariants, under the assumption that the multi-state consists exclusively of pure states.

Besides completely characterizing the unitary-invariant properties of multi-states of fixed Hilbert space dimension, Bargmann invariants can also be used to relax the requirement of system dependence, as we now discuss.

Let us assume that the index set of a multi-state is a finite set, so that $I \simeq \{1,\dots,n\} \equiv [n]$. In this case, we describe the set of all Bargmann invariants as an instance of a specific \emph{quantum realizability problem}~\cite{fraser2023estimationtheoreticapproachquantum}: 
\begin{quote}
Given a value $\Delta \in \mathbb{C}$, is there some multi-state $\varrho:[n] \to \mathcal{D}(\mathbb{C}^d)$ such that $\Delta = \mathrm{Tr}(\rho_1 \cdots \rho_n) \equiv \Delta(\varrho)$? 
\end{quote}
If the answer is yes, we say that $\Delta \in \mathfrak{B}_{n}^{(d)}$.  Succinctly, we define this set as 
\[
\mathfrak B_{n}^{(d)}  := \{\Delta \in \mathbb{C} \mid \exists \varrho \in \mathcal{D}(\mathbb{C}^d)^n \text{~s.~t.~}\Delta = \Delta({\varrho})\}.
\]

If for a certain $\Delta$ there exists $\varrho$ such that $\Delta = \Delta(\varrho)$ we say that $\Delta$ is \emph{realized by} $\varrho$ in the Hilbert space $\mathcal{H} = \mathbb{C}^d$. More generally, we say that a complex-number $\Delta \in \mathbb{C}$ is quantum realizable if there exists some multi-state $\varrho = (\rho_1,\dots,\rho_n)$ with respect to some Hilbert space $\mathcal{H}$ such that $\Delta = \Delta(\varrho)$. Refs.~\cite{fernandes2024unitary,wagner2024quantumcircuits} have pointed out the simple (yet useful) remark that $\mathrm{Im}[\mathrm{Tr}(\rho_1 \cdots \rho_n)] \neq 0$ implies that $\{\rho_i\}_i$ has imaginarity as by Def.~\ref{def: set imaginarity}.

An important subset of $\mathfrak B_{n}^{(d)}$ is the set of all values $\Delta$ realizable by multi-states of \emph{pure states}, that we denote as $\Psi: [n] \to \mathcal{P}_1(\mathbb{C}^d)$. Therefore, we define
\[
\mathfrak B_{n}^{(d)}\vert_{\mathrm{pure}}  := \{\Delta \in \mathbb{C} \mid \exists \Psi \in \mathcal{P}_1(\mathbb{C}^d)^n \text{ s. t. }\Delta = \Delta({\Psi})\}.
\]

These definitions lead to the following ascending chain of sets:
\[
\mathfrak{B}_{n}^{(2)} \subseteq \mathfrak{B}_{n}^{(3)} \subseteq \cdots \subseteq \mathfrak{B}_{n}^{(d)} \subseteq \cdots \subseteq \mathfrak{B}_{n}^{(n)} = \mathfrak{B}_n 
\]
where we define 
\[
\mathfrak{B}_n = \bigcup_{d=2}^{\infty} \mathfrak{B}_{n}^{(d)}.
\]

Ultimately, the set $\mathfrak{B}_n$ for a given fixed $n$ is the subset of complex values that Bargmann invariants can reach. On a series of recent results these sets have been completely characterized~\cite{fernandes2024unitary,xu2025numericalrangesbargmanninvariants,mao2025bargmann,zhang2025geometrysets,pratapsi2025elementarycharacterizationbargmanninvariants} and it has been showed that the above ascending chain of sets \emph{collapses}, i.e., that $$\mathfrak{B}_n = \mathfrak{B}_n^{(2)} = \mathfrak{B}_n\vert_{\mathrm{pure}} $$ for all $n \in \mathbb{N}$. Moreover, there is a \emph{structure theorem} for Bargmann invariants in $\mathfrak{B}_n$~\cite{pratapsi2025elementarycharacterizationbargmanninvariants}, which states that 
\begin{equation}
    \mathfrak{B}_n = \mathfrak{B}_n^{(3)}\vert_{\mathrm{circ}},
\end{equation}
namely, that to every Bargmann invariant $\Delta(\varrho) = \mathrm{Tr}(\rho_1 \cdots \rho_n)\in \mathfrak{B}_n$ there exists a \emph{pure} multi-state $\Psi = (\psi_1,\ldots,\psi_n) \in \mathcal{D}(\mathbb{C}^3)$ such that the associated Gram matrix $(G_{(\vert \psi_1\rangle,\ldots,\vert \psi_n\rangle)})_{ij} = \langle \psi_i\vert \psi_j\rangle $ is a \emph{circulant} matrix. 

\subsection{Sets of tuples of two-state overlaps}

There are other sets of Bargmann invariants that have been investigated in the literature. For example, Refs.~\cite{wagner2024inequalities,wagner2024quantumcircuits,wagner2025unitary,galvao2020quantum,giordani2020experimental,giordani2021witnesses,giordani2023experimental,wagner2024coherence} have also considered more structured sets of tuples of second-order Bargmann invariants. In fact, in this case, Refs.~\cite{wagner2024inequalities,galvao2020quantum} have introduced a graph-theoretic framework---connected to Kochen--Specker contextuality~\cite{budroni2021kochenspeckerreview}---where we can link the index sets of multi-states to vertices of a graph, and describe the quantum realizability problem relative to this graph. 

The formalism from Ref.~\cite{wagner2024inequalities} is built from the following simple idea. To organize which overlaps matter in a given experimentally relevant situation, it is convenient to label the states by a set and select only the pairs whose overlaps enter the constraints. These labels and pairs form the vertices and edges of a simple graph, where each edge carries the corresponding label associated to an overlap value.  We now make this precise.

Let $V = [n]$ and $E \subseteq \{\{i,j\} \mid i,j \in V\}$. We call a function $\mathsf{\Delta}:E \to [0,1]$ an \emph{edge-weighting} since, formally, the elements $e \in E$ can be viewed as edges in the simple graph defined by $G=(V,E)$. A multi-state $\varrho: V \to \mathcal{D}(\mathcal{H})$ is then a function returning a quantum state in $\mathcal{H}$ for each vertex $V$ of the graph $G$. This graph-theoretic view will not be necessary in what follows, so we can simply view $V$ as a specification of an index set and $E$ as the specification of which two-state overlaps are relevant.

We say that $\mathsf{\Delta}:E\to [0,1]$ is \emph{quantum realizable} in a Hilbert space $\mathcal{H}$ if there exists a multi-state $\varrho: V\to \mathcal{D}(\mathcal{H})$ such that $$\mathsf\Delta(\{i,j\}) = \Delta_{ij} = \mathrm{Tr}(\varrho(i)\varrho(j))$$ for every $\{i,j\} \in E$. Note that, from this perspective, $\Delta_{ij}$ are merely scalars with no reference (\textit{a priori}) to quantum states, while $\mathrm{Tr}(\varrho(i)\varrho(j))$ is the overlap that reaches the value of that scalar for $\{i,j\}$. Whenever there exists such a multi-state, we write $\mathsf{\Delta}=\mathsf{\Delta}(\varrho)$. In this case, we have then a set of quantum correlations defined by such a quantum realization problem:

\begin{definition}[Adapted from Ref.~\cite{wagner2024inequalities}]
    Let $G=(V,E)$ be a simple graph. The set of quantum-realizable edge weightings is defined as
    \begin{equation}
        \mathfrak{Q}(G) := \{\mathsf{\Delta} \in [0,1]^{E} \mid \exists \varrho \mathrm{~s.~t.~} \mathsf{\Delta} = \mathsf{\Delta}(\varrho)\}
    \end{equation}
    where $\varrho: V \to \mathcal{D}(\mathcal{H})$ for some Hilbert space $\mathcal{H}$.
\end{definition}

In simple terms, $\mathfrak{Q}(G)$ is the set of all possible tuples of the form $(\mathrm{Tr}(\rho_i\rho_j))_{\{i,j\} \in E}$. If we let $E = \{\{1,2\},\{1,3\},\{2,3\}\}$---which is the edge-set of a 3-cycle graph---$\mathfrak{Q}(G)$ is equivalent to the set of all possible tuples 
$$\mathsf{\Delta}(\varrho) = (\mathrm{Tr}(\rho_1\rho_2),\mathrm{Tr}(\rho_1\rho_3),\mathrm{Tr}(\rho_2\rho_3)),$$
for all possible multi-states $\varrho:\{1,2,3\} \to \mathcal{D}(\mathcal{H})$ with respect to all possible Hilbert spaces $\mathcal{H}$.

Similarly to the case of the sets of Bargmann invariants $\mathfrak{B}_n$ we can also consider restrictions of these sets. The most relevant to us are the restrictions to pure states

\begin{equation}
    \mathfrak{Q}(G)|_{\mathrm{pure}} = \left\{\mathsf{\Delta} \mid \exists \Psi:V \to \mathcal{P}_1(\mathcal{H}) \text{~s.~t.~} \mathsf{\Delta} = \mathsf{\Delta}(\Psi)\right\},
\end{equation}
and the restriction to imaginarity-free states
\begin{equation}\label{eq: edge-weightings realizable with set real states}
\mathfrak{Q}(G)|_{\mathrm{real}} := \{\mathsf{\Delta}  \mid \exists \varrho: V \to \mathcal{R}(\mathcal{H},\mathbbm{A}) \text{~s.~t.~}\mathsf{\Delta} = \mathsf{\Delta}(\varrho)\}.
\end{equation}

Provided we have the function $\mathsf{\Delta}: E \to [0,1]$ and its notion of quantum realization, we can apply it to notions of quantum realizations of any function $f(\mathsf{\Delta})$. For example, if we let $$P(x_1,x_2,x_3) = x_1^2+x_2^2+x_3^2$$ we can talk about all values $P$ can take provided that $x_1 = \Delta_{12}(\varrho) \equiv \mathrm{Tr}(\rho_1\rho_2)$, $x_2 =\Delta_{13}(\varrho) \equiv \mathrm{Tr}(\rho_1\rho_3)$, $x_3 = \Delta_{23}(\varrho) \equiv \mathrm{Tr}(\rho_2\rho_3)$ in which case we write 
$$P(\mathsf{\Delta}(\varrho)) = (\mathrm{Tr}(\rho_1\rho_2))^2+(\mathrm{Tr}(\rho_1\rho_3))^2+(\mathrm{Tr}(\rho_2\rho_3))^2.$$

This notion will be relevant to us later on in Sec.~\ref{sec:singlequbit_coherence}, when stating and interpreting Theorem~\ref{theorem:Bargmanns_dependence_overlaps}.

\subsection{Gram matrix of Bloch vectors}\label{sec:bloch_rep}

Let $\rho \in \mathcal{D}(\mathbb{C}^2)$ be any single-qubit state, there exists a unique vector $\mathbf{r}=(x,y,z)\in\mathbb{R}^3$ such that $|| \mathbf{r}||^2=x^2+y^2+z^2\leq 1,$  and 
$$ \rho= \frac{\mathbb{1}+ \mathbf{r} \cdot \boldsymbol{\sigma}}{2}.$$
The vector $\mathbf{r}$ is known as the \emph{Bloch vector} of the state $\rho.$
Given an $n$-tuple of single qubit states $\varrho=(\rho_1,\rho_2,\ldots,\rho_n)\in \mathcal{D}(\mathbb{C}^2)^n,$ we can construct the $n \times n$ Gram matrix associated to the related tuple of Bloch vectors $\mathbf{r}_\varrho := (\mathbf{r}_1,\ldots,\mathbf{r}_n)$ given by 
$$G_{\mathbf{r}_\varrho}:=(\langle \mathbf{r}_k, \mathbf{r}_l\rangle)_{kl}.$$ 
The overlap between two single-qubit states $\rho_i,\rho_j \in \varrho$ is a function of the scalar product between the two Bloch vectors
\begin{equation}\label{eq:inner_Bargmann}
    \mathrm{Tr}[\rho_i\rho_j]=\mathrm{Tr}\left[\frac{(\mathbb{1}+ \mathbf{r}_i \cdot \boldsymbol{\sigma})}{2} \frac{(\mathbb{1}+ \mathbf{r}_j \cdot \boldsymbol{\sigma})}{2}\right]=\frac{1+\langle \mathbf{r}_i, \mathbf{r}_j\rangle   }{2}.
\end{equation} 
So 
\begin{equation}\label{eq:inner_product_equal_overlaps}
\langle \mathbf{r}_i, \mathbf{r}_j\rangle =2 \mathrm{Tr}[\rho_i\rho_j]-1,
\end{equation}
which implies that each entry of the Gram matrix $G_{\mathbf{r}_\varrho}$ is in one-to-one correspondence with some Bargmann invariant of order 2. 

We can make this construction more general~\cite{kimura2003bloch,bertlmann2008blochvectors}, given an orthogonal (with respect to the Hilbert-Schmidt inner product) basis of $d\times d$ matrices, e.g. the set of (generalized) Gell--Mann matrices together with the identity, where $U_0= \mathbb{1}_d $, $U_i$ are traceless and Hermitian matrices satisfying $\mathrm{Tr}[U_i U_j]=d\delta_{i,j}.$

Each density matrix $\rho \in \mathcal{D}(\mathbb{C}^d)$ can be written as 
$$ \rho= \frac{1}{d}\left(\mathbb{1}_d  +\sum_{j=1}^{d^2-1} r_j U_j\right) $$
where $r_j=\tr[\rho U_j]$. If we are given a multi-state $\varrho=(\rho_1,\cdots, \rho_n) \in \mathcal{D}(\mathbb{C}^d)$, each $\varrho(i) = \rho_i$ corresponds to an $d^2-1$  dimensional real vector $\mathbf{r}_i$. Given the related tuple of associated generalized Bloch vectors $\mathbf{r}_\varrho$ we define the $n \times n$ Gram matrix 

$$G_{\mathbf{r}_\varrho}:=(\langle \mathbf{r}_k, \mathbf{r}_l\rangle)_{k,l} $$
as before, where the vectors $\mathbf{r}_i$ lie in a $d^2-1$ dimensional real vector space. Similarly to Eq.~\eqref{eq:inner_product_equal_overlaps}, we find that
\begin{equation}\label{eq:inner_product_gen_bloch_vectors_overlap}
    \langle \mathbf{r}_i,\mathbf{r}_j\rangle = d\,\mathrm{Tr}(\rho_i\rho_i)-1.
\end{equation}

\section{Characterization of single-qubit multi-state imaginarity}\label{sec:imaginarity_results}

\subsection{Rank-based criteria}

We start by considering the quantum imaginarity of a finite set of single-qubit quantum states. Ref.~\cite{fernandes2024unitary} has shown that assuming \emph{purity} it is possible to witness imaginarity using only two-state overlaps. Here we show that we can trade the assumption of purity by that of Hilbert space dimension.
	 
Given an $n$-tuple of qubit states $\varrho=(\rho_1,\rho_2,\ldots,\rho_n)\in \mathcal{D}(\mathbb{C}^2)^n,$ let $\mathbf{r}_k=(x_k,y_k,z_k)\in \mathbb{R}^3$ be the Bloch vector of $\rho_k \in \varrho([n])$ for any $k \in [n]$. Our goal is to show that the image of $\varrho$ is imaginarity-free \emph{iff} the Gram matrix of Bloch vectors 
$$(G_{\mathbf{r}_\varrho})_{kl}:=\langle \mathbf{r}_k, \mathbf{r}_l\rangle$$
is at most \textit{rank two}.

From the point of view of the Bloch sphere representation, our goal is motivated by a simple geometric intuition illustrated in Fig.~\ref{fig:Theorem_1_Bloch_sphere}. Any set of states lying entirely within the $X$–$Z$ plane of the Bloch sphere is clearly real-represented with respect to the canonical basis $\mathbbm{A}_{\mathrm{st}} =\{\vert 0\rangle,\vert 1\rangle\}$. Due to unitary invariance, any set of states lying within a plane formed by a great circle of the Bloch sphere should be imaginarity-free, since one can unitarily rotate that plane---along with all the states it contains---into the $X$–$Z$ plane. This geometric intuition is so compelling that one is naturally led to conjecture that this is \emph{the only way} a set of single-qubit states can be imaginarity-free. In what follows, we rigorously prove that this intuition is indeed correct.

\begin{figure}[t]
    \centering
    \includegraphics[width=\columnwidth]{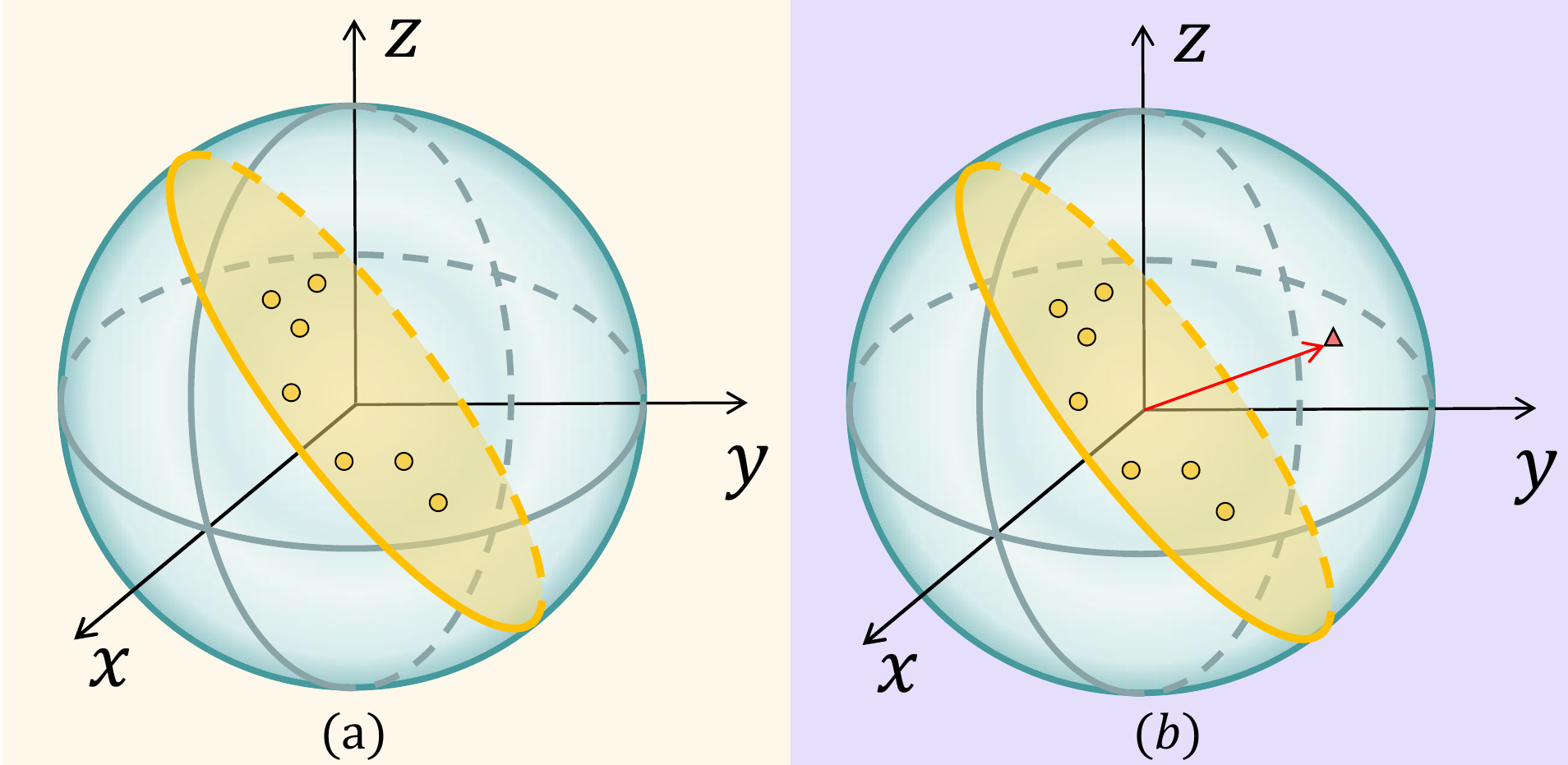}
    \caption{\textbf{Geometric interpretation of single-qubit multi-state imaginarity.} (a) Any set of states lying within a plane defined by a great circle of the Bloch sphere is imaginarity-free. (b) If, for every possible such choice of plane, there exists at least one state (represented by a red triangle) that never lies within it, the set of states has imaginarity. }
    \label{fig:Theorem_1_Bloch_sphere}
\end{figure}

We start by showing the simple direction:
\begin{lemma}\label{lemma: real implies rank less than two}
    If $\varrho:[n] \to \mathcal{\mathcal{D}}(\mathbb{C}^2)^n$ is an imaginarity-free single-qubit multi-state, i.e. $\varrho \in \mathcal{R}_n(\mathbb{C}^2)$, then $\mathrm{rank}(G_\varrho)\leq 2$.
\end{lemma}

\begin{proof}
For each $U\in \mathrm{SU}(2)$, we define the action of $U$ on the multi-state $\varrho$ as 
$$U\varrho U^\dagger:=(U\rho_1U^\dagger,U\rho_2U^\dagger,\ldots, U\rho_nU^\dagger ).$$  The Gram matrix of Bloch vectors is invariant under this map, i.e., for all $U \in \mathrm{SU}(2)$ we have that 
\begin{equation}
    G_{\mathbf{r}_\varrho} = G_{\mathbf{r}_{U \varrho U^\dagger}}.
\end{equation}
This follows easily from the fact that, to every fixed choice of unitary $U$, the map $\rho\mapsto U\rho U^\dagger $ induces an orthogonal (i.e., inner-product preserving) map $\Phi_U:\mathbb{R}^3 \to \mathbb{R}^3$ given by $\mathbf{r} \mapsto  \mathbf{r}_U$ where $ \rho= \frac{\mathbb{1}+ \mathbf{r} \cdot \boldsymbol{\sigma}}{2}$ and  $ U\rho U^\dagger= \frac{\mathbb{1}+ \mathbf{r}_U \cdot \boldsymbol{\sigma}}{2}.$ To see this,  let $ \tau= \frac{\mathbb{1}+ \mathbf{s} \cdot \boldsymbol{\sigma}}{2}$, we have Eq. \eqref{eq:inner_Bargmann}  
and 
\begin{align*}
\mathrm{Tr}[U\rho U^\dagger U\tau U^\dagger]&=\mathrm{Tr}\left[\frac{(\mathbb{1}+ \mathbf{r}_U \cdot \boldsymbol{\sigma})}{2} \frac{(\mathbb{1}+ \mathbf{s}_U \cdot \boldsymbol{\sigma})}{2}\right]\\
&=\frac{1+\langle \mathbf{r}_U, \mathbf{s}_U\rangle   }{2}.
\end{align*}
 As $\mathrm{Tr}[U\rho U^\dagger U\tau U^\dagger]=\mathrm{Tr}[U\rho\tau U^\dagger]=\mathrm{Tr}[\rho\tau]$,  we conclude 
 $$ \frac{1+\langle \mathbf{r}, \mathbf{s}\rangle   }{2}= \mathrm{Tr}[\rho\tau]= \mathrm{Tr}[U\rho U^\dagger U\tau U^\dagger]=\frac{1+\langle \mathbf{r}_U, \mathbf{s}_U\rangle   }{2}.$$
 Therefore, 
$$\langle \mathbf{r}, \mathbf{s}\rangle =\langle \mathbf{r}_U, \mathbf{s}_U\rangle =\langle \Phi_U(\mathbf{r}),\Phi_U(\mathbf{s})  \rangle.$$

We conclude the argument acknowledging that if $\varrho=(\rho_1,\rho_2,\ldots,\rho_n)$ is imaginarity free, there exists some basis $\mathbbm{A}$ with respect to which all the density matrices are real-only. This implies that there exists a unitary $U\in \mathrm{U}(2)$---describing the change of basis from $\mathbbm{A}$ to the standard basis $\mathbbm{A}_{\mathrm{st}} = \{\vert 0\rangle, \vert 1\rangle\}$---such that each element in $U\varrho U^\dagger$ lie within the great-circle plane of the Bloch sphere given by the $X-Z$ plane, from which we have that $\mathrm{rank}(G_{\mathbf{r}_{U\varrho U^\dagger}})\leq 2$. From unitary-invariance of this Gram matrix $G_{\mathbf{r}_\varrho}=G_{\mathbf{r}_{U\varrho U^\dagger}},$ we have 
$$\mathrm{rank}(G_{\mathbf{r}_\varrho})\leq 2.$$ This concludes the proof.
\end{proof}

Noting that for all $\varrho\in \mathcal{D}(\mathbb{C}^2)^n$ implies that $\mathrm{rank}(G_{\mathbf{r}_\varrho})\leq 3$ we obtain a \emph{sufficient} condition for $\varrho\in \mathcal{D}(\mathbb{C}^2)^n$ to have quantum imaginarity: $\mathrm{rank}(G_{\mathbf{r}_\varrho})= 3.$ In the following, we will show that this rank condition is also a necessary condition.    

First, recall that elements $U\in \mathrm{SU}(2)$ (i.e., $U$ is a $2\times 2$ unitary  with $\det[U]=1$), can be written as 
$$U=\left[\begin{array}{cc}
	\alpha & \beta\\[2mm]
	-\overline{\beta} & \overline{\alpha}
	\end{array}\right]
$$
where $\alpha=a+\mathrm{i} b, \beta =c+\mathrm{i} d \in \mathbb{C}$ and $a^2+b^2+c^2+d^2=|\alpha|^2+|\beta|^2=1.$ Therefore, this identifies elements of $\mathrm{SU}(2)$ with the 3-sphere $\mathbb{S}^3 = \{\mathbf{r}\in \mathbb{R}^4 \mid \Vert \mathbf{r} \Vert = 1\} \subset \mathbb{R}^4$.  With respect to this representation, the matrix form of $\Phi_U$ is given by
\begin{equation}\label{eq:matrix_PhiU}
{\footnotesize\left[
\begin{array}{cccc}
	a^2-b^2-c^2+d^2& 2(ab+cd)& 2(bd-ac)\\[2mm]
	2(cd-ab)& a^2-b^2+c^2-d^2& 2(bc+ad)\\[2mm]
	2(bd+ac)& 2(bc-ad)& a^2+b^2-c^2-d^2
\end{array}\right]}.
\end{equation}
Moreover, $\det[\Phi_U]=(a^2+b^2+c^2+d^2)^3=1.$ Therefore, $\Phi_U$ is in fact in $\mathrm{SO}(3).$ This is a well-known matrix construction, and $\Phi$ as defined describes the (surjective Lie group) homomorphism between $\mathrm{SU}(2)$ and $\mathrm{SO}(3)$. This result is known as the \emph{double cover lemma} since the mapping $U \stackrel{\Phi}{\mapsto} \Phi_U$ maps both $U(\alpha,\beta)$ and $U(-\alpha,-\beta)$ to the same element in $\mathrm{SO}(3)$, hence the `double'. 
  
\begin{lemma}[Double cover lemma, see   the Thm 2.6 of Ref.  \cite{Kosmann-Schwarzbach:2022}]\label{lemma:SU2toSO3}
   The map $\Phi:\mathrm{SU}(2) \to \mathrm{SO}(3)$ defined by $U \mapsto \Phi_U$, with $\Phi_U$ given by Eq.~\eqref{eq:matrix_PhiU} is a two-to-one surjective (Lie group) homomorphism. 
\end{lemma}

With this technical ingredient in hand, we are ready to show the following:

\begin{theorem}\label{theorem: Characterization_Imaginarity}
    A single-qubit multi-state $\varrho:[n] \to \mathcal{D}(\mathbb{C}^2)$ has  imaginarity, i.e. $\varrho \notin \mathcal{R}_n(\mathbb{C}^2)$, iff $\mathrm{rank}(G_{\mathbf{r}_\varrho})= 3$.
\end{theorem}

\begin{proof}
From Lemma~\ref{lemma: real implies rank less than two} we have that $\varrho \notin \mathcal{R}_n(\mathbb{C}^2) \implies \mathrm{rank}(G_{\mathbf{r}_\varrho})= 3$. Note that, for every single-qubit multi-state $\varrho:[n] \to \mathcal{D}(\mathbb{C}^2)$ we have that $\mathrm{rank}(G_{\mathbf{r}_\varrho})\leq 3$. Therefore, we proceed to show that $\mathrm{rank}(G_{\mathbf{r}_\varrho}) \leq 2$ implies $\varrho \in \mathcal{R}_n(\mathbb{C}^2)$.

If $\mathrm{rank}(G_{\mathbf{r}_\varrho})=1,$ one has $\mathrm{span}_\mathbb{R}\{\mathbf{r}_k\}_{k=1}^n= 1$. Assume, without loss of generality, that $\mathbf{r}_1\neq \mathbf{0}$ and set a reference Bloch vector to be the normalized vector $\mathbf{r} := \frac{\mathbf{r}_1}{||\mathbf{r}_1||}$. Since the rank is 1,  for each $k$ there exists some real number $x_k\in \mathbb{R}$ such  that  $\mathbf{r}_k= x_k\mathbf{r}.$ Clearly, there exist some orthogonal rotation $R\in \mathrm{SO}(3)$ setting $\mathbf{r}$ to $\mathbf{e}_1=(1,0, 0)^T$, i.e., $R \mathbf{r} =\mathbf{e}_1$. As $\Phi$ is surjective, we have some $U\in \mathrm{SU}(2)$ such that $\Phi_U=R$. Therefore, $ \Phi_U \mathbf{r}_k= x_k \Phi_U \mathbf{r}=x_k R \mathbf{r}=x_k \mathbf{e}_1=(x_k,0,0)^T.$ Hence  $$  U\rho_kU^\dagger= \frac{\mathbb{1}+( \Phi_U \mathbf{r}_k )\cdot \boldsymbol{\sigma}}{2}=  \frac{\mathbb{1}+ r_k \sigma_1}{2}$$
is a real density matrix with respect to $\mathbbm{A}_{\mathrm{st}} = \{\vert 0 \rangle, \vert 1 \rangle\}$ for all $k=1,2,\ldots,n$. 

If  $\mathrm{rank}(G_{\mathbf{r}_\varrho})=2 ,$  one has $\mathrm{span}_\mathbb{R}\{\mathbf{r}_k\}_{k=1}^n= 2.$ There exist two orthogonal unit vectors $\mathbf{s},\mathbf{t}\in \mathbb{R}^3$ such that 
$$\mathrm{span}_\mathbb{R}\{\mathbf{r}_k\}_{k=1}^n= \mathrm{span}_\mathbb{R}\{ \mathbf{s}, \mathbf{t}\}.$$
Therefore, for each $\mathbf{r}_k$ there exists $x_k,z_k\in \mathbb{R}$ such that 
$$ \mathbf{r}_k=x_k \mathbf{s}+ z_k\mathbf{t}.$$ 
There are two unit vectors $ \pm \frac{\mathbf{s} \times \mathbf{t}}{||\mathbf{s} \times \mathbf{t}||}$ orthogonal to both $\mathbf{s}$ and $\mathbf{t}$. We now choose $\mathbf{u} \in \{ \frac{\mathbf{s} \times \mathbf{t}}{||\mathbf{s} \times \mathbf{t}||},- \frac{\mathbf{s} \times \mathbf{t}}{||\mathbf{s} \times \mathbf{t}||}\}$ such that $
	\{\mathbf{s}, \mathbf{u},\mathbf{t}\}$ forms a right-hand frame for $\mathbb{R}^3$, i.e. a spanning set of vectors such that $\mathbf{s} \times \mathbf{u} = \mathbf{t}$.  In this case, we can construct a matrix $R$ which rotates our frame to the standard orthonormal basis  $$R:=\mathbf{e}_1\mathbf{s}^T+\mathbf{e}_2\mathbf{u}^T+\mathbf{e}_3\mathbf{t}^T.$$ From the fact that $\{\mathbf{s}, \mathbf{u},\mathbf{t}\}$ forms a right-hand frame we have that $\det(R)=+1$ which, by construction, implies that $R \in \mathrm{SO}(3)$.

    By construction, $R \mathbf{s}=\mathbf{e}_1$ and $R \mathbf{t}=\mathbf{e}_3.$ So for each $k$, 
	$$R \mathbf{r}_k=R(x_k \mathbf{s}+ z_k\mathbf{t})=x_k(R\mathbf{s})+z_k(R\mathbf{t})=(x_k,0,z_k)^T.$$ As $\Phi$ is surjective, we have some $U\in \mathrm{SU}(2)$ such that $\Phi_U=R$. For this $U $ and each $k\in [n]$: 
	$$  U\rho_k U^\dagger=   \frac{\mathbb{1}+( R \mathbf{r}_k )\cdot \boldsymbol{\sigma}}{2} =\frac{\mathbb{1}+x_k\sigma_1+z_k \sigma_3}{2} $$ which is a real density matrix relative to $\mathbbm{A}_{\mathrm{st}}$ since $\sigma_1=\sigma_x$, $\sigma_3=\sigma_z$  are also.  Therefore, the image of $\varrho$ is imaginarity-free. 
\end{proof}

\subsection{A simple remark beyond single-qubit states}

Using the generalized Bloch vector representation of $d$-dimensional states, it is simple to generalize Lemma~\ref{lemma: real implies rank less than two}.

\begin{proposition} \label{proposition:PartialCharacterization}
		If $\varrho: [n] \to \mathcal{D}(\mathbb{C}^d)$ is an imaginarity-free multi-state, i.e. $\varrho \in \mathcal{R}_n(\mathbb{C}^d)$, then   $\mathrm{rank}(G_{\mathbf{r}_\varrho})\leq \frac{d(d+1)}{2}-1$, where $\mathbf{r}_\varrho$ are generalized Bloch vectors given by the generalized Gell--Mann matrices.
\end{proposition}

\begin{proof}
    If $\varrho$ is imaginarity-free there exists $U$ such that $U\varrho U^\dagger$ are all real matrices relative to some basis $\mathbbm{A}$ for $\mathcal{H}$. The subspace of real symmetric $d \times d$ matrices has dimension $d(d+1)/2$. Subtracting the identity, we are left with $d(d+1)/2-1$ components of  generalized Bloch vectors spanning at most $d(d+1)/2-1$ dimensions. Therefore, $\mathrm{rank}(G_{ \mathbf{r}_{U \varrho U^\dagger}}) = \mathrm{rank}(G_{\mathbf{r}_\varrho}) \leq d(d+1)/2-1.$   
\end{proof}

Note that whenever $d=2$ we recover Lemma~\ref{lemma: real implies rank less than two} since in this case $\mathrm{rank}(G_{\mathbf{r}_\varrho})\leq\frac{2(2+1)}{2}-1 = 2$ and the generalized Bloch vectors reduce to the usual ones.  

We can, furthermore, comment on what prevents us to obtaining a result similar to Theorem~\ref{theorem: Characterization_Imaginarity}. As before, there exists a mapping $\Phi_d$ from $\mathrm{SU}(d)$ to $\mathrm{SO}(d^2-1).$ However, the dimension gap  
\begin{align*}
\mathrm{dim}_\mathbb{R}[\mathrm{SU}(d)]&=d^2-1<\frac{(d^2-1)(d^2-2)}{2}\\&=\mathrm{dim}_\mathbb{R}[\mathrm{SO}(d^2-1)]
\end{align*}
for all $d\geq3$ makes the map $\Phi_d$ fail to be surjective. This prevents this to be a sufficient condition as well. 

We can also provide an explicit counter-example. Let us consider $d=3$ and $\{\lambda_i\}_{i=1}^8$ to be the Gell--Mann matrices~\cite[Tab. I, pg. 8]{gellmann1962symmetries}. Note that, in this case, $\mathrm{Tr}(\lambda_i\lambda_j) = 2\delta_{ij}$ so that $\langle \mathbf{r}_i,\mathbf{r}_j\rangle = 2\mathrm{Tr}(\rho_i\rho_j)-1$. Take three qutrits given by 
\begin{align}
    \rho_1 =\frac{1}{3}(\mathbb{1}_{3}+\lambda_1),\, 
    \rho_2 =\frac{1}{3}(\mathbb{1}_{3}+\lambda_4),\,
    \rho_3 =\frac{1}{3}(\mathbb{1}_{3}+\lambda_7).
\end{align}
Above, $\lambda_1,\lambda_4$ are two symmetric matrices while $\lambda_7$ is an antisymmetric matrix. A simple calculation shows that $\mathrm{Tr}(\rho_i^2)=\sfrac{5}{9}$ and $\mathrm{Tr}(\rho_i\rho_j)=\sfrac{1}{3}$ for $i\neq j$. Therefore, $\mathrm{rank}(G_{\mathbf{r}_\varrho}) = 3 < 3(3+1)/2-1 = 5$. Nevertheless, the associated multi-state $\varrho = (\rho_1,\rho_2,\rho_3)$ has imaginarity since
\begin{equation}
    \Delta(\varrho) = \mathrm{Tr}(\rho_1\rho_2\rho_3) = \frac{1}{27}(3+\mathrm{Tr}(\lambda_1\lambda_4\lambda_7))=\frac{1}{27}(3+i),
\end{equation}
and therefore $\mathrm{Im}[\Delta(\varrho)] \neq 0.$

Moreover, it is trivial to see that every multi-state $\varrho:\{1,2\} \to \mathcal{D}(\mathcal{H})$ is imaginarity-free, for all possible Hilbert spaces $\mathcal{H}$. We state this as a lemma for future reference:

\begin{lemma}\label{lemma:every_multi_state_imaginarity_free}
    Every multi-state $\varrho: \{1,2\} \to \mathcal{D}(\mathcal{H})$ is imaginarity-free.
\end{lemma}

\subsection{Implications for Bargmann invariants realizable by imaginarity-free multi-states}

We can use Theorem~\ref{theorem: Characterization_Imaginarity} to show a few structural results about sets of tuples of Bargmann invariants $\mathfrak{Q}(G)^{(d)}$ when we restrict realizations to a certain Hilbert space dimension $d$. We start by showing that for any graph $G = (V,E)$ where $\vert V \vert =  n \geq 3$ not all two-state overlaps (provided that we restrict the Hilbert space dimension) can be reached by real-only single-qubit states. 

\begin{corollary}\label{corollary:gap_qubits_real_nonreal}
    Let $G=(V,E)$ be a simple graph of $n$ nodes. Then:
    \begin{enumerate}
    \item for $n \in \{1,2\}$, it holds that $\mathfrak{Q}(G)^{(d)} \vert_{\mathrm{real}} = \mathfrak{Q}(G)^{(d)}$ for all integers $d\geq 2$.
    \item for $n \geq 3$, it holds that $\mathfrak{Q}(G)^{(2)} \vert_{\mathrm{real}} \subsetneq \mathfrak{Q}(G)^{(2)}.$
    \item for $n\geq 3$, $d\geq 2$, and $\mathsf{\Delta}:E\to[0,1]$ it holds that $\mathsf{\Delta}(\varrho)\in \mathfrak{Q}(G)^{(d)}\vert_{\mathrm{real}}$ for all multi-states $\varrho:\{1,2\}\to \mathcal{D}(\mathcal{H})$. 
    \end{enumerate}
\end{corollary}

\begin{proof}
The first and third part follow from Lemma~\ref{lemma:every_multi_state_imaginarity_free}, i.e. from the fact that two states in any finite-dimensional Hilbert space span at most a two-dimensional space, which is isometrically isomorphic to qubit space. 

As for the second part, from $\mathrm{Tr}[\rho_k \rho_l]=\frac{1+\langle \mathbf{r}_k, \mathbf{r}_l\rangle}{2}$ and Theorem~\ref{theorem: Characterization_Imaginarity} it suffices to show that for every $n \geq 3$ there exists $\varrho:[n] \to \mathcal{D}(\mathbb{C}^2)$ such that $\mathrm{rank}(G_\varrho) = 3$, which is trivially true. 
\end{proof}	

Ref.~\cite{fernandes2024unitary} had previously shown that for $n=3$ these two sets are equal, i.e. $$\mathfrak{Q}(K_3) = \mathfrak{Q}(K_3)\vert_{\mathrm{real}},$$ and that for $n=4$ there is a gap between the sets if we restrict to the pure states. Symbolically,   $$\mathfrak{Q}(K_4)\vert_{\mathrm{pure}} \cap \mathfrak{Q}(K_4)\vert_{\mathrm{real}} \subsetneq \mathfrak{Q}(K_4)\vert_{\mathrm{pure}}.$$ 
Corollary~\ref{corollary:gap_qubits_real_nonreal} implies that a single two-state overlap cannot witness imaginarity. It also implies that in order to experimentally witness imaginarity semi-device independently by measuring two-state overlaps one can choose between two assumptions: either an assumption on the \textit{purity} of the states, and then follow the proposal from Ref.~\cite{fernandes2024unitary}; or an assumption on their underlying Hilbert space dimension, and then experimentally estimate the rank of $G_{\mathbf{r}_\varrho}$.  

Another interesting aspect is that third-order invariants \textit{completely} characterize the imaginarity of single-qubit (possibly mixed) states. Let $\rho_k=\frac{\mathbb{1}+ \mathbf{r}_k \cdot \boldsymbol{\sigma}}{2}.$ The condition  $\mathrm{rank}(G_{\mathbf{r}_\varrho})= 3$ is satisfied iff $\mathrm{span}_\mathbb{R}\{\mathbf{r}_k\}_{k=1}^n=3$ which holds iff there exists $k,l,m$ such that $\mathbf{r}_k,\mathbf{r}_l,\mathbf{r}_m$ are \emph{linearly independent}. Hence $\det[(\mathbf{r}_k,\mathbf{r}_l,\mathbf{r}_m)]\neq 0.$ It has been noted by Ref.~\cite[App. B]{zhang2025geometrysets} that $$\mathrm{Im}(\mathrm{Tr}[\rho_k\rho_l\rho_m])=\frac{1}{4}\det([\mathbf{r}_k,\mathbf{r}_l,\mathbf{r}_m]).$$
So we obtain that $\varrho$ has quantum imaginarity \emph{iff} there exists some $k,l,m$ such that $ \mathrm{Tr}[\rho_k\rho_l\rho_m] \notin \mathfrak{B}_3\vert_{\mathrm{real}} \subseteq \mathbb{R}.$

\begin{corollary}\label{corollary:single_invariant_sufficient}
    Fix $n \in \mathbb{N}$. $\varrho \notin \mathcal{R}_n(\mathbb{C}^2) \iff \mathrm{Tr}(\rho_k\rho_l\rho_m) \notin \mathfrak{B}_3\vert_{\mathrm{real}}$ for some triplet of labels $k,l,m\in [n]$.
\end{corollary}

In particular, $\varrho \in \mathcal{R}_3(\mathbb{C}^2)$ iff $\Delta(\varrho) = \mathrm{Tr}(\rho_1\rho_2\rho_3) \notin \mathfrak{B}_3\vert_{\mathrm{real}} = [-\sfrac{1}{8},1]$. In other words, the quantum imaginarity of a set of three quantum states in a single-qubit system is completely characterized by their third order Bargmann invariant \emph{alone}. 

Interestingly, for higher-dimensional systems third-order invariants are \emph{not} sufficient to characterize multi-state imaginarity as we now show with an example. Choose $\mathcal{H} = \mathbb{C}^4$. We want to construct a multi-state $\varphi:\{1,2,3\} \to \mathcal{D}(\mathbb{C}^2)$ such that $\varphi \notin \mathcal{R}_3(\mathbb{C}^4)$ but that, nevertheless, $\mathrm{Tr}(\varphi_1\varphi_2\varphi_3) \in \mathfrak{B}_3\vert_{\mathrm{real}}$. 

First, we consider two multi-states $\varrho,\varsigma: \{1,2,3\} \to \mathcal{D}(\mathbb{C}^2)$ given by: {$$
\begin{array}{lll}
\rho_1=\left[\begin{array}{cc}
\frac{1}{3} & \frac{\mathrm{i}}{3} \\[2mm]
-\frac{\mathrm{i}}{3} & \frac{2}{3}
\end{array}\right], & \rho_2=\left[\begin{array}{cc}
\frac{1}{4} & \frac{\mathrm{i}}{5} \\[2mm]
-\frac{\mathrm{i}}{5} & \frac{3}{4}
\end{array}\right], & \rho_3=\left[\begin{array}{cc}
\frac{1}{6} & \frac{1}{7} \\[2mm]
\frac{1}{7} & \frac{5}{6}
\end{array}\right], \\[6mm]
\sigma_1=\left[\begin{array}{cc}
\frac{3}{4} & \frac{\mathrm{i}}{4} \\[2mm]
-\frac{\mathrm{i}}{4} & \frac{1}{4}
\end{array}\right], & \sigma_2=\left[\begin{array}{cc}
\frac{4}{5} & \frac{1}{8} \\[2mm]
\frac{1}{8} & \frac{1}{5}
\end{array}\right], & \sigma_3=\left[\begin{array}{cc}
\frac{1}{6} & \frac{\mathrm{i}}{7} \\[2mm]
-\frac{\mathrm{i}}{7} & \frac{5}{6}
\end{array}\right].
\end{array}
$$
In this case, we have that 
$$
\begin{aligned}
& \operatorname{Tr}\left[\rho_1 \rho_2 \rho_3\right]=\frac{1}{3 \times 4 \times 5 \times 6 \times 7}(1253+36 \mathrm{i}), \\
& \operatorname{Tr}\left[\sigma_1 \sigma_2 \sigma_3\right]=\frac{1}{4 \times 5 \times 6 \times 7 \times 8}(1192-200 \mathrm{i}),
\end{aligned}
$$
which implies that each of these multi-states has imaginarity. Now, we define a new multi-state 
$$
\varphi_i=\lambda \rho_i \oplus (1-\lambda)\sigma_i \equiv  \left[\begin{array}{cc}
\lambda \rho_i & \mathbb{0}_2 \\
\mathbb{0}_2 & (1-\lambda) \sigma_i
\end{array}\right],\\
$$
for all $i=1,2,3$. The third-order invariant in this case is given by
$$
\operatorname{Tr}[\varphi_1\varphi_2\varphi_3]=\lambda^3 \operatorname{Tr}\left[\rho_1 \rho_2 \rho_3\right]+(1-\lambda)^3 \operatorname{Tr}\left[\sigma_1 \sigma_2 \sigma_3\right].
$$

If we choose $\lambda\in (0,1)$ such that 
\begin{equation}\label{eq:lambda}\lambda^3 \frac{36 \mathrm{i}}{3 \times 4 \times 5 \times 6 \times 7}  + (1-\lambda)^3\frac{-200 \mathrm{i}}{4 \times 5 \times 6 \times 7 \times 8}  =0, \end{equation}
we obtain by construction that $\operatorname{Tr}[\varphi_1\varphi_2\varphi_3]\in \mathfrak{B}_3\vert_{\mathrm{real}} \subseteq \mathbb{R}.$ However, higher-order invariants are capable to witness the fact that the multi-state $\varphi$ as constructed \emph{has} imaginarity. Since $$
\begin{aligned}
& \operatorname{Tr}\left[\rho_1^2\rho_2 \rho_3\right]=\frac{1}{3^2 \times 4 \times 5 \times 6 \times 7}( 3199 +  108\mathrm{i}), \\
& \operatorname{Tr}\left[\sigma_1^2 \sigma_2 \sigma_3\right]=\frac{1}{4^2 \times 5 \times 6 \times 7 \times 8}(3760 -  800\mathrm{i}
),
\end{aligned}
$$
we have the fourth-order invariant 
$$\operatorname{Tr}[\varphi_1^2\varphi_2\varphi_3]=\lambda^4 \operatorname{Tr}\left[\rho_1^2 \rho_2 \rho_3\right]+(1-\lambda)^4 \operatorname{Tr}\left[\sigma_1^2 \sigma_2 \sigma_3\right]\notin \mathbb{R}$$ 
witnesses the imaginarity of $\varphi:\{1,2,3\} \to \mathcal{D}(\mathbb{C}^4)$. The invariant $\operatorname{Tr}[\varphi_1^2\varphi_2\varphi_3] \notin \mathfrak{B}_4\vert_{\mathrm{real}}$ cannot be real since otherwise
$$\lambda^4 \frac{36 \mathrm{i}}{3 \times 4 \times 5 \times 6 \times 7}  + (1-\lambda)^4\frac{-200 \mathrm{i}}{4 \times 5 \times 6 \times 7 \times 8}  =0, $$ 
which is in contradiction with Eq.~\eqref{eq:lambda}.

This construction provides a concrete example to a broad result from Ref.~\cite{oszmaniec2024measuring}, where the authors showed that Bargmann invariants of repeated entries are, in general, necessary to completely characterize the unitary-invariant properties of a set of \emph{mixed} states.

\subsection{On the geometry of the set of imaginarity-free multi-states}

For the case of quantum coherence, Designolle et al.~\cite{designolle2021set} had previously showed that $\mathcal{I}_n(\mathcal{H})$ is not convex. Here, we show that $\mathcal{R}_n(\mathcal{H})$ is also, in general, not a convex set. 
\begin{corollary}
    The set $\mathcal{R}_3(\mathbb{C}^2)$ is non-convex. In other words, there exist $\varrho,\varsigma \in \mathcal{R}_3(\mathbb{C}^2)$ and $p \in (0,1)$ such that the new multi-state defined by 
    \begin{equation}
        \xi = p \varrho + (1-p)\varsigma 
    \end{equation}
    has imaginarity. 
\end{corollary}

\begin{proof}
    Let us consider $\varrho:\{1,2,3\} \to \mathcal{D}(\mathbb{C}^2)$ such that $\varrho(i) := \frac{\mathbb{1}+\mathbf{r}_i \cdot \boldsymbol{\sigma}}{2}$ defined by the Bloch vectors 
        \begin{align*}
        \mathbf{r}_1 = (1,0,0),\,
        \mathbf{r}_2 = (0,1,0),\,
        \mathbf{r}_3 = (\sfrac{1}{\sqrt{2}},\sfrac{1}{\sqrt{2}},0),
    \end{align*}
    and $\varsigma:\{1,2,3\} \to \mathcal{D}(\mathbb{C}^2)$ such that $\varsigma(i) := \frac{\mathbb{1}+\mathbf{s}_i \cdot \boldsymbol{\sigma}}{2}$ defined instead by the vectors
      \begin{align*}
        \mathbf{s}_1 = (0,\sfrac{1}{\sqrt{2}},\sfrac{1}{\sqrt{2}}),\,
        \mathbf{s}_2 = (0,1,0),\,
        \mathbf{s}_3 = (0,0,1).
    \end{align*}   
    Let $A = [\mathbf{r}_1;\mathbf{r}_2;\mathbf{r}_3]$ and $B=[\mathbf{s}_1;\mathbf{s}_2;\mathbf{s}_3]$. We now note that $\mathrm{rank}(G_{\mathbf{r}_{\varrho}}) = \mathrm{rank}(A)$ and $\mathrm{rank}(G_{\mathbf{s}_{\varsigma}}) = \mathrm{rank}(B)$ (this follows from the fact that $G_{\mathbf{r}_\varrho} = AA^T$ and $\mathrm{rank}(AA^T) = \mathrm{rank}(A)$). From Theorem~\ref{theorem: Characterization_Imaginarity}, we have that $\varrho,\varsigma \in \mathcal{R}_3(\mathbb{C}^2)$. Nevertheless, for every $p \in (0,1)$ if we write $\xi$ as the convex combination of $\varrho$ and $\varsigma$ with weight $p$ we have that $\mathrm{rank}(G_{\mathbf{r}_{\xi}}) = \mathrm{rank}(C_p)$ where
    \begin{equation}
        C_p = \left(\begin{matrix}
            1-p & \sfrac{p}{\sqrt{2}} & \sfrac{p}{\sqrt{2}}\\
            0 & 1 & 0 \\
            \sfrac{(1-p)}{\sqrt{s}} & \sfrac{(1-p)}{\sqrt{s}} & p 
        \end{matrix}\right)
    \end{equation}
    which has determinant $\det(C_p) = \sfrac{(1-p)p}{2} \neq 0$ for all $p \in (0,1)$. From this we conclude that $\mathrm{rank}(G_{\mathbf{r}_{\xi}}) = 3$ for \emph{every} $0 < p<1$. Due to Theorem~\ref{theorem: Characterization_Imaginarity}, we conclude that $\xi \notin \mathcal{R}_3(\mathbb{C}^2)$ and therefore that $\mathcal{R}_3(\mathbb{C}^2)$ is not convex.  
\end{proof}

Indeed, our constructive proof demonstrates not only that the set is nonconvex, but also provides an explicit example of a line segment whose interior lies entirely outside the imaginarity‐free set while its extremal points lie inside.

\section{Bounding quantifiers of multi-state imaginarity}

We now move from the characterization of multi-state imaginarity to the quantification of this resource. Following the work by Miyazaki and  Matsumoto~\cite{miyazaki2022imaginarityfree}, we can define a function $\mathsf{Im}_{R_1}:\mathcal{R}_n(\mathcal{H}) \to \mathbb{R}_{\geq 0}$ aimed to quantify the imaginarity of a multi-state: 	
	$$\begin{array}{rl}
		\mathsf{Im}_{R_1}(\varrho):=&\displaystyle \min\limits_{U}\frac{1}{n} \sum_{j=1}^n \mathsf{Im}_R(U\rho_j U^\dagger),
	\end{array}$$
where, for the case of single-qubit multi-states, the minimization ranges over all unitaries $U \in \mathrm{U}(2)$ and $\mathsf{Im}_R(\rho)$ is the generalized robustness of imaginarity of a single-state~\cite{wu2021operational,wu2021resource,wu2023resource} relative to the standard basis $\mathbbm{A}_{\mathrm{st}} = \{\vert 0\rangle,\vert 1\rangle\}$: 
	$$\mathsf{Im}_R(\rho):=\min_{\tau \in \mathcal{D}(\mathbb{C}^2)} \left\{s\geq 0\mid \frac{  \rho+s \tau }{s+1} \in \mathcal{R}(\mathbb{C}^2,\mathbbm{A}_{\mathrm{st}})\right\}.$$

Since $\mathsf{Im}_R(\rho)= \sfrac{1}{2}\Vert \rho-\rho^T\Vert_1$~\cite{wu2021operational,wu2021resource}, using the Bloch vector decomposition of a state $\rho(\mathbf{r}) = \frac{\mathbb{1}+ \mathbf{r} \cdot \boldsymbol{\sigma}}{2} $ where $ \mathbf{r}=(r_x,r_y,r_z)$ we have that for the single-qubit case the robustness of imaginarity relative to the standard basis becomes simply 
\begin{equation}
\mathsf{Im}_R(\rho)=|r_y|.
\end{equation}
With that, the multi-state function $\mathsf{Im}_{R_1}(\varrho)$ becomes 
	\begin{align}\label{eq:Cal_CI1}
		\mathsf{Im}_{R_1}(\varrho)&=\displaystyle \min\limits_{\mathbf{p} \in \mathbb{S}^2}\frac{1}{n} \sum_{j=1}^n ||\mathbf{r}_j|| \cdot |\cos (\mathbf{r}_j, \mathbf{p})|\\&= \displaystyle \min\limits_{\mathbf{p} \in \mathbb{S}^2}\frac{1}{n} \sum_{j=1}^n | \langle \mathbf{r}_j, \mathbf{p}\rangle|,\nonumber
	\end{align}
 where $\mathbb{S}^2 = \{\mathbf{r}\in\mathbb{R}^3 \mid \Vert \mathbf{r} \Vert=1\}$ is the unit sphere. Moreover, can we write $\mathsf{Im}_{R_1}(\varrho)$ as 
    \begin{equation}
        \mathsf{Im}_{R_1}(\varrho) = \min_{\psi \in \mathcal{P}_1(\mathbb{C}^2)}\frac{1}{n}\sum_{i=1}^n\vert2\mathrm{Tr}(\varrho(i)\psi)-1\vert 
    \end{equation}
    so that $\mathsf{Im}_{R_1}$ becomes a function of two-state overlaps only. We can then frame this as a semidefinite programming (SDP) optimization problem~\cite{tavakoli2023semidefinite}: Given $\varrho = (\rho_1,\dots,\rho_n)$ we need to solve
    \begin{align*}
        \mathrm{minimize}\,\,&\frac{1}{n}\sum_{i=1}^n\vert2\mathrm{Tr}(\rho_i X)-1\vert,\\
        \text{subject to}\,\, &X \geq 0\\
        &\mathrm{Tr}(X)=1.
    \end{align*}
which is a semidefinite program (up to a standard trick of transforming an optimization over $\vert \cdot \vert$ into a linear program). 

Given that one can experimentally infer $G_{\mathbf{r}_\varrho}$ via the measurement of two-state overlaps, it is also interesting to see whether we can bound the possible values  $\mathsf{Im}_{R_1}(\varrho)$ can take dependent on properties of the Gram matrix  $G_{\mathbf{r}_\varrho}$. 

We now show that the eigenvalues of $G_{\mathbf{r}_\varrho}$ provide upper and lower bounds for $\mathsf{Im}_{R_1}(\varrho)$:

\begin{theorem} \label{theorem: Quantification_imaginarity}
		Fix \(n \geq 3\),  and any $n$-tuple of states $\varrho=(\rho_1,\rho_2,\cdots,\rho_n)\in \mathcal{D}(\mathbb{C}^2)^n$. Then  
		$$\frac{\lambda_3 (G_{\mathbf{r}_\varrho})}{n}\leq \mathsf{Im}_{R_1}(\varrho)\leq \sqrt{\frac{\lambda_3 (G_{\mathbf{r}_\varrho})}{n}}\leq \frac{1}{\sqrt{3}}$$
		where $\lambda_1 (G_{\mathbf{r}_\varrho})\geq \lambda_2 (G_{\mathbf{r}_\varrho})\geq \lambda_3 (G_{\mathbf{r}_\varrho}) \geq \lambda_4 (G_{\mathbf{r}_\varrho}) =\cdots =  \lambda_n (G_{\mathbf{r}_\varrho})=0$ denote the eigenvalues of $G_{\mathbf{r}_\varrho} $ in decreasing order.
\end{theorem}

\begin{proof}
It is elementary to see (using the Cauchy--Schwarz inequality) that the following inequality always hold: 
	$$ \sum_{j=1}^n | \langle \mathbf{r}_j, \mathbf{p}\rangle|^2 \leq \sum_{j=1}^n | \langle \mathbf{r}_j, \mathbf{p}\rangle| \leq  {\sqrt{n} \sqrt{\sum_{j=1}^n  | \langle \mathbf{r}_j, \mathbf{p}\rangle|^2}} .$$  Taking minimum over the unit sphere $\mathbf{p} \in \mathbb{S}^2$ we obtain
	\begin{equation}\label{eq:IR1_up_lower}
	    m \leq  \min\limits_{\mathbf{p} \in \mathbb{S}^2}  \sum_{j=1}^n | \langle \mathbf{r}_j, \mathbf{p}\rangle|\leq \sqrt{nm}.
	\end{equation}  
	where $ m= \displaystyle \min\limits_{\mathbf{p} \in \mathbb{S}^2}  \sum_{j=1}^n | \langle \mathbf{r}_j, \mathbf{p}\rangle|^2$.  Moreover, one finds that 
	\begin{equation}\label{eq:minLamda}
	\displaystyle m=\min\limits_{\mathbf{p} \in S^2} \sum_{j=1}^n | \langle \mathbf{r}_j, \mathbf{p}\rangle|^2= \lambda_3(G_{\mathbf{r}_\varrho}) .
    \end{equation}
    In fact,  $$\sum_{j=1}^n | \langle \mathbf{r}_j, \mathbf{p}\rangle|^2=\langle \mathbf{p} | \left(\sum_{j=1}^n   |\mathbf{r}_j\rangle \langle   \mathbf{r}_j|\right)  |\mathbf{p}\rangle$$
    and if we minimize the right hand side we get that $m$ is equal to the smallest eigenvalue of the  $3\times 3$ matrix  $\sum_{j=1}^n   |\mathbf{r}_j\rangle \langle   \mathbf{r}_j|$.

    If we denote $A=[\mathbf{r}_1;\ldots;\mathbf{r}_n]$ as a $3\times n$ matrix, then  $$AA^\dagger=\sum_{j=1}^n   |\mathbf{r}_j\rangle \langle   \mathbf{r}_j|, \text{ and }  A^\dagger A=G_{\mathbf{r}_\varrho}.$$
    
    Note that $AA^\dagger$ and $A^\dagger A$ share the same nonzero eigenvalues. So the minimal eigenvalue of $\sum_{j=1}^n   |\mathbf{r}_j\rangle \langle   \mathbf{r}_j|$, which is $3 \times 3$ and therefore has only 3 eigenvalues, is equal to the third largest eigenvalue of $G_{\mathbf{r}_\varrho}$ in decreasing order, since all others need to be equal to zero as the Gram matrix has rank $3$. Therefore, by Eq.~\eqref{eq:IR1_up_lower} and Eq.~\eqref{eq:minLamda}, we have the upper and lower bounds on $\mathsf{Im}_{R_1}(\varrho)$. 

    To conclude, note that $G_{\mathbf{r}_\varrho}$  is at most rank three, so at most three eigenvalues are nonzero,
$$\lambda_1+\lambda_2+\lambda_3=\lambda_1+\lambda_2+\lambda_3+\cdots+\lambda_n= \mathrm{Tr}[G_{\mathbf{r}_\varrho}]\leq n.$$
Hence 
$$3 \lambda_3\leq \lambda_1+\lambda_2+\lambda_3 \leq n  $$
 which implies  $$\sqrt{\frac{\lambda_3 (G_{\mathbf{r}_\varrho})}{n}}\leq \frac{1}{\sqrt{3}}.$$ 
\end{proof}

\section{Characterizing all Bargmann invariants of single-qubit states}\label{sec:structure_Bargmann_invariants}

We now proceed to investigate multi-state coherence in general, moving beyond the specific case of imaginarity discussed so far. From the perspective of unitary-invariance, the results of Ref.~\cite{oszmaniec2024measuring} reviewed in Sec.~\ref{sect:Pre} show that estimating sufficiently many Bargmann invariants associated to a multi-state $\varrho$ yield necessary and sufficient information for determining whether $\varrho \in \mathcal{I}_n(\mathcal{H})$. This follows since membership in this set can be viewed as an instance of a unitary equivalence problem.

In this section we present one of the main results of our work. When restricted to single-qubit multi-states $\varrho: [n] \to \mathcal{D}(\mathbb{C}^2)$, the collection of two-state overlaps $\mathsf{\Delta}(\varrho) = (\mathrm{Tr}(\rho_i\rho_j))_{i,j}$ encode \emph{almost all} the necessary and sufficient unitary-invariant information to characterize this notion of basis-independent coherence. 

\begin{theorem}\label{theorem:Bargmanns_dependence_overlaps}
For any $n$-th order Bargmann invariant $\operatorname{Tr}(\rho_1 \cdots \rho_n) \in \mathfrak{B}_n^{(2)}$ where $\rho_i\in \mathcal{D}(\mathbb{C}^2)$, its value is completely determined by the overlaps ${(\operatorname{Tr}(\rho_i\rho_j))}_{i,j=1}^n$ up to complex conjugation.
More precisely, there exist polynomials $P_n, Q_n \in \mathds{Q}[\Delta_{11}, \Delta_{12}, ..., \Delta_{nn}]$  such that  the Bargmann invariant satisfies the quadratic equation:
$$[\operatorname{Tr}(\rho_1 \cdots \rho_n)]^2-2P_n(\mathsf{\Delta}(\varrho)) \operatorname{Tr}(\rho_1 \cdots \rho_n)+Q_n(\mathsf{\Delta}(\varrho)) =0$$
where $P_n(\mathsf{\Delta}(\varrho))$ is the quantum realization of a polynomial obtained by substituting $\Delta_{ij}$ with $\Delta_{ij}(\varrho) = \mathrm{Tr}[\rho_i\rho_j].$ 
\end{theorem}

We prove this theorem in Appendix~\ref{app:proof_1}. We can view this result as a   generalization for all possible values $n \geq 3$  of the results from Ref.~\cite{zhang2025geometrysets}. Consider the case $n=3$.  Ref.~\cite{zhang2025geometrysets} has shown that 
$$ 
\left\{\begin{array}{l}
\mathrm{Re}[\Delta_3(\varrho)]=  \frac{1}{4}\left(1+\sum_{1 \leqslant i<j \leqslant 3}\left\langle\mathbf{r}_i, \mathbf{r}_j\right\rangle \right), \\
\mathrm{Im}[\Delta_3(\varrho)]= \frac{1}{4}\operatorname{det}\left(\mathbf{r}_1;\mathbf{r}_2;\mathbf{r}_3\right).
\end{array}\right.
$$

This translates to polynomials that read
 $$P_3(\mathsf{\Delta})=\frac{1}{2}\left( \Delta_{12}+\Delta_{13}+\Delta_{23}-1\right) ,$$ 
 $$
 \begin{array}{rcl}
 Q_3(\mathsf{\Delta})&=&\frac{1}{16} \det\left(\begin{array}{ccc}
 2\Delta_{11}-1& 2\Delta_{12}-1& 2\Delta_{13}-1\\[2mm]
 2\Delta_{21}-1& 2\Delta_{22}-1& 2\Delta_{23}-1\\[2mm]
 2\Delta_{31}-1& 2\Delta_{32}-1& 2\Delta_{33}-1 
 \end{array}\right)^2\\[3mm]
 && +\frac{1}{4}(\Delta_{12}+\Delta_{13}+\Delta_{23}-1)^2.
 \end{array}
 $$

In Ref.~\cite{zhang2025geometrysets} (see also Appendix~\ref{app:proof_1}) the cases $n=4$ and $n=5$ are also provided. 

As a matter of fact, the proof of Theorem~\ref{theorem:Bargmanns_dependence_overlaps} provides an iterative algorithm for updating the polynomials $P_n$ and $Q_n$ that yield the real and imaginary part of a (single-qubit) Bargmann invariant $\mathrm{Tr}(\rho_1 \cdots \rho_n)$ provided that one knows $P_{n-1}$ and $Q_{n-1}$. 

\section{Characterization of single-qubit multi-state coherence}\label{sec:singlequbit_coherence}

\subsection{Rank-based criteria}

It is trivial---but merits to be mentioned for completeness of our discussion---to show an analog of Theorem~\ref{theorem: Characterization_Imaginarity} for the case of quantum coherence:

\begin{theorem} \label{theorem: Characterization_Coherence}
		A single-qubit multi-state $\varrho: [n] \to \mathcal{D}(\mathbb{C}^2)$ has coherence, i.e. $\varrho \notin \mathcal{I}_n(\mathbb{C}^2)$, iff $\mathrm{rank}(G_{\mathbf{r}_\varrho})\geq 2$.
\end{theorem}

\begin{proof}
Clearly, $\varrho$ has no coherence iff the elements in $\varrho([n])$ are mutually commutative. It suffices to prove that $\rho_i,\rho_j \in \varrho([n])$ commute for all $i,j$ iff $\mathrm{rank}(G_{\mathbf{r}_\varrho})\leq 1$. If all elements in $\varrho([n])$ are mutually commutative there exists some unitary $U\in \mathrm{U}(2)$ such that $U\rho_k U^\dagger$ are all in diagonal form. Equivalently, $\Phi_U \mathbf{r}_k \propto (1,0 ,0)^T$ which also implies that $\mathbf{r}_k\propto \Phi_U^{-1} (1,0,0)^T.$ Therefore, the rank of $G_{\mathbf{r}_\varrho}$ is no larger than 1. On the other hand, if $\mathrm{rank}(G_{\mathbf{r}_\varrho})\leq 1$, without loss of generality, we assume that $\mathbf{r}_k=a_k \mathbf{r}_1$ for some real number $a_k$. As $\mathbf{r}_k\times \mathbf{r}_l=\mathbf{0},$ we have
\begin{align*}
( \mathbf{r}_k \boldsymbol{\sigma} )( \mathbf{r}_l \boldsymbol{\sigma} )&=\langle \mathbf{r}_k, \mathbf{r}_l\rangle \mathbb{1}+ \mathrm{i}(\mathbf{r}_k \times \mathbf{r}_l) \boldsymbol{\sigma} \\&=\langle \mathbf{r}_k, \mathbf{r}_l\rangle \mathbb{1}=( \mathbf{r}_l \boldsymbol{\sigma} )( \mathbf{r}_k \boldsymbol{\sigma} ),  
\end{align*}
 which implies the commutative of $\rho_k$ and $\rho_l$.
\end{proof}

For the case of quantum coherence we have an even simpler situation than the geometric one provided by Fig.~\ref{fig:Theorem_1_Bloch_sphere}. A set of two single-qubit states is coherent iff their Bloch vectors all lie in a rotated axis of the Bloch sphere.

It is relevant to mention that multi-states $\varrho: \{1,2\} \to \mathcal{D}(\mathcal{H})$ can be coherent (contrasting  Lemma~\ref{lemma:every_multi_state_imaginarity_free} which applies to the specific case of imaginarity) as we will shortly see, but this cannot be witnessed by a single two-state overlap alone, unless one assumes that the states are \emph{pure}. 

\subsection{A simple remark beyond single-qubit states}

We can state a necessary requirement for generic high-dimensional multi-states $\varrho: [n] \to \mathcal{D}(\mathbb{C}^d)$ to be incoherent, similar to Proposition~\ref{proposition:PartialCharacterization}. 

\begin{proposition}\label{proposition:PartialCharacterization_Coherence}
    If $\varrho: [n] \to \mathcal{D}(\mathbb{C}^d)$ is an incoherent multi-state, i.e. $\varrho \in \mathcal{I}_n(\mathbb{C}^d)$, then $\mathrm{rank}(G_{\mathbf{r}_\varrho}) \leq d-1$, where $\mathbf{r}_\varrho$ are generalized Bloch vectors given with respect to the Gell--Mann basis. 
\end{proposition}

\begin{proof}
    If $\varrho$ is incoherent, its image is described by a set of pairwise commuting states. Therefore, there exists a basis $\mathbbm{A}$ for which the associated density matrices are all diagonal in this basis. The traceless diagonal Gell--Mann matrices required to write such matrices span a $d-1$ space, which implies that the associated Gram matrix $G_{\mathbf{r}_{U\varrho U^\dagger}}$ has rank at most $d-1$. From unitary invariance, $G_{\mathbf{r}_\varrho}$ has also at most rank $d-1$. 
\end{proof}

It is simple to write down a counter-example. Take $d=3$ and the Gell--Mann matrices $\{\lambda_i\}_{i=1}^8$. Consider the two states given by 
\begin{align}
    \rho_1 &=\frac{1}{3}(\mathbb{1}_{3} + \frac{1}{2}\lambda_1) ,\\
    \rho_2 &= \frac{1}{3}(\mathbb{1}_{3} + \frac{1}{2}\lambda_3).
\end{align}
This pair $\{\rho_1,\rho_2\}$ is set coherent since $[\rho_1,\rho_2] \neq 0$ (which follows from the fact that $[\lambda_1,\lambda_3] \neq 0$). However, $\mathrm{rank}(G_{\mathbf{r}_\varrho}) =2 = 3-1$.

We can leverage the system independence of Bargmann invariants mentioned in Sec.~\ref{sect:Pre} to describe a simple necessary criteria for multi-state coherence, which can be viewed as a generalization of a remark made in Ref.~\cite{fernandes2024unitary}. There, the authors comment that its possible to use the imaginary part of Bargmann invariants to witness multi-state imaginarity since one has the implication
\begin{equation}
    \mathrm{Im}[\mathrm{Tr}(\rho_1\rho_2\rho_3)] \neq 0 \Rightarrow (\rho_1,\rho_2,\rho_3) \notin \mathcal{R}_3(\mathcal{H}),
\end{equation}
which holds for every space $\mathcal{H}$, regardless of the purity of the quantum states. 

We can note that the condition $\mathrm{Im}[\mathrm{Tr}(\rho_1\rho_2\rho_3)] \neq 0$ can be viewed as the violation of an \emph{equality constraint}, which follows from the fact that 
\begin{equation}
    \mathrm{Tr}(\rho_1\rho_2\rho_3) \neq \mathrm{Tr}(\rho_1\rho_3\rho_2).
\end{equation}
When the equality is satisfied this has previously been called the \emph{weak commutativity}~\cite{miyazaki2022imaginarityfree} condition, mostly in the context of multi-parameter estimation in quantum metrology~\cite{matsumoto2002newapproach}.

We can consider the following  simple witnesses of set coherence valid for all Hilbert spaces $\mathcal{H}$: let $\varrho: [n] \to \mathcal{D}(\mathcal{H})$ be any multi-state, $i_1,\ldots,i_m \in [n]$ any finite sequence of labels, and $\mathtt{P}$ any permutation of $m$ indexes. We have that 
\begin{equation}\label{eq:equality_constraints_Bargmanns}
    \mathrm{Tr}(\rho_{i_1} \cdots \rho_{i_m}) \neq \mathrm{Tr}(\rho_{\mathtt{P}^{-1}(i_1)} \cdots \rho_{\mathtt{P}^{-1}(i_m)}) \Rightarrow \varrho \notin \mathcal{I}_n(\mathcal{H}).
\end{equation}
We can illustrate this with a simple example (borrowing from the two multi-states considered by Designolle et al.~\cite{designolle2021set} to show that $\mathcal{I}_n(\mathcal{H})$ is not convex). Let $\varrho,\varsigma:\{1,2\} \to \mathcal{D}(\mathbb{C}^2)$ given by
\begin{equation}\label{eq: vertex label rho example}
    \varrho(1) = \vert 0\rangle \langle 0 \vert, \varrho(2) = \frac{1}{3}\vert 0\rangle \langle 0 \vert + \frac{2}{3}\vert 1\rangle \langle 1 \vert 
\end{equation}
and
\begin{equation}\label{eq: vertex label sigma example}
    \varsigma(1) = \vert +\rangle \langle + \vert, \varsigma(2) = \frac{1}{4}\vert +\rangle \langle + \vert + \frac{3}{4}\vert -\rangle \langle - \vert. 
\end{equation}
We have that both $\varrho,\varsigma \in \mathcal{I}_2(\mathbb{C}^2)$. It is easy to see that for every $\omega \in (0,1)$ the multi-state $\xi := \omega \varrho+(1-\omega)\varsigma \notin \mathcal{I}_2(\mathbb{C}^2)$ since  
\begin{align*}
    \text{Tr}(\xi_1\xi_1\xi_2\xi_2) &= \frac{1}{144} (9 + 18 \omega + 94 \omega^2 - 54 \omega^3 + 13 \omega^4)\\
    \text{Tr}(\xi_1\xi_2\xi_1\xi_2) &= \frac{1}{144} (9 + 18 \omega + 90 \omega^2 - 46 \omega^3 + 9 \omega^4)
\end{align*}
from which we conclude that 
\begin{equation}
    \text{Tr}(\xi_1\xi_1\xi_2\xi_2)-\text{Tr}(\xi_1\xi_2\xi_1\xi_2) = \frac{1}{36}\omega^2(1-\omega)^2 \neq 0.
\end{equation}
Therefore, using this simple equality-based witness one can infer the coherence of a multi-state $\xi$.

\section{Bounding quantifiers of multi-state coherence}

Proceeding similarly as in the case of multi-state imaginarity, we can provide explicit forms and precise bounds on the quantifier $\mathsf{C}_{R_1}:\mathcal{I}_n(\mathcal{H}) \to \mathbb{R}_{\geq 0}$ for multi-state coherence (proposed by Designolle et al.~\cite{designolle2021set}), which is defined as
$$\begin{array}{rl}
	\mathsf{C}_{R_1}(\varrho):=&\displaystyle \min\limits_{U}\frac{1}{n} \sum_{j=1}^n \mathsf{C}_R(U\rho_j U^\dagger),
 \end{array}$$
where, for single qubits, $\mathsf{C}_R$ denotes the generalized robustness of quantum coherence $$\mathsf{C}_R(\rho):=\min_{\tau \in \mathcal{D}(\mathbb{C}^2)} \{s\geq 0\mid \frac{  \rho+s \tau }{s+1} \in \mathcal{I}(\mathbb{C}^2,\mathbbm{A}_{\mathrm{st}})\}.$$
Using the Bloch vector decomposition of a state $\rho(\mathbf{r}) = \frac{\mathbb{1}_{2}+ \mathbf{r} \cdot \boldsymbol{\sigma}}{2} $ where $ \mathbf{r}=(r_x,r_y,r_z)$ we have that for the single-qubit case the robustness of coherence relative to the standard basis becomes simply 
\begin{equation}\label{eq:C_R_Im_R}
\mathsf{C}_R(\rho)=|r_x+\mathrm{i}\, r_y|=\sqrt{r_x^2+r_y^2}\geq \mathsf{Im}_R(\rho).
\end{equation}
As before when focusing on $\mathsf{Im}_{R_1}(\varrho)$, the minimization in the definition of $\mathsf{C}_{R_1}(\varrho)$ is taken over all unitaries $U \in \mathrm{U}(2)$ for the case of single-qubit multi-states.

By definition (and Eq. \eqref{eq:C_R_Im_R}]), for every multi-state $\varrho$ the following inequality holds:  
 $$\mathsf{Im}_{R_1}(\varrho)\leq \mathsf{C}_{R_1}(\varrho).$$  
Moreover, in the case of single-qubit multi-states we have that $\mathsf{C}_{R_1}(\varrho) $ takes the form  
 \begin{equation}\label{eq:Cal_CR1}
 		\mathsf{C}_{R_1}(\varrho)=\displaystyle \min\limits_{\mathbf{p} \in \mathbb{S}^2}\frac{1}{n} \sum_{j=1}^n ||\mathbf{r}_j|| \cdot |\sin (\mathbf{r}_j, \mathbf{p})| 
 \end{equation} 
where $(\mathbf{r}_j, \mathbf{p})$ represents the angle between the Bloch vector $\mathbf{r}_j$ of $\rho_j\in \varrho$ and a unit vector $\mathbf{p}$. 

Equation~\eqref{eq:Cal_CR1} is taken from Ref.~\cite{designolle2021set} and it is obtained as follow: Let $\rho(\mathbf{r}) = \frac{\mathbb{1}+ \mathbf{r} \cdot \boldsymbol{\sigma}}{2} $ where $ \mathbf{r}=(r_x,r_y,r_z)$. Then, the Bloch vector of $U\rho U^\dagger$ is $\Phi_U \mathbf{r}$ where $\Phi_U\in \mathrm{SO}(3)$.  Since 
$$ \mathsf{C}_{R}(\rho)= \sqrt{r_x^2+r_y^2}= || \mathbf{r}|| \cdot  |\sin (\mathbf{r}, \mathrm{e}_z) |, $$
where $e_z=(0,0,1)$ which follows from $|\sin (\mathbf{r}, \mathrm{e}_z) |= \sqrt{1-\cos^2 (\mathbf{r}, \mathrm{e}_z)}=\frac{\sqrt{||\mathbf{r}||^2-\langle \mathbf{r}, \mathrm{e}_z\rangle^2}}{||\mathbf{r}||}=\frac{\sqrt{r_x^2+r_y^2}}{{||\mathbf{r}||}}$ we find that
$$\mathsf{C}_{R}(U\rho U^\dagger)= \Vert \Phi_U\mathbf{r}\Vert \cdot  |\sin (\Phi_U\mathbf{r}, \mathrm{e}_z) |=\Vert \mathbf{r} \Vert \cdot  |\sin (\mathbf{r}, \Phi_U\mathrm{e}_z) |.$$
Above, we have used the fact that unitaries preserve purity as viewed from the homomorphism $\Phi$, i.e. that $||\Phi_U\mathbf{r}||=||\mathbf{r}||$. Moreover, we have also used that $(\Phi_U\mathbf{r}, \mathrm{e}_z)= (\mathbf{r}, \Phi_U^T\mathrm{e}_z).$
When optimizing $U$ over all unitaries,  $\Phi_U^T$ ranges over all elements in $\mathrm{SO}(3)$, so defining $\mathbf{p}:=\Phi_U^T\mathrm{e}_z$ makes the optimization to be equivalent to one ranging over all unit vectors in $\mathbb{S}^2.$

Using that $\vert \sin(\mathbf{r}_j,\mathbf{p})\vert = \sqrt{1-\left(\frac{\langle \mathbf{r}_j,\mathbf{p}\rangle}{\Vert \mathbf{r}_j\Vert}\right)^2}$ we can re-write $\mathsf{C}_{R_1}(\varrho)$ as
\begin{align*}
    &\mathsf{C}_{R_1}(\varrho)=\displaystyle \min\limits_{\mathbf{p} \in \mathbb{S}^2}\frac{1}{n} \sum_{j=1}^n \sqrt{\Vert \mathbf{r}_j\Vert^2 - \langle \mathbf{r}_j,\mathbf{p}\rangle^2 }\\
    &=\displaystyle \min\limits_{\psi\in\mathcal{P}_1(\mathbb{C}^2)}\frac{1}{n} \sum_{j=1}^n \sqrt{(2\mathrm{Tr}(\rho_j^2)-1)-(2\mathrm{Tr}(\rho_j\psi)-1)^2},\\
\end{align*}
which is now fully expressed in terms of two-state overlaps.~Since the expression is non-linear it is unclear whether there is a trick to transform the problem above into an SDP. 

For pure states, we know that the leading diagonal of $G_{\mathbf{r}_\varrho}$ has all entries equal to $1$ in which case we have that its trace, which equals the sum of all eigenvalues, is equal to 
$$\lambda_1 (G_{\mathbf{r}_\varrho})+\lambda_2 (G_{\mathbf{r}_\varrho})+\lambda_3 (G_{\mathbf{r}_\varrho})= n.$$
It is not difficult to show, following the same strategy as the one used in Theorem~\ref{theorem: Quantification_imaginarity}, that:

\begin{theorem} \label{theorem: Quantification_Coherence}
	Fix any \(n \geq 2\), and any $n$-tuple of states $\varrho=(\rho_1,\rho_2,\cdots,\rho_n)\in \mathcal{D}(\mathbb{C}^2)^n$. Then  
    \begin{align*}
    \frac{\lambda_2 (G_{\mathbf{r}_\varrho})+\lambda_3 (G_{\mathbf{r}_\varrho})}{n}\leq \mathsf{C}_{R_1}(\varrho)\leq \sqrt{\frac{\lambda_2 (G_{\mathbf{r}_\varrho})+\lambda_3 (G_{\mathbf{r}_\varrho})}{n}}
    \end{align*}
	where $\lambda_1 (G_{\mathbf{r}_\varrho})\geq \lambda_2 (G_{\mathbf{r}_\varrho})\geq \lambda_3 (G_{\mathbf{r}_\varrho}) \geq \cdots \geq  \lambda_n (G_{\mathbf{r}_\varrho})$ denote the   eigenvalues of $G_{\mathbf{r}_\varrho} $ in decreasing order. Moreover, $\mathsf{C}_{R_1}(\varrho) \leq\sqrt{\frac{2}{3}}$ for all $\varrho \in \mathcal{D}(\mathbb{C}^2)^n$.
\end{theorem}

\begin{proof}
    Using the Cauchy--Schwarz inequality 
    \begin{align*}
    &\sum_{j=1}^n \sqrt{\Vert \mathbf{r}_j\Vert^2 - \langle \mathbf{r}_j,\mathbf{p}\rangle^2 } \leq \sqrt{n \sum_{j=1}^n (\Vert \mathbf{r}_j \Vert^2-\langle \mathbf{r}_j,\mathbf{p}\rangle^2) }\\
    &=\sqrt{n (\mathrm{Tr}[G_{\mathbf{r}_\varrho}]-\langle \mathbf{p}\vert \sum_{j=1}^n \vert \mathbf{r}_j \rangle \langle \mathbf{r}_j \vert \,\,\vert \mathbf{p}\rangle )}.
    \end{align*}
    Above, we have used that $\sum_j\Vert \mathbf{r}_j\Vert^2 = \sum_j \langle \mathbf{r}_j,\mathbf{r}_j\rangle = \mathrm{Tr}(G_{\mathbf{r}_\varrho})$.
    
    Let us denote $A$ the $3 \times n$ matrix of column vectors $A = [\mathbf{r}_1;\ldots;\mathbf{r}_n]$ such that $AA^\dagger =  \sum_{j=1}^n \vert \mathbf{r}_j \rangle \langle \mathbf{r}_j \vert $. Dividing by $n$ and minimizing over all possible $\mathbf{p} \in \mathbb{S}^2$ we obtain
    \begin{align*}
        \mathsf{C}_{R_1}(\varrho) &\leq \min_{\mathbf{p}\in\mathbb{S}^2} \sqrt{\frac{\sum_{j=1}^n\lambda_j(G_{\mathbf{r}_\varrho})-\langle \mathbf{p}\vert AA^\dagger \vert \mathbf{p}\rangle }{n}}\\
        &=  \sqrt{\frac{\sum_{j=1}^n\lambda_j(G_{\mathbf{r}_\varrho})-\max_{\mathbf{p}\in\mathbb{S}^2}\langle \mathbf{p}\vert AA^\dagger \vert \mathbf{p}\rangle }{n}}\\
        &= \sqrt{\frac{\sum_{j=1}^n\lambda_j(G_{\mathbf{r}_\varrho})-\lambda_1(G_{\mathbf{r}_\varrho})}{n}}\\
        &= \sqrt{\frac{\lambda_2(G_{\mathbf{r}_\varrho})+\lambda_3(G_{\mathbf{r}_\varrho})}{n}}.
    \end{align*}
    Above, we have used that $$\max_{\mathbf{p}\in\mathbb{S}^2}\,\langle \mathbf{p}\vert AA^\dagger \vert \mathbf{p}\rangle=\lambda_1(AA^\dagger) = \lambda_1(G_{\mathbf{r}_\varrho})$$ since their spectrum is equal (up to zeros), i.e. $$\mathrm{spec}(AA^\dagger)\setminus \{0\} = \mathrm{spec}(G_{\mathbf{r}_\varrho})\setminus \{0\}.$$

Using $\sum_j \sqrt{a_j} \geq  {\sum_j a_j}$ which holds whenever $0\leq a_j \leq 1$ for all $j$ we have that
    \begin{equation}
        \sum_{j=1}^n \sqrt{\Vert \mathbf{r}_j\Vert^2 - \langle \mathbf{r}_j,\mathbf{p}\rangle^2 } \geq  {\sum_{j=1}^n (\Vert \mathbf{r}_j \Vert^2-\langle \mathbf{r}_j,\mathbf{p}\rangle^2) }.
    \end{equation}
    Therefore
    \begin{align*}
        \mathsf{C}_{R_1}(\varrho) &\geq \min_{\mathbf{p} \in \mathbb{S}^2}\frac{1}{n} {\sum_{j=1}^n (\Vert \mathbf{r}_j \Vert^2-\langle \mathbf{r}_j,\mathbf{p}\rangle^2) }\\
        &=\frac{1}{n} {\sum_{j=1}^n (\Vert \mathbf{r}_j \Vert^2-\max_{\mathbf{p}\in \mathbb{S}^2}\langle \mathbf{r}_j,\mathbf{p}\rangle^2)}\\
        &=\frac{1}{n} \left({\lambda_2(G_{\mathbf{r}_\varrho})+\lambda_3(G_{\mathbf{r}_\varrho})}\right).
    \end{align*}
    To conclude, again we note that $\lambda_1+\lambda_2+\lambda_3=n$ which implies that $\lambda_1 \geq n/3$. Therefore,
    \begin{equation*}
        \lambda_2+\lambda_3 = n-\lambda_1\leq n-\frac{n}{3} = \frac{2n}{3},
    \end{equation*}
    and we obtain our universal upper-bound
    \begin{equation}
        \mathsf{C}_{R_1}(\varrho) \leq \sqrt{\frac{\sfrac{2n}{3}}{n}} = \sqrt{\frac{2}{3}}.
    \end{equation}
\end{proof}

\section{Physical implications and applications}\label{sec:applications}

\subsection{On the relationship between spin chirality and multi-state imaginarity}

Corollary~\ref{corollary:single_invariant_sufficient} has a physical interpretation related to the concept of \emph{scalar spin chirality} following the results from Ref.~\cite{reascos2023quantum}. Given a totally separable tripartite pure spin state $$\vert \psi \rangle = \vert \mathbf{n}_1 \rangle \otimes \vert \mathbf{n}_2 \rangle \otimes \vert \mathbf{n}_3 \rangle$$ the (separable) scalar spin chirality is defined as $$\chi_{1,2,3} = S^3\langle \mathbf{n}_1, \mathbf{n}_2 \times \mathbf{n}_3\rangle = S^3\det([\mathbf{n}_1;\mathbf{n}_2;\mathbf{n}_3])$$ where $S$ is a constant. We can extend this notion of chirality to mixed states, allowing for the associated states to be mixed. Therefore, three (possibly mixed) single-qubit spin-$\sfrac{1}{2}$ states are chiral \emph{iff} the associated single-qubit multi-state has imaginarity. 

In chiral magnets, chirality gives rise to a phenomena known as the \emph{topological Hall effect}~\cite{kimbell2022challenges}, which is known to be a non-trivial manifestation of the \emph{Berry phase}~\cite{berry1984quantal,ye1999berry}. Corollary~\ref{corollary:single_invariant_sufficient} shows that this (separable) chirality happens iff there are triplets of spin states having imaginarity in this basis-independent sense. As a matter of fact, our Theorem~\ref{theorem:Bargmanns_dependence_overlaps} shows, furthermore, that for single-qubit spin-$\sfrac{1}{2}$ states the only information provided by the Berry phase---which is equal in module to the phase of a Bargmann invariant~\cite{simon1993bargmann}---not encoded already in the two-state overlaps is the information of whether the path in the Bloch sphere formed by the associated triplet of Bloch vectors takes is clockwise (say,  $\rho_1\rho_2\rho_3$) or counterclockwise ($\rho_3\rho_2\rho_1$).

\subsection{Operational advantage for  discrimination quantified by the robustness of multi-state imaginarity}

In abstract treatments of resource theories (focusing on a single state) robustness-based quantifiers are known to have an associated operational interpretation via instances of discrimination tasks~\cite{designolle2021set,uola2019quantifying,kuroiwa2024every,kuroiwa2024robustness,wu2021operational,takagi2019operational,napoli2016robustness,piani2016robustness,skrzypczyk2019all,oszmaniec2019operational,takagi2019general}.

For example, it is straightforward to translate the operational advantage in sub-channel discrimination discussed by Designolle et al.~\cite{designolle2021set} quantified there by the robustness of set coherence to the case of the robustness of set imaginarity since the only truly relevant aspect for the proof in Ref.~\cite{designolle2021set} stems from the form of the quantifier and the convexity of the sets $\mathcal{I}(\mathcal{H},\mathbbm{A})$. 

In the task of sub-channel discrimination, one is given a quantum instrument $\mathcal{T} = \{\mathcal{T}_a\}_a$ and a measurement $M=\{E_a\}_a$, and the goal of the task is to correctly guess which map $\mathcal{T}_a$ was applied on a state $\rho$ that can be prepared at will. 

The probability of successfully guessing $\mathcal{T}_a$ provided the state $\rho$ was prepared and a measurement $M$ was performed is given by
\begin{equation}
    p_{\mathrm{succ}}(\rho,\mathcal{T},M) = \sum_{a}\mathrm{Tr}[\mathcal{T}_a(\rho)E_a].
\end{equation}
We remark that the task above can also be made probabilistic assuming that a map $\mathcal{T}_a$ is applied with probability $p_a$. As mentioned also by Kang-Da Wu et al.~\cite{wu2021operational}, it holds that  
\begin{align*}
    \sup_{\mathcal{T},M} \frac{p_{\mathrm{succ}}(\rho,\mathcal{T},M)}{\max_{\sigma \in \mathcal{R}(\mathcal{H},\mathbbm{A}_{\mathrm{st}})}p_{\mathrm{succ}}(\sigma,\mathcal{T},M)} = 1 + \mathsf{Im}_R(\rho).
\end{align*}

Provided with that intuition, we can extend this to the case where we have a set of sub-channel discrimination tasks, as opposed to a single one. In this case, we are given lists of instruments $\boldsymbol{\mathcal{T}} =(\mathcal{T}^{(1)},\ldots, \mathcal{T}^{(n)})$ and measurements $\boldsymbol{M} = (M^{(1)},\ldots,M^{(n)})$. Our full protocol is then to implement a sub-channel discrimination task for each label $i=1,\ldots,n$, and our goal is to guess the instruments $\mathcal{T}_a^{(i)}$ for each $i$ given that we have prepared an initial state $\varrho(i)$ and performed measurement $M^{(i)}$. 

Summing over all instances of the above where we let $\rho$  be in the image of some multi-state, and then maximizing over all possible unitaries (which amounts to maximizing over all possible choices of reference basis $\mathbbm{A}$) we end up with 
\begin{align*}
    &\frac{1}{n}\min_U \sum_{i=1}^n \sup_{\boldsymbol{\mathcal{T}},\boldsymbol{M}} 
    \frac{p_{\mathrm{succ}}(\rho_i,\mathcal{T}^{(i)},M^{(i)})}
         {\displaystyle\max_{\sigma_i \in \mathcal{R}(\mathcal{H},U\mathbbm{A}_{\mathrm{st}})} 
         p_{\mathrm{succ}}(\sigma_i,\mathcal{T}^{(i)},M^{(i)})} 
    \\&= 1 + \mathsf{Im}_{R_1}(\varrho).
\end{align*}

The interpretation of the above is then as the one provided by Ref.~\cite{designolle2021set} for the case of set coherence: regardless of one's choice of reference basis $\mathbbm{A}$ there exists a collection of sub-channel discrimination tasks  in which the multi-state having non-zero imaginarity $\varrho$ outperforms any imaginarity-free multi-state. 

\subsection{On the relevance of multi-state coherence and imaginarity to multi-photon indistinguishability}

As shown by Refs.~\cite{shchesnovich2015partial,shchesnovich2018collective}, the output statistics of photon‐resolving measurements in linear‐optical interferometers depend \emph{solely} on the Bargmann invariants of the single‐photon internal states and on the unitary modeling the interferometer. Numerous experiments have since confirmed this prediction~\cite{pont2022quantifying,menssen2017distinguishability,jones2020multiparticle,jones2023distinguishability}.  

In practice, calibrating a linear photonic device often includes characterizing how indistinguishable the incoming photons are---a key figure of merit for, e.g., photon bunching and Boson Sampling validation~\cite{seron2023boson}. When photons are not perfectly identical, pairwise overlaps capture only part of the story: higher‐order Bargmann invariants can provide a more complete, and sometimes counter-intuitive~\cite{rodari2024experimentalobservationcounterintuitivefeatures}, picture.

Indeed, Ref.~\cite{jones2023distinguishability} recently showed that three‐photon indistinguishability---when the only distinguishing degree of freedom is, say, polarization---requires a three‐state Bargmann invariant.  They contrast an ``incoherent partial distinguishability'' regime, in which all three polarization Bloch vectors lie on a line (our notion of multi‐state incoherence), with a ``coherent partial distinguishability'' regime, where the three vectors span a nonzero volume in Bloch space (our notion of multi‐state imaginarity). By performing a tritter experiment, they directly accessed the imaginary part of the third‐order Bargmann invariant---which is proportional to that volume---and thereby distinguished these two regimes even when the photons themselves were mixed.  

Our analysis shows that, in two‐dimensional Hilbert spaces, one can in principle infer this same volume from the determinant of the Gram matrix \(G_{\mathbf r_\varrho}\), i.e.\ from pairwise overlaps \emph{alone}.  

However, the output statistics of such experiments is better described as agnostic to the underlying Hilbert space dimensionality such that only unitary-invariant information truly matters: dimension upper-bounds and purity assumptions rarely hold exactly. Moreover, our counterexamples demonstrate that a single third‐order invariant can fail to witness multi-state imaginarity without these assumptions. Accordingly, the most robust experimental strategy is to measure the higher‐order Bargmann invariants directly---just as in Ref.~\cite{jones2023distinguishability}---rather than rely solely on two‐state overlaps or purity‐and‐dimension assumptions.  This gives a more fine‐grained calibration of photon indistinguishability in realistic, imperfect sources. In this case, one can then use our witnesses based on equality constraints on Bargmann invariants of the form of Eq.~\eqref{eq:equality_constraints_Bargmanns}.  

\section{Discussion and outlook}\label{sec:discussion}

We have provided a complete characterization of multi-state imaginarity and coherence in qubit systems using Bargmann invariants and Gram matrix methods. We show that a single-qubit multi-state has imaginarity if and only if the Gram matrix of its Bloch vectors has rank three, while it exhibits coherence if and only if this Gram matrix has rank at least two. Geometrically, this reveals that imaginarity-free states lie within a common plane of the Bloch sphere, while incoherent states are constrained to a single axis. 

Beyond this geometric characterization of multi-state coherence and imaginarity, we have also showed that (for single-qubit systems) \emph{all} higher-order Bargmann invariants are completely determined by two-state overlaps \emph{up to complex conjugation}, with the sign of the imaginary part being the only independent unitary-invariant information not encoded in pairwise overlaps. In particular, this also provides a characterization of the form of Bargmann invariants and on the information content provided by two-state overlaps of single-qubit states. 

We have furthermore extended some of our results to provide necessary conditions---which can be viewed as coherence witnesses---for arbitrary Hilbert space dimensions using the values of higher-order Bargmann invariants, generalizing one of the insights from Ref.~\cite{fernandes2024unitary} for multi-state imaginarity to the more general case of multi-state coherence.   

Future work could approach our results from a formal resource-theoretic perspective.~In particular, it would be interesting to analyze notions of free transformations associated with these two multi-state resources, as well as the induced pre-order governing resource conversion. Such an analysis could shed new light on potential applications of the physical situations discussed here, as well as on the connection initially drawn in Ref.~\cite{miyazaki2022imaginarityfree} between multi-state imaginarity and multi-parameter state estimation in quantum metrology. Furthermore, it would be worthwhile to explore how our techniques extend to the case of nonstabilizerness~\cite{wagner2025unitary},  which represents a distinguished form of coherence relevant to magic-state injection schemes in fault-tolerant quantum computation.
	
\begin{acknowledgements}
RW would like to thank Filipa C. R. Peres, Ernesto Galv\~{a}o, and Felix Ahnefeld for discussions on coherence and imaginarity, and David Schmid and Yujie Zhang for discussions on multi-states. M-SL acknowledges  support from  National Natural Science Foundation of China  under Grant No. 12371458, the Guangdong Basic and Applied Basic Research Foundation under Grants No. 2023A1515012074   and No. 2024A1515010380. RW acknowledges support from the European Research Council (ERC) under the European Union’s Horizon 2020 research and innovation programme (grant agreement No.856432, HyperQ). LZ  is supported by Zhejiang Provincial Natural Science Foundation of China under Grants No. LZ23A010005
\end{acknowledgements}

\bibliography{bibliography}

\begin{thebibliography}{122}%
\makeatletter
\providecommand \@ifxundefined [1]{%
 \@ifx{#1\undefined}
}%
\providecommand \@ifnum [1]{%
 \ifnum #1\expandafter \@firstoftwo
 \else \expandafter \@secondoftwo
 \fi
}%
\providecommand \@ifx [1]{%
 \ifx #1\expandafter \@firstoftwo
 \else \expandafter \@secondoftwo
 \fi
}%
\providecommand \natexlab [1]{#1}%
\providecommand \enquote  [1]{``#1''}%
\providecommand \bibnamefont  [1]{#1}%
\providecommand \bibfnamefont [1]{#1}%
\providecommand \citenamefont [1]{#1}%
\providecommand \href@noop [0]{\@secondoftwo}%
\providecommand \href [0]{\begingroup \@sanitize@url \@href}%
\providecommand \@href[1]{\@@startlink{#1}\@@href}%
\providecommand \@@href[1]{\endgroup#1\@@endlink}%
\providecommand \@sanitize@url [0]{\catcode `\\12\catcode `\$12\catcode `\&12\catcode `\#12\catcode `\^12\catcode `\_12\catcode `\%12\relax}%
\providecommand \@@startlink[1]{}%
\providecommand \@@endlink[0]{}%
\providecommand \url  [0]{\begingroup\@sanitize@url \@url }%
\providecommand \@url [1]{\endgroup\@href {#1}{\urlprefix }}%
\providecommand \urlprefix  [0]{URL }%
\providecommand \Eprint [0]{\href }%
\providecommand \doibase [0]{https://doi.org/}%
\providecommand \selectlanguage [0]{\@gobble}%
\providecommand \bibinfo  [0]{\@secondoftwo}%
\providecommand \bibfield  [0]{\@secondoftwo}%
\providecommand \translation [1]{[#1]}%
\providecommand \BibitemOpen [0]{}%
\providecommand \bibitemStop [0]{}%
\providecommand \bibitemNoStop [0]{.\EOS\space}%
\providecommand \EOS [0]{\spacefactor3000\relax}%
\providecommand \BibitemShut  [1]{\csname bibitem#1\endcsname}%
\let\auto@bib@innerbib\@empty
\bibitem [{\citenamefont {Chitambar}\ and\ \citenamefont {Gour}(2019)}]{chitambar2019quantum}%
  \BibitemOpen
  \bibfield  {author} {\bibinfo {author} {\bibfnamefont {E.}~\bibnamefont {Chitambar}}\ and\ \bibinfo {author} {\bibfnamefont {G.}~\bibnamefont {Gour}},\ }\bibfield  {title} {\bibinfo {title} {Quantum resource theories},\ }\href {https://doi.org/10.1103/revmodphys.91.025001} {\bibfield  {journal} {\bibinfo  {journal} {Reviews of Modern Physics}\ }\textbf {\bibinfo {volume} {91}},\ \bibinfo {pages} {025001} (\bibinfo {year} {2019})}\BibitemShut {NoStop}%
\bibitem [{\citenamefont {Gour}(2024)}]{gour2024resourcesquantumworld}%
  \BibitemOpen
  \bibfield  {author} {\bibinfo {author} {\bibfnamefont {G.}~\bibnamefont {Gour}},\ }\href {https://arxiv.org/abs/2402.05474} {\bibinfo {title} {Resources of the quantum world}} (\bibinfo {year} {2024}),\ \Eprint {https://arxiv.org/abs/2402.05474} {arXiv:2402.05474 [quant-ph]} \BibitemShut {NoStop}%
\bibitem [{\citenamefont {Coecke}\ \emph {et~al.}(2016)\citenamefont {Coecke}, \citenamefont {Fritz},\ and\ \citenamefont {Spekkens}}]{coecke2016mathematical}%
  \BibitemOpen
  \bibfield  {author} {\bibinfo {author} {\bibfnamefont {B.}~\bibnamefont {Coecke}}, \bibinfo {author} {\bibfnamefont {T.}~\bibnamefont {Fritz}},\ and\ \bibinfo {author} {\bibfnamefont {R.~W.}\ \bibnamefont {Spekkens}},\ }\bibfield  {title} {\bibinfo {title} {{A mathematical theory of resources}},\ }\href {https://doi.org/10.1016/j.ic.2016.02.008} {\bibfield  {journal} {\bibinfo  {journal} {Information and Computation}\ }\textbf {\bibinfo {volume} {250}},\ \bibinfo {pages} {59–86} (\bibinfo {year} {2016})}\BibitemShut {NoStop}%
\bibitem [{\citenamefont {Shchesnovich}(2015)}]{shchesnovich2015partial}%
  \BibitemOpen
  \bibfield  {author} {\bibinfo {author} {\bibfnamefont {V.~S.}\ \bibnamefont {Shchesnovich}},\ }\bibfield  {title} {\bibinfo {title} {Partial indistinguishability theory for multiphoton experiments in multiport devices},\ }\href {https://doi.org/10.1103/PhysRevA.91.013844} {\bibfield  {journal} {\bibinfo  {journal} {Phys. Rev. A}\ }\textbf {\bibinfo {volume} {91}},\ \bibinfo {pages} {013844} (\bibinfo {year} {2015})}\BibitemShut {NoStop}%
\bibitem [{\citenamefont {Shchesnovich}\ and\ \citenamefont {Bezerra}(2018)}]{shchesnovich2018collective}%
  \BibitemOpen
  \bibfield  {author} {\bibinfo {author} {\bibfnamefont {V.~S.}\ \bibnamefont {Shchesnovich}}\ and\ \bibinfo {author} {\bibfnamefont {M.~E.~O.}\ \bibnamefont {Bezerra}},\ }\bibfield  {title} {\bibinfo {title} {Collective phases of identical particles interfering on linear multiports},\ }\href {https://doi.org/10.1103/PhysRevA.98.033805} {\bibfield  {journal} {\bibinfo  {journal} {Phys. Rev. A}\ }\textbf {\bibinfo {volume} {98}},\ \bibinfo {pages} {033805} (\bibinfo {year} {2018})}\BibitemShut {NoStop}%
\bibitem [{\citenamefont {Jones}\ \emph {et~al.}(2020)\citenamefont {Jones}, \citenamefont {Menssen}, \citenamefont {Chrzanowski}, \citenamefont {Wolterink}, \citenamefont {Shchesnovich},\ and\ \citenamefont {Walmsley}}]{jones2020multiparticle}%
  \BibitemOpen
  \bibfield  {author} {\bibinfo {author} {\bibfnamefont {A.~E.}\ \bibnamefont {Jones}}, \bibinfo {author} {\bibfnamefont {A.~J.}\ \bibnamefont {Menssen}}, \bibinfo {author} {\bibfnamefont {H.~M.}\ \bibnamefont {Chrzanowski}}, \bibinfo {author} {\bibfnamefont {T.~A.~W.}\ \bibnamefont {Wolterink}}, \bibinfo {author} {\bibfnamefont {V.~S.}\ \bibnamefont {Shchesnovich}},\ and\ \bibinfo {author} {\bibfnamefont {I.~A.}\ \bibnamefont {Walmsley}},\ }\bibfield  {title} {\bibinfo {title} {{Multiparticle Interference of Pairwise Distinguishable Photons}},\ }\href {https://doi.org/10.1103/PhysRevLett.125.123603} {\bibfield  {journal} {\bibinfo  {journal} {Phys. Rev. Lett.}\ }\textbf {\bibinfo {volume} {125}},\ \bibinfo {pages} {123603} (\bibinfo {year} {2020})}\BibitemShut {NoStop}%
\bibitem [{\citenamefont {Jones}\ \emph {et~al.}(2023)\citenamefont {Jones}, \citenamefont {Kumar}, \citenamefont {D'Aurelio}, \citenamefont {Bayerbach}, \citenamefont {Menssen},\ and\ \citenamefont {Barz}}]{jones2023distinguishability}%
  \BibitemOpen
  \bibfield  {author} {\bibinfo {author} {\bibfnamefont {A.~E.}\ \bibnamefont {Jones}}, \bibinfo {author} {\bibfnamefont {S.}~\bibnamefont {Kumar}}, \bibinfo {author} {\bibfnamefont {S.}~\bibnamefont {D'Aurelio}}, \bibinfo {author} {\bibfnamefont {M.}~\bibnamefont {Bayerbach}}, \bibinfo {author} {\bibfnamefont {A.~J.}\ \bibnamefont {Menssen}},\ and\ \bibinfo {author} {\bibfnamefont {S.}~\bibnamefont {Barz}},\ }\bibfield  {title} {\bibinfo {title} {Distinguishability and mixedness in quantum interference},\ }\href {https://doi.org/10.1103/PhysRevA.108.053701} {\bibfield  {journal} {\bibinfo  {journal} {Phys. Rev. A}\ }\textbf {\bibinfo {volume} {108}},\ \bibinfo {pages} {053701} (\bibinfo {year} {2023})}\BibitemShut {NoStop}%
\bibitem [{\citenamefont {Ji}\ \emph {et~al.}(2018)\citenamefont {Ji}, \citenamefont {Liu},\ and\ \citenamefont {Song}}]{ji2018pseudorandom}%
  \BibitemOpen
  \bibfield  {author} {\bibinfo {author} {\bibfnamefont {Z.}~\bibnamefont {Ji}}, \bibinfo {author} {\bibfnamefont {Y.-K.}\ \bibnamefont {Liu}},\ and\ \bibinfo {author} {\bibfnamefont {F.}~\bibnamefont {Song}},\ }\bibinfo {title} {Pseudorandom quantum states},\ in\ \href {https://doi.org/10.1007/978-3-319-96878-0_5} {\emph {\bibinfo {booktitle} {Advances in Cryptology – CRYPTO 2018}}}\ (\bibinfo  {publisher} {Springer International Publishing},\ \bibinfo {year} {2018})\ p.\ \bibinfo {pages} {126–152}\BibitemShut {NoStop}%
\bibitem [{\citenamefont {Kretschmer}(2021)}]{kretschmer2021quantumpseudo}%
  \BibitemOpen
  \bibfield  {author} {\bibinfo {author} {\bibfnamefont {W.}~\bibnamefont {Kretschmer}},\ }\bibfield  {title} {\bibinfo {title} {Quantum pseudorandomness and classical complexity},\ }in\ \href {https://doi.org/10.4230/LIPICS.TQC.2021.2} {\emph {\bibinfo {booktitle} {16th Conference on the Theory of Quantum Computation, Communication and Cryptography (TQC 2021)}}}\ (\bibinfo  {publisher} {Schloss Dagstuhl – Leibniz-Zentrum für Informatik},\ \bibinfo {year} {2021})\BibitemShut {NoStop}%
\bibitem [{\citenamefont {Bansal}\ \emph {et~al.}(2025)\citenamefont {Bansal}, \citenamefont {Mok}, \citenamefont {Bharti}, \citenamefont {Koh},\ and\ \citenamefont {Haug}}]{bansal2025pseudorandom}%
  \BibitemOpen
  \bibfield  {author} {\bibinfo {author} {\bibfnamefont {N.}~\bibnamefont {Bansal}}, \bibinfo {author} {\bibfnamefont {W.-K.}\ \bibnamefont {Mok}}, \bibinfo {author} {\bibfnamefont {K.}~\bibnamefont {Bharti}}, \bibinfo {author} {\bibfnamefont {D.~E.}\ \bibnamefont {Koh}},\ and\ \bibinfo {author} {\bibfnamefont {T.}~\bibnamefont {Haug}},\ }\bibfield  {title} {\bibinfo {title} {Pseudorandom density matrices},\ }\href {https://doi.org/10.1103/PRXQuantum.6.020322} {\bibfield  {journal} {\bibinfo  {journal} {PRX Quantum}\ }\textbf {\bibinfo {volume} {6}},\ \bibinfo {pages} {020322} (\bibinfo {year} {2025})}\BibitemShut {NoStop}%
\bibitem [{\citenamefont {Gu}\ \emph {et~al.}(2024)\citenamefont {Gu}, \citenamefont {Leone}, \citenamefont {Ghosh}, \citenamefont {Eisert}, \citenamefont {Yelin},\ and\ \citenamefont {Quek}}]{gu2024pseudomagic}%
  \BibitemOpen
  \bibfield  {author} {\bibinfo {author} {\bibfnamefont {A.}~\bibnamefont {Gu}}, \bibinfo {author} {\bibfnamefont {L.}~\bibnamefont {Leone}}, \bibinfo {author} {\bibfnamefont {S.}~\bibnamefont {Ghosh}}, \bibinfo {author} {\bibfnamefont {J.}~\bibnamefont {Eisert}}, \bibinfo {author} {\bibfnamefont {S.~F.}\ \bibnamefont {Yelin}},\ and\ \bibinfo {author} {\bibfnamefont {Y.}~\bibnamefont {Quek}},\ }\bibfield  {title} {\bibinfo {title} {Pseudomagic quantum states},\ }\href {https://doi.org/10.1103/PhysRevLett.132.210602} {\bibfield  {journal} {\bibinfo  {journal} {Phys. Rev. Lett.}\ }\textbf {\bibinfo {volume} {132}},\ \bibinfo {pages} {210602} (\bibinfo {year} {2024})}\BibitemShut {NoStop}%
\bibitem [{\citenamefont {Haug}\ \emph {et~al.}(2025)\citenamefont {Haug}, \citenamefont {Bharti},\ and\ \citenamefont {Koh}}]{haug2025pseudorandom}%
  \BibitemOpen
  \bibfield  {author} {\bibinfo {author} {\bibfnamefont {T.}~\bibnamefont {Haug}}, \bibinfo {author} {\bibfnamefont {K.}~\bibnamefont {Bharti}},\ and\ \bibinfo {author} {\bibfnamefont {D.~E.}\ \bibnamefont {Koh}},\ }\bibfield  {title} {\bibinfo {title} {Pseudorandom unitaries are neither real nor sparse nor noise-robust},\ }\href {https://doi.org/10.22331/q-2025-06-04-1759} {\bibfield  {journal} {\bibinfo  {journal} {{Quantum}}\ }\textbf {\bibinfo {volume} {9}},\ \bibinfo {pages} {1759} (\bibinfo {year} {2025})}\BibitemShut {NoStop}%
\bibitem [{\citenamefont {Tanggara}\ \emph {et~al.}(2025)\citenamefont {Tanggara}, \citenamefont {Gu},\ and\ \citenamefont {Bharti}}]{tanggara2025neartermpseudorandompseudoresourcequantum}%
  \BibitemOpen
  \bibfield  {author} {\bibinfo {author} {\bibfnamefont {A.}~\bibnamefont {Tanggara}}, \bibinfo {author} {\bibfnamefont {M.}~\bibnamefont {Gu}},\ and\ \bibinfo {author} {\bibfnamefont {K.}~\bibnamefont {Bharti}},\ }\href {https://arxiv.org/abs/2504.17650} {\bibinfo {title} {Near-term pseudorandom and pseudoresource quantum states}} (\bibinfo {year} {2025}),\ \Eprint {https://arxiv.org/abs/2504.17650} {arXiv:2504.17650 [quant-ph]} \BibitemShut {NoStop}%
\bibitem [{\citenamefont {Aaronson}\ \emph {et~al.}(2023)\citenamefont {Aaronson}, \citenamefont {Bouland}, \citenamefont {Fefferman}, \citenamefont {Ghosh}, \citenamefont {Vazirani}, \citenamefont {Zhang},\ and\ \citenamefont {Zhou}}]{aaronson2023quantumpseudoentanglement}%
  \BibitemOpen
  \bibfield  {author} {\bibinfo {author} {\bibfnamefont {S.}~\bibnamefont {Aaronson}}, \bibinfo {author} {\bibfnamefont {A.}~\bibnamefont {Bouland}}, \bibinfo {author} {\bibfnamefont {B.}~\bibnamefont {Fefferman}}, \bibinfo {author} {\bibfnamefont {S.}~\bibnamefont {Ghosh}}, \bibinfo {author} {\bibfnamefont {U.}~\bibnamefont {Vazirani}}, \bibinfo {author} {\bibfnamefont {C.}~\bibnamefont {Zhang}},\ and\ \bibinfo {author} {\bibfnamefont {Z.}~\bibnamefont {Zhou}},\ }\href {https://arxiv.org/abs/2211.00747} {\bibinfo {title} {Quantum pseudoentanglement}} (\bibinfo {year} {2023}),\ \Eprint {https://arxiv.org/abs/2211.00747} {arXiv:2211.00747 [quant-ph]} \BibitemShut {NoStop}%
\bibitem [{\citenamefont {Rivas}\ \emph {et~al.}(2014)\citenamefont {Rivas}, \citenamefont {Huelga},\ and\ \citenamefont {Plenio}}]{rivas2014quantum}%
  \BibitemOpen
  \bibfield  {author} {\bibinfo {author} {\bibfnamefont {A.}~\bibnamefont {Rivas}}, \bibinfo {author} {\bibfnamefont {S.~F.}\ \bibnamefont {Huelga}},\ and\ \bibinfo {author} {\bibfnamefont {M.~B.}\ \bibnamefont {Plenio}},\ }\bibfield  {title} {\bibinfo {title} {Quantum non-markovianity: characterization, quantification and detection},\ }\href {https://doi.org/10.1088/0034-4885/77/9/094001} {\bibfield  {journal} {\bibinfo  {journal} {Reports on Progress in Physics}\ }\textbf {\bibinfo {volume} {77}},\ \bibinfo {pages} {094001} (\bibinfo {year} {2014})}\BibitemShut {NoStop}%
\bibitem [{\citenamefont {Designolle}\ \emph {et~al.}(2021)\citenamefont {Designolle}, \citenamefont {Uola}, \citenamefont {Luoma},\ and\ \citenamefont {Brunner}}]{designolle2021set}%
  \BibitemOpen
  \bibfield  {author} {\bibinfo {author} {\bibfnamefont {S.}~\bibnamefont {Designolle}}, \bibinfo {author} {\bibfnamefont {R.}~\bibnamefont {Uola}}, \bibinfo {author} {\bibfnamefont {K.}~\bibnamefont {Luoma}},\ and\ \bibinfo {author} {\bibfnamefont {N.}~\bibnamefont {Brunner}},\ }\bibfield  {title} {\bibinfo {title} {{Set Coherence: Basis-Independent Quantification of Quantum Coherence}},\ }\href {https://doi.org/10.1103/PhysRevLett.126.220404} {\bibfield  {journal} {\bibinfo  {journal} {Phys. Rev. Lett.}\ }\textbf {\bibinfo {volume} {126}},\ \bibinfo {pages} {220404} (\bibinfo {year} {2021})}\BibitemShut {NoStop}%
\bibitem [{\citenamefont {Miyazaki}\ and\ \citenamefont {Matsumoto}(2022)}]{miyazaki2022imaginarityfree}%
  \BibitemOpen
  \bibfield  {author} {\bibinfo {author} {\bibfnamefont {J.}~\bibnamefont {Miyazaki}}\ and\ \bibinfo {author} {\bibfnamefont {K.}~\bibnamefont {Matsumoto}},\ }\bibfield  {title} {\bibinfo {title} {Imaginarity-free quantum multiparameter estimation},\ }\href {https://doi.org/10.22331/q-2022-03-10-665} {\bibfield  {journal} {\bibinfo  {journal} {{Quantum}}\ }\textbf {\bibinfo {volume} {6}},\ \bibinfo {pages} {665} (\bibinfo {year} {2022})}\BibitemShut {NoStop}%
\bibitem [{\citenamefont {Salazar}\ \emph {et~al.}(2022)\citenamefont {Salazar}, \citenamefont {Czartowski},\ and\ \citenamefont {de~Oliveira~Junior}}]{salazar2022resource}%
  \BibitemOpen
  \bibfield  {author} {\bibinfo {author} {\bibfnamefont {R.}~\bibnamefont {Salazar}}, \bibinfo {author} {\bibfnamefont {J.}~\bibnamefont {Czartowski}},\ and\ \bibinfo {author} {\bibfnamefont {A.}~\bibnamefont {de~Oliveira~Junior}},\ }\href@noop {} {\bibinfo {title} {{Resource theory of Absolute Negativity}}} (\bibinfo {year} {2022}),\ \Eprint {https://arxiv.org/abs/2205.13480} {arXiv:2205.13480 [quant-ph]} \BibitemShut {NoStop}%
\bibitem [{\citenamefont {Buscemi}\ \emph {et~al.}(2020)\citenamefont {Buscemi}, \citenamefont {Chitambar},\ and\ \citenamefont {Zhou}}]{buscemi2020complete}%
  \BibitemOpen
  \bibfield  {author} {\bibinfo {author} {\bibfnamefont {F.}~\bibnamefont {Buscemi}}, \bibinfo {author} {\bibfnamefont {E.}~\bibnamefont {Chitambar}},\ and\ \bibinfo {author} {\bibfnamefont {W.}~\bibnamefont {Zhou}},\ }\bibfield  {title} {\bibinfo {title} {{Complete Resource Theory of Quantum Incompatibility as Quantum Programmability}},\ }\href {https://doi.org/10.1103/PhysRevLett.124.120401} {\bibfield  {journal} {\bibinfo  {journal} {Phys. Rev. Lett.}\ }\textbf {\bibinfo {volume} {124}},\ \bibinfo {pages} {120401} (\bibinfo {year} {2020})}\BibitemShut {NoStop}%
\bibitem [{\citenamefont {Uola}\ \emph {et~al.}(2019)\citenamefont {Uola}, \citenamefont {Kraft}, \citenamefont {Shang}, \citenamefont {Yu},\ and\ \citenamefont {G\"uhne}}]{uola2019quantifying}%
  \BibitemOpen
  \bibfield  {author} {\bibinfo {author} {\bibfnamefont {R.}~\bibnamefont {Uola}}, \bibinfo {author} {\bibfnamefont {T.}~\bibnamefont {Kraft}}, \bibinfo {author} {\bibfnamefont {J.}~\bibnamefont {Shang}}, \bibinfo {author} {\bibfnamefont {X.-D.}\ \bibnamefont {Yu}},\ and\ \bibinfo {author} {\bibfnamefont {O.}~\bibnamefont {G\"uhne}},\ }\bibfield  {title} {\bibinfo {title} {{Quantifying Quantum Resources with Conic Programming}},\ }\href {https://doi.org/10.1103/PhysRevLett.122.130404} {\bibfield  {journal} {\bibinfo  {journal} {Phys. Rev. Lett.}\ }\textbf {\bibinfo {volume} {122}},\ \bibinfo {pages} {130404} (\bibinfo {year} {2019})}\BibitemShut {NoStop}%
\bibitem [{\citenamefont {Martins}\ \emph {et~al.}(2020)\citenamefont {Martins}, \citenamefont {Savi},\ and\ \citenamefont {Angelo}}]{martins2020quantum}%
  \BibitemOpen
  \bibfield  {author} {\bibinfo {author} {\bibfnamefont {E.}~\bibnamefont {Martins}}, \bibinfo {author} {\bibfnamefont {M.~F.}\ \bibnamefont {Savi}},\ and\ \bibinfo {author} {\bibfnamefont {R.~M.}\ \bibnamefont {Angelo}},\ }\bibfield  {title} {\bibinfo {title} {Quantum incompatibility of a physical context},\ }\href {https://doi.org/10.1103/PhysRevA.102.050201} {\bibfield  {journal} {\bibinfo  {journal} {Phys. Rev. A}\ }\textbf {\bibinfo {volume} {102}},\ \bibinfo {pages} {050201} (\bibinfo {year} {2020})}\BibitemShut {NoStop}%
\bibitem [{\citenamefont {Ducuara}\ \emph {et~al.}(2020)\citenamefont {Ducuara}, \citenamefont {Lipka-Bartosik},\ and\ \citenamefont {Skrzypczyk}}]{ducuara2020multiobject}%
  \BibitemOpen
  \bibfield  {author} {\bibinfo {author} {\bibfnamefont {A.~F.}\ \bibnamefont {Ducuara}}, \bibinfo {author} {\bibfnamefont {P.}~\bibnamefont {Lipka-Bartosik}},\ and\ \bibinfo {author} {\bibfnamefont {P.}~\bibnamefont {Skrzypczyk}},\ }\bibfield  {title} {\bibinfo {title} {Multiobject operational tasks for convex quantum resource theories of state-measurement pairs},\ }\href {https://doi.org/10.1103/PhysRevResearch.2.033374} {\bibfield  {journal} {\bibinfo  {journal} {Phys. Rev. Res.}\ }\textbf {\bibinfo {volume} {2}},\ \bibinfo {pages} {033374} (\bibinfo {year} {2020})}\BibitemShut {NoStop}%
\bibitem [{\citenamefont {Selby}\ \emph {et~al.}(2023{\natexlab{a}})\citenamefont {Selby}, \citenamefont {Schmid}, \citenamefont {Wolfe}, \citenamefont {Sainz}, \citenamefont {Kunjwal},\ and\ \citenamefont {Spekkens}}]{selby2023contextuality}%
  \BibitemOpen
  \bibfield  {author} {\bibinfo {author} {\bibfnamefont {J.~H.}\ \bibnamefont {Selby}}, \bibinfo {author} {\bibfnamefont {D.}~\bibnamefont {Schmid}}, \bibinfo {author} {\bibfnamefont {E.}~\bibnamefont {Wolfe}}, \bibinfo {author} {\bibfnamefont {A.~B.}\ \bibnamefont {Sainz}}, \bibinfo {author} {\bibfnamefont {R.}~\bibnamefont {Kunjwal}},\ and\ \bibinfo {author} {\bibfnamefont {R.~W.}\ \bibnamefont {Spekkens}},\ }\bibfield  {title} {\bibinfo {title} {{Contextuality without Incompatibility}},\ }\href {https://doi.org/10.1103/PhysRevLett.130.230201} {\bibfield  {journal} {\bibinfo  {journal} {Phys. Rev. Lett.}\ }\textbf {\bibinfo {volume} {130}},\ \bibinfo {pages} {230201} (\bibinfo {year} {2023}{\natexlab{a}})}\BibitemShut {NoStop}%
\bibitem [{\citenamefont {Wagner}\ \emph {et~al.}(2024{\natexlab{a}})\citenamefont {Wagner}, \citenamefont {Barbosa},\ and\ \citenamefont {Galv\~ao}}]{wagner2024inequalities}%
  \BibitemOpen
  \bibfield  {author} {\bibinfo {author} {\bibfnamefont {R.}~\bibnamefont {Wagner}}, \bibinfo {author} {\bibfnamefont {R.~S.}\ \bibnamefont {Barbosa}},\ and\ \bibinfo {author} {\bibfnamefont {E.~F.}\ \bibnamefont {Galv\~ao}},\ }\bibfield  {title} {\bibinfo {title} {Inequalities witnessing coherence, nonlocality, and contextuality},\ }\href {https://doi.org/10.1103/PhysRevA.109.032220} {\bibfield  {journal} {\bibinfo  {journal} {Phys. Rev. A}\ }\textbf {\bibinfo {volume} {109}},\ \bibinfo {pages} {032220} (\bibinfo {year} {2024}{\natexlab{a}})}\BibitemShut {NoStop}%
\bibitem [{\citenamefont {Galv\~{a}o}\ and\ \citenamefont {Brod}(2020)}]{galvao2020quantum}%
  \BibitemOpen
  \bibfield  {author} {\bibinfo {author} {\bibfnamefont {E.~F.}\ \bibnamefont {Galv\~{a}o}}\ and\ \bibinfo {author} {\bibfnamefont {D.~J.}\ \bibnamefont {Brod}},\ }\bibfield  {title} {\bibinfo {title} {Quantum and classical bounds for two-state overlaps},\ }\href {https://doi.org/10.1103/PhysRevA.101.062110} {\bibfield  {journal} {\bibinfo  {journal} {Phys. Rev. A}\ }\textbf {\bibinfo {volume} {101}},\ \bibinfo {pages} {062110} (\bibinfo {year} {2020})}\BibitemShut {NoStop}%
\bibitem [{\citenamefont {Zhang}\ \emph {et~al.}(2025{\natexlab{a}})\citenamefont {Zhang}, \citenamefont {Yīng},\ and\ \citenamefont {Schmid}}]{zhang2025quantifierswitnessesnonclassicalitymeasurements}%
  \BibitemOpen
  \bibfield  {author} {\bibinfo {author} {\bibfnamefont {Y.}~\bibnamefont {Zhang}}, \bibinfo {author} {\bibfnamefont {Y.}~\bibnamefont {Yīng}},\ and\ \bibinfo {author} {\bibfnamefont {D.}~\bibnamefont {Schmid}},\ }\href {https://arxiv.org/abs/2504.02944} {\bibinfo {title} {Quantifiers and witnesses for the nonclassicality of measurements and of states}} (\bibinfo {year} {2025}{\natexlab{a}}),\ \Eprint {https://arxiv.org/abs/2504.02944} {arXiv:2504.02944 [quant-ph]} \BibitemShut {NoStop}%
\bibitem [{\citenamefont {Zhang}\ \emph {et~al.}(2025{\natexlab{b}})\citenamefont {Zhang}, \citenamefont {Schmid}, \citenamefont {Yīng},\ and\ \citenamefont {Spekkens}}]{zhang2025reassessingboundaryclassicalnonclassical}%
  \BibitemOpen
  \bibfield  {author} {\bibinfo {author} {\bibfnamefont {Y.}~\bibnamefont {Zhang}}, \bibinfo {author} {\bibfnamefont {D.}~\bibnamefont {Schmid}}, \bibinfo {author} {\bibfnamefont {Y.}~\bibnamefont {Yīng}},\ and\ \bibinfo {author} {\bibfnamefont {R.~W.}\ \bibnamefont {Spekkens}},\ }\href {https://arxiv.org/abs/2503.05884} {\bibinfo {title} {Reassessing the boundary between classical and nonclassical for individual quantum processes}} (\bibinfo {year} {2025}{\natexlab{b}}),\ \Eprint {https://arxiv.org/abs/2503.05884} {arXiv:2503.05884 [quant-ph]} \BibitemShut {NoStop}%
\bibitem [{\citenamefont {Gour}\ \emph {et~al.}(2018)\citenamefont {Gour}, \citenamefont {Heinosaari},\ and\ \citenamefont {Spekkens}}]{gour2018Incompatibility}%
  \BibitemOpen
  \bibfield  {author} {\bibinfo {author} {\bibfnamefont {G.}~\bibnamefont {Gour}}, \bibinfo {author} {\bibfnamefont {T.}~\bibnamefont {Heinosaari}},\ and\ \bibinfo {author} {\bibfnamefont {R.}~\bibnamefont {Spekkens}},\ }\href {http://meetings.aps.org/link/BAPS.2018.MAR.S26.4} {\bibinfo {title} {The resource theory of incompatibility}},\ \bibinfo {howpublished} {Presented at the APS March Meeting 2018, Volume 63, Number 1, Session S26: Quantum Resource Theories I} (\bibinfo {year} {2018})\BibitemShut {NoStop}%
\bibitem [{\citenamefont {Horodecki}\ \emph {et~al.}(2007)\citenamefont {Horodecki}, \citenamefont {Sen(De)},\ and\ \citenamefont {Sen}}]{horodecki2007quantification}%
  \BibitemOpen
  \bibfield  {author} {\bibinfo {author} {\bibfnamefont {M.}~\bibnamefont {Horodecki}}, \bibinfo {author} {\bibfnamefont {A.}~\bibnamefont {Sen(De)}},\ and\ \bibinfo {author} {\bibfnamefont {U.}~\bibnamefont {Sen}},\ }\bibfield  {title} {\bibinfo {title} {Quantification of quantum correlation of ensembles of states},\ }\href {https://doi.org/10.1103/PhysRevA.75.062329} {\bibfield  {journal} {\bibinfo  {journal} {Phys. Rev. A}\ }\textbf {\bibinfo {volume} {75}},\ \bibinfo {pages} {062329} (\bibinfo {year} {2007})}\BibitemShut {NoStop}%
\bibitem [{\citenamefont {Piani}\ \emph {et~al.}(2014)\citenamefont {Piani}, \citenamefont {Narasimhachar},\ and\ \citenamefont {Calsamiglia}}]{piani2014quantumness}%
  \BibitemOpen
  \bibfield  {author} {\bibinfo {author} {\bibfnamefont {M.}~\bibnamefont {Piani}}, \bibinfo {author} {\bibfnamefont {V.}~\bibnamefont {Narasimhachar}},\ and\ \bibinfo {author} {\bibfnamefont {J.}~\bibnamefont {Calsamiglia}},\ }\bibfield  {title} {\bibinfo {title} {Quantumness of correlations, quantumness of ensembles and quantum data hiding},\ }\href {https://doi.org/10.1088/1367-2630/16/11/113001} {\bibfield  {journal} {\bibinfo  {journal} {New Journal of Physics}\ }\textbf {\bibinfo {volume} {16}},\ \bibinfo {pages} {113001} (\bibinfo {year} {2014})}\BibitemShut {NoStop}%
\bibitem [{\citenamefont {Fuchs}\ and\ \citenamefont {Sasaki}(2003)}]{fuchs2003squeezing}%
  \BibitemOpen
  \bibfield  {author} {\bibinfo {author} {\bibfnamefont {C.~A.}\ \bibnamefont {Fuchs}}\ and\ \bibinfo {author} {\bibfnamefont {M.}~\bibnamefont {Sasaki}},\ }\bibfield  {title} {\bibinfo {title} {Squeezing quantum information through a classical channel: measuring the ``quantumness'' of a set of quantum states},\ }\href {https://doi.org/10.5555/2011544.2011545} {\bibfield  {journal} {\bibinfo  {journal} {Quantum Info. Comput.}\ }\textbf {\bibinfo {volume} {3}},\ \bibinfo {pages} {377–404} (\bibinfo {year} {2003})}\BibitemShut {NoStop}%
\bibitem [{\citenamefont {Horodecki}\ \emph {et~al.}(2006)\citenamefont {Horodecki}, \citenamefont {Horodecki}, \citenamefont {Horodecki},\ and\ \citenamefont {Piani}}]{horodecki2006quantumness}%
  \BibitemOpen
  \bibfield  {author} {\bibinfo {author} {\bibfnamefont {M.}~\bibnamefont {Horodecki}}, \bibinfo {author} {\bibfnamefont {P.}~\bibnamefont {Horodecki}}, \bibinfo {author} {\bibfnamefont {R.}~\bibnamefont {Horodecki}},\ and\ \bibinfo {author} {\bibfnamefont {M.}~\bibnamefont {Piani}},\ }\bibfield  {title} {\bibinfo {title} {{Quantumness of Ensembles from the No-Broadcasting Principle}},\ }\href {https://doi.org/10.1142/S0219749906001748} {\bibfield  {journal} {\bibinfo  {journal} {International Journal of Quantum Information}\ }\textbf {\bibinfo {volume} {4}},\ \bibinfo {pages} {105} (\bibinfo {year} {2006})}\BibitemShut {NoStop}%
\bibitem [{\citenamefont {Baumgratz}\ \emph {et~al.}(2014)\citenamefont {Baumgratz}, \citenamefont {Cramer},\ and\ \citenamefont {Plenio}}]{baumgratz2014quantifying}%
  \BibitemOpen
  \bibfield  {author} {\bibinfo {author} {\bibfnamefont {T.}~\bibnamefont {Baumgratz}}, \bibinfo {author} {\bibfnamefont {M.}~\bibnamefont {Cramer}},\ and\ \bibinfo {author} {\bibfnamefont {M.~B.}\ \bibnamefont {Plenio}},\ }\bibfield  {title} {\bibinfo {title} {{Quantifying Coherence}},\ }\href {https://doi.org/10.1103/PhysRevLett.113.140401} {\bibfield  {journal} {\bibinfo  {journal} {Phys. Rev. Lett.}\ }\textbf {\bibinfo {volume} {113}},\ \bibinfo {pages} {140401} (\bibinfo {year} {2014})}\BibitemShut {NoStop}%
\bibitem [{\citenamefont {Streltsov}\ \emph {et~al.}(2017)\citenamefont {Streltsov}, \citenamefont {Adesso},\ and\ \citenamefont {Plenio}}]{streltsov2017colloquium}%
  \BibitemOpen
  \bibfield  {author} {\bibinfo {author} {\bibfnamefont {A.}~\bibnamefont {Streltsov}}, \bibinfo {author} {\bibfnamefont {G.}~\bibnamefont {Adesso}},\ and\ \bibinfo {author} {\bibfnamefont {M.~B.}\ \bibnamefont {Plenio}},\ }\bibfield  {title} {\bibinfo {title} {Colloquium: Quantum coherence as a resource},\ }\href {https://doi.org/10.1103/RevModPhys.89.041003} {\bibfield  {journal} {\bibinfo  {journal} {Reviews of Modern Physics}\ }\textbf {\bibinfo {volume} {89}},\ \bibinfo {pages} {041003} (\bibinfo {year} {2017})}\BibitemShut {NoStop}%
\bibitem [{\citenamefont {Ahnefeld}\ \emph {et~al.}(2022)\citenamefont {Ahnefeld}, \citenamefont {Theurer}, \citenamefont {Egloff}, \citenamefont {Matera},\ and\ \citenamefont {Plenio}}]{ahnefeld2022coherence}%
  \BibitemOpen
  \bibfield  {author} {\bibinfo {author} {\bibfnamefont {F.}~\bibnamefont {Ahnefeld}}, \bibinfo {author} {\bibfnamefont {T.}~\bibnamefont {Theurer}}, \bibinfo {author} {\bibfnamefont {D.}~\bibnamefont {Egloff}}, \bibinfo {author} {\bibfnamefont {J.~M.}\ \bibnamefont {Matera}},\ and\ \bibinfo {author} {\bibfnamefont {M.~B.}\ \bibnamefont {Plenio}},\ }\bibfield  {title} {\bibinfo {title} {{Coherence as a Resource for Shor's Algorithm}},\ }\href {https://doi.org/10.1103/PhysRevLett.129.120501} {\bibfield  {journal} {\bibinfo  {journal} {Phys. Rev. Lett.}\ }\textbf {\bibinfo {volume} {129}},\ \bibinfo {pages} {120501} (\bibinfo {year} {2022})}\BibitemShut {NoStop}%
\bibitem [{\citenamefont {Naseri}\ \emph {et~al.}(2022)\citenamefont {Naseri}, \citenamefont {Kondra}, \citenamefont {Goswami}, \citenamefont {Fellous-Asiani},\ and\ \citenamefont {Streltsov}}]{naseri2022entanglement}%
  \BibitemOpen
  \bibfield  {author} {\bibinfo {author} {\bibfnamefont {M.}~\bibnamefont {Naseri}}, \bibinfo {author} {\bibfnamefont {T.~V.}\ \bibnamefont {Kondra}}, \bibinfo {author} {\bibfnamefont {S.}~\bibnamefont {Goswami}}, \bibinfo {author} {\bibfnamefont {M.}~\bibnamefont {Fellous-Asiani}},\ and\ \bibinfo {author} {\bibfnamefont {A.}~\bibnamefont {Streltsov}},\ }\bibfield  {title} {\bibinfo {title} {Entanglement and coherence in the bernstein-vazirani algorithm},\ }\href {https://doi.org/10.1103/PhysRevA.106.062429} {\bibfield  {journal} {\bibinfo  {journal} {Phys. Rev. A}\ }\textbf {\bibinfo {volume} {106}},\ \bibinfo {pages} {062429} (\bibinfo {year} {2022})}\BibitemShut {NoStop}%
\bibitem [{\citenamefont {Ahnefeld}\ \emph {et~al.}(2025)\citenamefont {Ahnefeld}, \citenamefont {Theurer},\ and\ \citenamefont {Plenio}}]{ahnefeld2025coherenceresourcephaseestimation}%
  \BibitemOpen
  \bibfield  {author} {\bibinfo {author} {\bibfnamefont {F.}~\bibnamefont {Ahnefeld}}, \bibinfo {author} {\bibfnamefont {T.}~\bibnamefont {Theurer}},\ and\ \bibinfo {author} {\bibfnamefont {M.~B.}\ \bibnamefont {Plenio}},\ }\href {https://arxiv.org/abs/2505.18544} {\bibinfo {title} {Coherence as a resource for phase estimation}} (\bibinfo {year} {2025}),\ \Eprint {https://arxiv.org/abs/2505.18544} {arXiv:2505.18544 [quant-ph]} \BibitemShut {NoStop}%
\bibitem [{\citenamefont {Gour}(2017)}]{gour2017quantum}%
  \BibitemOpen
  \bibfield  {author} {\bibinfo {author} {\bibfnamefont {G.}~\bibnamefont {Gour}},\ }\bibfield  {title} {\bibinfo {title} {Quantum resource theories in the single-shot regime},\ }\href {https://doi.org/10.1103/PhysRevA.95.062314} {\bibfield  {journal} {\bibinfo  {journal} {Phys. Rev. A}\ }\textbf {\bibinfo {volume} {95}},\ \bibinfo {pages} {062314} (\bibinfo {year} {2017})}\BibitemShut {NoStop}%
\bibitem [{\citenamefont {Hickey}\ and\ \citenamefont {Gour}(2018)}]{hickey2018quantifying}%
  \BibitemOpen
  \bibfield  {author} {\bibinfo {author} {\bibfnamefont {A.}~\bibnamefont {Hickey}}\ and\ \bibinfo {author} {\bibfnamefont {G.}~\bibnamefont {Gour}},\ }\bibfield  {title} {\bibinfo {title} {{Quantifying the Imaginarity of Quantum Mechanics}},\ }\href {https://doi.org/10.1088/1751-8121/aabe9c} {\bibfield  {journal} {\bibinfo  {journal} {Journal of Physics A: Mathematical and Theoretical}\ }\textbf {\bibinfo {volume} {51}},\ \bibinfo {pages} {414009} (\bibinfo {year} {2018})}\BibitemShut {NoStop}%
\bibitem [{\citenamefont {Wu}\ \emph {et~al.}(2021{\natexlab{a}})\citenamefont {Wu}, \citenamefont {Kondra}, \citenamefont {Rana}, \citenamefont {Scandolo}, \citenamefont {Xiang}, \citenamefont {Li}, \citenamefont {Guo},\ and\ \citenamefont {Streltsov}}]{wu2021operational}%
  \BibitemOpen
  \bibfield  {author} {\bibinfo {author} {\bibfnamefont {K.-D.}\ \bibnamefont {Wu}}, \bibinfo {author} {\bibfnamefont {T.~V.}\ \bibnamefont {Kondra}}, \bibinfo {author} {\bibfnamefont {S.}~\bibnamefont {Rana}}, \bibinfo {author} {\bibfnamefont {C.~M.}\ \bibnamefont {Scandolo}}, \bibinfo {author} {\bibfnamefont {G.-Y.}\ \bibnamefont {Xiang}}, \bibinfo {author} {\bibfnamefont {C.-F.}\ \bibnamefont {Li}}, \bibinfo {author} {\bibfnamefont {G.-C.}\ \bibnamefont {Guo}},\ and\ \bibinfo {author} {\bibfnamefont {A.}~\bibnamefont {Streltsov}},\ }\bibfield  {title} {\bibinfo {title} {Operational {R}esource {T}heory of {I}maginarity},\ }\href {https://doi.org/10.1103/physrevlett.126.090401} {\bibfield  {journal} {\bibinfo  {journal} {Phys. Rev. Lett.}\ }\textbf {\bibinfo {volume} {126}},\ \bibinfo {pages} {090401} (\bibinfo {year} {2021}{\natexlab{a}})}\BibitemShut {NoStop}%
\bibitem [{\citenamefont {Wu}\ \emph {et~al.}(2021{\natexlab{b}})\citenamefont {Wu}, \citenamefont {Kondra}, \citenamefont {Rana}, \citenamefont {Scandolo}, \citenamefont {Xiang}, \citenamefont {Li}, \citenamefont {Guo},\ and\ \citenamefont {Streltsov}}]{wu2021resource}%
  \BibitemOpen
  \bibfield  {author} {\bibinfo {author} {\bibfnamefont {K.-D.}\ \bibnamefont {Wu}}, \bibinfo {author} {\bibfnamefont {T.~V.}\ \bibnamefont {Kondra}}, \bibinfo {author} {\bibfnamefont {S.}~\bibnamefont {Rana}}, \bibinfo {author} {\bibfnamefont {C.~M.}\ \bibnamefont {Scandolo}}, \bibinfo {author} {\bibfnamefont {G.-Y.}\ \bibnamefont {Xiang}}, \bibinfo {author} {\bibfnamefont {C.-F.}\ \bibnamefont {Li}}, \bibinfo {author} {\bibfnamefont {G.-C.}\ \bibnamefont {Guo}},\ and\ \bibinfo {author} {\bibfnamefont {A.}~\bibnamefont {Streltsov}},\ }\bibfield  {title} {\bibinfo {title} {Resource theory of imaginarity: Quantification and state conversion},\ }\href {https://doi.org/10.1103/PhysRevA.103.032401} {\bibfield  {journal} {\bibinfo  {journal} {Phys. Rev. A}\ }\textbf {\bibinfo {volume} {103}},\ \bibinfo {pages} {032401} (\bibinfo {year} {2021}{\natexlab{b}})}\BibitemShut {NoStop}%
\bibitem [{\citenamefont {Wu}\ \emph {et~al.}(2024)\citenamefont {Wu}, \citenamefont {Kondra}, \citenamefont {Scandolo}, \citenamefont {Rana}, \citenamefont {Xiang}, \citenamefont {Li}, \citenamefont {Guo},\ and\ \citenamefont {Streltsov}}]{wu2023resource}%
  \BibitemOpen
  \bibfield  {author} {\bibinfo {author} {\bibfnamefont {K.-D.}\ \bibnamefont {Wu}}, \bibinfo {author} {\bibfnamefont {T.~V.}\ \bibnamefont {Kondra}}, \bibinfo {author} {\bibfnamefont {C.~M.}\ \bibnamefont {Scandolo}}, \bibinfo {author} {\bibfnamefont {S.}~\bibnamefont {Rana}}, \bibinfo {author} {\bibfnamefont {G.-Y.}\ \bibnamefont {Xiang}}, \bibinfo {author} {\bibfnamefont {C.-F.}\ \bibnamefont {Li}}, \bibinfo {author} {\bibfnamefont {G.-C.}\ \bibnamefont {Guo}},\ and\ \bibinfo {author} {\bibfnamefont {A.}~\bibnamefont {Streltsov}},\ }\bibfield  {title} {\bibinfo {title} {{Resource Theory of Imaginarity in Distributed Scenarios}},\ }\href {https://doi.org/10.1038/s42005-024-01649-y} {\bibfield  {journal} {\bibinfo  {journal} {Communications Physics}\ }\textbf {\bibinfo {volume} {7}},\ \bibinfo {pages} {171} (\bibinfo {year} {2024})}\BibitemShut {NoStop}%
\bibitem [{\citenamefont {Elliott}(2025)}]{elliott2025strictadvantagecomplexquantum}%
  \BibitemOpen
  \bibfield  {author} {\bibinfo {author} {\bibfnamefont {T.~J.}\ \bibnamefont {Elliott}},\ }\bibfield  {title} {\bibinfo {title} {Strict advantage of complex quantum theory in a communication task},\ }\href {https://doi.org/10.1103/PhysRevA.111.062401} {\bibfield  {journal} {\bibinfo  {journal} {Phys. Rev. A}\ }\textbf {\bibinfo {volume} {111}},\ \bibinfo {pages} {062401} (\bibinfo {year} {2025})}\BibitemShut {NoStop}%
\bibitem [{\citenamefont {Zhu}(2021)}]{zhu2021hidingmasking}%
  \BibitemOpen
  \bibfield  {author} {\bibinfo {author} {\bibfnamefont {H.}~\bibnamefont {Zhu}},\ }\bibfield  {title} {\bibinfo {title} {Hiding and masking quantum information in complex and real quantum mechanics},\ }\href {https://doi.org/10.1103/PhysRevResearch.3.033176} {\bibfield  {journal} {\bibinfo  {journal} {Phys. Rev. Res.}\ }\textbf {\bibinfo {volume} {3}},\ \bibinfo {pages} {033176} (\bibinfo {year} {2021})}\BibitemShut {NoStop}%
\bibitem [{\citenamefont {Kedem}(2012)}]{kedem2012usingtechnical}%
  \BibitemOpen
  \bibfield  {author} {\bibinfo {author} {\bibfnamefont {Y.}~\bibnamefont {Kedem}},\ }\bibfield  {title} {\bibinfo {title} {Using technical noise to increase the signal-to-noise ratio of measurements via imaginary weak values},\ }\href {https://doi.org/10.1103/physreva.85.060102} {\bibfield  {journal} {\bibinfo  {journal} {Phys. Rev. A}\ }\textbf {\bibinfo {volume} {85}},\ \bibinfo {pages} {060102} (\bibinfo {year} {2012})}\BibitemShut {NoStop}%
\bibitem [{\citenamefont {Dixon}\ \emph {et~al.}(2009)\citenamefont {Dixon}, \citenamefont {Starling}, \citenamefont {Jordan},\ and\ \citenamefont {Howell}}]{dixon2009ultrasensitive}%
  \BibitemOpen
  \bibfield  {author} {\bibinfo {author} {\bibfnamefont {P.~B.}\ \bibnamefont {Dixon}}, \bibinfo {author} {\bibfnamefont {D.~J.}\ \bibnamefont {Starling}}, \bibinfo {author} {\bibfnamefont {A.~N.}\ \bibnamefont {Jordan}},\ and\ \bibinfo {author} {\bibfnamefont {J.~C.}\ \bibnamefont {Howell}},\ }\bibfield  {title} {\bibinfo {title} {{Ultrasensitive Beam Deflection Measurement via Interferometric Weak Value Amplification}},\ }\href {https://doi.org/10.1103/PhysRevLett.102.173601} {\bibfield  {journal} {\bibinfo  {journal} {Phys. Rev. Lett.}\ }\textbf {\bibinfo {volume} {102}},\ \bibinfo {pages} {173601} (\bibinfo {year} {2009})}\BibitemShut {NoStop}%
\bibitem [{\citenamefont {Hosten}\ and\ \citenamefont {Kwiat}(2008)}]{hosten2008observation}%
  \BibitemOpen
  \bibfield  {author} {\bibinfo {author} {\bibfnamefont {O.}~\bibnamefont {Hosten}}\ and\ \bibinfo {author} {\bibfnamefont {P.}~\bibnamefont {Kwiat}},\ }\bibfield  {title} {\bibinfo {title} {Observation of the {S}pin {H}all {E}ffect of {L}ight via {W}eak {M}easurements},\ }\href {https://doi.org/10.1126/science.1152697} {\bibfield  {journal} {\bibinfo  {journal} {Science}\ }\textbf {\bibinfo {volume} {319}},\ \bibinfo {pages} {787} (\bibinfo {year} {2008})}\BibitemShut {NoStop}%
\bibitem [{\citenamefont {Brunner}\ and\ \citenamefont {Simon}(2010)}]{brunner2010measuringsmall}%
  \BibitemOpen
  \bibfield  {author} {\bibinfo {author} {\bibfnamefont {N.}~\bibnamefont {Brunner}}\ and\ \bibinfo {author} {\bibfnamefont {C.}~\bibnamefont {Simon}},\ }\bibfield  {title} {\bibinfo {title} {Measuring {S}mall {L}ongitudinal {P}hase {S}hifts: {W}eak {M}easurements or {S}tandard {I}nterferometry?},\ }\href {https://doi.org/10.1103/physrevlett.105.010405} {\bibfield  {journal} {\bibinfo  {journal} {Phys. Rev. Lett.}\ }\textbf {\bibinfo {volume} {105}},\ \bibinfo {pages} {010405} (\bibinfo {year} {2010})}\BibitemShut {NoStop}%
\bibitem [{\citenamefont {Gherardini}\ and\ \citenamefont {De~Chiara}(2024)}]{gherardini2024quasiprobabilities}%
  \BibitemOpen
  \bibfield  {author} {\bibinfo {author} {\bibfnamefont {S.}~\bibnamefont {Gherardini}}\ and\ \bibinfo {author} {\bibfnamefont {G.}~\bibnamefont {De~Chiara}},\ }\bibfield  {title} {\bibinfo {title} {{Quasiprobabilities in Quantum Thermodynamics and Many-Body Systems}},\ }\href {https://doi.org/10.1103/PRXQuantum.5.030201} {\bibfield  {journal} {\bibinfo  {journal} {PRX Quantum}\ }\textbf {\bibinfo {volume} {5}},\ \bibinfo {pages} {030201} (\bibinfo {year} {2024})}\BibitemShut {NoStop}%
\bibitem [{\citenamefont {Hernández-Gómez}\ \emph {et~al.}(2024)\citenamefont {Hernández-Gómez}, \citenamefont {Isogawa}, \citenamefont {Belenchia}, \citenamefont {Levy}, \citenamefont {Fabbri}, \citenamefont {Gherardini},\ and\ \citenamefont {Cappellaro}}]{hernandezGomez2024interferometry}%
  \BibitemOpen
  \bibfield  {author} {\bibinfo {author} {\bibfnamefont {S.}~\bibnamefont {Hernández-Gómez}}, \bibinfo {author} {\bibfnamefont {T.}~\bibnamefont {Isogawa}}, \bibinfo {author} {\bibfnamefont {A.}~\bibnamefont {Belenchia}}, \bibinfo {author} {\bibfnamefont {A.}~\bibnamefont {Levy}}, \bibinfo {author} {\bibfnamefont {N.}~\bibnamefont {Fabbri}}, \bibinfo {author} {\bibfnamefont {S.}~\bibnamefont {Gherardini}},\ and\ \bibinfo {author} {\bibfnamefont {P.}~\bibnamefont {Cappellaro}},\ }\bibfield  {title} {\bibinfo {title} {Interferometry of quantum correlation functions to access quasiprobability distribution of work},\ }\href {https://doi.org/10.1038/s41534-024-00913-x} {\bibfield  {journal} {\bibinfo  {journal} {npj Quantum Information}\ }\textbf {\bibinfo {volume} {10}},\ \bibinfo {pages} {115} (\bibinfo {year} {2024})}\BibitemShut {NoStop}%
\bibitem [{\citenamefont {Sajjan}\ \emph {et~al.}(2023)\citenamefont {Sajjan}, \citenamefont {Singh}, \citenamefont {Selvarajan},\ and\ \citenamefont {Kais}}]{sajjan2023imaginary}%
  \BibitemOpen
  \bibfield  {author} {\bibinfo {author} {\bibfnamefont {M.}~\bibnamefont {Sajjan}}, \bibinfo {author} {\bibfnamefont {V.}~\bibnamefont {Singh}}, \bibinfo {author} {\bibfnamefont {R.}~\bibnamefont {Selvarajan}},\ and\ \bibinfo {author} {\bibfnamefont {S.}~\bibnamefont {Kais}},\ }\bibfield  {title} {\bibinfo {title} {Imaginary components of out-of-time-order correlator and information scrambling for navigating the learning landscape of a quantum machine learning model},\ }\href {https://doi.org/10.1103/PhysRevResearch.5.013146} {\bibfield  {journal} {\bibinfo  {journal} {Phys. Rev. Res.}\ }\textbf {\bibinfo {volume} {5}},\ \bibinfo {pages} {013146} (\bibinfo {year} {2023})}\BibitemShut {NoStop}%
\bibitem [{\citenamefont {Wigderson}(2019)}]{wigderson2019mathematics}%
  \BibitemOpen
  \bibfield  {author} {\bibinfo {author} {\bibfnamefont {A.}~\bibnamefont {Wigderson}},\ }\href {https://doi.org/10.1515/9780691192543} {\emph {\bibinfo {title} {Mathematics and Computation}}}\ (\bibinfo  {publisher} {Princeton University Press},\ \bibinfo {year} {2019})\BibitemShut {NoStop}%
\bibitem [{\citenamefont {Popescu}(2009)}]{popescu2009unitaryinvariants}%
  \BibitemOpen
  \bibfield  {author} {\bibinfo {author} {\bibfnamefont {G.}~\bibnamefont {Popescu}},\ }\href@noop {} {\emph {\bibinfo {title} {{Unitary Invariants in Multivariable Operator Theory}}}}\ (\bibinfo  {publisher} {American Mathematical Society},\ \bibinfo {year} {2009})\ p.~\bibinfo {pages} {91}\BibitemShut {NoStop}%
\bibitem [{\citenamefont {Chien}\ and\ \citenamefont {Waldron}(2016)}]{chien2016characterization}%
  \BibitemOpen
  \bibfield  {author} {\bibinfo {author} {\bibfnamefont {T.-Y.}\ \bibnamefont {Chien}}\ and\ \bibinfo {author} {\bibfnamefont {S.}~\bibnamefont {Waldron}},\ }\bibfield  {title} {\bibinfo {title} {{A Characterization of Projective Unitary Equivalence of Finite Frames and Applications}},\ }\href {https://doi.org/10.1137/15m1042140} {\bibfield  {journal} {\bibinfo  {journal} {SIAM Journal on Discrete Mathematics}\ }\textbf {\bibinfo {volume} {30}},\ \bibinfo {pages} {976} (\bibinfo {year} {2016})}\BibitemShut {NoStop}%
\bibitem [{\citenamefont {Oszmaniec}\ \emph {et~al.}(2024)\citenamefont {Oszmaniec}, \citenamefont {Brod},\ and\ \citenamefont {Galvão}}]{oszmaniec2024measuring}%
  \BibitemOpen
  \bibfield  {author} {\bibinfo {author} {\bibfnamefont {M.}~\bibnamefont {Oszmaniec}}, \bibinfo {author} {\bibfnamefont {D.~J.}\ \bibnamefont {Brod}},\ and\ \bibinfo {author} {\bibfnamefont {E.~F.}\ \bibnamefont {Galvão}},\ }\bibfield  {title} {\bibinfo {title} {Measuring relational information between quantum states, and applications},\ }\href {https://doi.org/10.1088/1367-2630/ad1a27} {\bibfield  {journal} {\bibinfo  {journal} {New Journal of Physics}\ }\textbf {\bibinfo {volume} {26}},\ \bibinfo {pages} {013053} (\bibinfo {year} {2024})}\BibitemShut {NoStop}%
\bibitem [{\citenamefont {Bargmann}(1964)}]{bargmann1964note}%
  \BibitemOpen
  \bibfield  {author} {\bibinfo {author} {\bibfnamefont {V.}~\bibnamefont {Bargmann}},\ }\bibfield  {title} {\bibinfo {title} {{Note on Wigner’s Theorem on Symmetry Operations}},\ }\href {https://doi.org/10.1063/1.1704188} {\bibfield  {journal} {\bibinfo  {journal} {Journal of Mathematical Physics}\ }\textbf {\bibinfo {volume} {5}},\ \bibinfo {pages} {862} (\bibinfo {year} {1964})}\BibitemShut {NoStop}%
\bibitem [{\citenamefont {Wagner}\ \emph {et~al.}(2024{\natexlab{b}})\citenamefont {Wagner}, \citenamefont {Schwartzman-Nowik}, \citenamefont {Paiva}, \citenamefont {Te’eni}, \citenamefont {Ruiz-Molero}, \citenamefont {Barbosa}, \citenamefont {Cohen},\ and\ \citenamefont {Galvão}}]{wagner2024quantumcircuits}%
  \BibitemOpen
  \bibfield  {author} {\bibinfo {author} {\bibfnamefont {R.}~\bibnamefont {Wagner}}, \bibinfo {author} {\bibfnamefont {Z.}~\bibnamefont {Schwartzman-Nowik}}, \bibinfo {author} {\bibfnamefont {I.~L.}\ \bibnamefont {Paiva}}, \bibinfo {author} {\bibfnamefont {A.}~\bibnamefont {Te’eni}}, \bibinfo {author} {\bibfnamefont {A.}~\bibnamefont {Ruiz-Molero}}, \bibinfo {author} {\bibfnamefont {R.~S.}\ \bibnamefont {Barbosa}}, \bibinfo {author} {\bibfnamefont {E.}~\bibnamefont {Cohen}},\ and\ \bibinfo {author} {\bibfnamefont {E.~F.}\ \bibnamefont {Galvão}},\ }\bibfield  {title} {\bibinfo {title} {{Quantum circuits for measuring weak values, Kirkwood–Dirac quasiprobability distributions, and state spectra}},\ }\href {https://doi.org/10.1088/2058-9565/ad124c} {\bibfield  {journal} {\bibinfo  {journal} {Quantum Science and Technology}\ }\textbf {\bibinfo {volume} {9}},\ \bibinfo {pages} {015030} (\bibinfo {year} {2024}{\natexlab{b}})}\BibitemShut {NoStop}%
\bibitem [{\citenamefont {Halpern}\ \emph {et~al.}(2018)\citenamefont {Halpern}, \citenamefont {Swingle},\ and\ \citenamefont {Dressel}}]{halpern2018quasiprobability}%
  \BibitemOpen
  \bibfield  {author} {\bibinfo {author} {\bibfnamefont {N.~Y.}\ \bibnamefont {Halpern}}, \bibinfo {author} {\bibfnamefont {B.}~\bibnamefont {Swingle}},\ and\ \bibinfo {author} {\bibfnamefont {J.}~\bibnamefont {Dressel}},\ }\bibfield  {title} {\bibinfo {title} {Quasiprobability behind the out-of-time-ordered correlator},\ }\href {https://doi.org/10.1103/physreva.97.042105} {\bibfield  {journal} {\bibinfo  {journal} {Phys. Rev. A}\ }\textbf {\bibinfo {volume} {97}},\ \bibinfo {pages} {042105} (\bibinfo {year} {2018})}\BibitemShut {NoStop}%
\bibitem [{\citenamefont {Quek}\ \emph {et~al.}(2024)\citenamefont {Quek}, \citenamefont {Kaur},\ and\ \citenamefont {Wilde}}]{quek2024multivariatetrace}%
  \BibitemOpen
  \bibfield  {author} {\bibinfo {author} {\bibfnamefont {Y.}~\bibnamefont {Quek}}, \bibinfo {author} {\bibfnamefont {E.}~\bibnamefont {Kaur}},\ and\ \bibinfo {author} {\bibfnamefont {M.~M.}\ \bibnamefont {Wilde}},\ }\bibfield  {title} {\bibinfo {title} {Multivariate trace estimation in constant quantum depth},\ }\href {https://doi.org/10.22331/q-2024-01-10-1220} {\bibfield  {journal} {\bibinfo  {journal} {{Quantum}}\ }\textbf {\bibinfo {volume} {8}},\ \bibinfo {pages} {1220} (\bibinfo {year} {2024})}\BibitemShut {NoStop}%
\bibitem [{\citenamefont {Pont}\ \emph {et~al.}(2022)\citenamefont {Pont}, \citenamefont {Albiero}, \citenamefont {Thomas}, \citenamefont {Spagnolo}, \citenamefont {Ceccarelli}, \citenamefont {Corrielli}, \citenamefont {Brieussel}, \citenamefont {Somaschi}, \citenamefont {Huet}, \citenamefont {Harouri}, \citenamefont {Lema\^{\i}tre}, \citenamefont {Sagnes}, \citenamefont {Belabas}, \citenamefont {Sciarrino}, \citenamefont {Osellame}, \citenamefont {Senellart},\ and\ \citenamefont {Crespi}}]{pont2022quantifying}%
  \BibitemOpen
  \bibfield  {author} {\bibinfo {author} {\bibfnamefont {M.}~\bibnamefont {Pont}}, \bibinfo {author} {\bibfnamefont {R.}~\bibnamefont {Albiero}}, \bibinfo {author} {\bibfnamefont {S.~E.}\ \bibnamefont {Thomas}}, \bibinfo {author} {\bibfnamefont {N.}~\bibnamefont {Spagnolo}}, \bibinfo {author} {\bibfnamefont {F.}~\bibnamefont {Ceccarelli}}, \bibinfo {author} {\bibfnamefont {G.}~\bibnamefont {Corrielli}}, \bibinfo {author} {\bibfnamefont {A.}~\bibnamefont {Brieussel}}, \bibinfo {author} {\bibfnamefont {N.}~\bibnamefont {Somaschi}}, \bibinfo {author} {\bibfnamefont {H.}~\bibnamefont {Huet}}, \bibinfo {author} {\bibfnamefont {A.}~\bibnamefont {Harouri}}, \bibinfo {author} {\bibfnamefont {A.}~\bibnamefont {Lema\^{\i}tre}}, \bibinfo {author} {\bibfnamefont {I.}~\bibnamefont {Sagnes}}, \bibinfo {author} {\bibfnamefont {N.}~\bibnamefont {Belabas}}, \bibinfo {author} {\bibfnamefont {F.}~\bibnamefont {Sciarrino}}, \bibinfo {author} {\bibfnamefont {R.}~\bibnamefont {Osellame}}, \bibinfo {author} {\bibfnamefont
  {P.}~\bibnamefont {Senellart}},\ and\ \bibinfo {author} {\bibfnamefont {A.}~\bibnamefont {Crespi}},\ }\bibfield  {title} {\bibinfo {title} {{Quantifying $n$-Photon Indistinguishability with a Cyclic Integrated Interferometer}},\ }\href {https://doi.org/10.1103/PhysRevX.12.031033} {\bibfield  {journal} {\bibinfo  {journal} {Phys. Rev. X}\ }\textbf {\bibinfo {volume} {12}},\ \bibinfo {pages} {031033} (\bibinfo {year} {2022})}\BibitemShut {NoStop}%
\bibitem [{\citenamefont {Simonov}\ \emph {et~al.}(2025)\citenamefont {Simonov}, \citenamefont {Wagner},\ and\ \citenamefont {Galvão}}]{simonov2025estimationmultivariatetracesstates}%
  \BibitemOpen
  \bibfield  {author} {\bibinfo {author} {\bibfnamefont {K.}~\bibnamefont {Simonov}}, \bibinfo {author} {\bibfnamefont {R.}~\bibnamefont {Wagner}},\ and\ \bibinfo {author} {\bibfnamefont {E.}~\bibnamefont {Galvão}},\ }\href {https://arxiv.org/abs/2505.20208} {\bibinfo {title} {Estimation of multivariate traces of states given partial classical information}} (\bibinfo {year} {2025}),\ \Eprint {https://arxiv.org/abs/2505.20208} {arXiv:2505.20208 [quant-ph]} \BibitemShut {NoStop}%
\bibitem [{\citenamefont {Kirkwood}(1933)}]{kirkwood1933quantum}%
  \BibitemOpen
  \bibfield  {author} {\bibinfo {author} {\bibfnamefont {J.~G.}\ \bibnamefont {Kirkwood}},\ }\bibfield  {title} {\bibinfo {title} {Quantum statistics of almost classical assemblies},\ }\href {https://doi.org/10.1103/PhysRev.44.31} {\bibfield  {journal} {\bibinfo  {journal} {Phys. Rev.}\ }\textbf {\bibinfo {volume} {44}},\ \bibinfo {pages} {31} (\bibinfo {year} {1933})}\BibitemShut {NoStop}%
\bibitem [{\citenamefont {Dirac}(1945)}]{dirac1945analogy}%
  \BibitemOpen
  \bibfield  {author} {\bibinfo {author} {\bibfnamefont {P.~A.~M.}\ \bibnamefont {Dirac}},\ }\bibfield  {title} {\bibinfo {title} {On the analogy between classical and quantum mechanics},\ }\href {https://doi.org/10.1103/RevModPhys.17.195} {\bibfield  {journal} {\bibinfo  {journal} {Rev. Mod. Phys.}\ }\textbf {\bibinfo {volume} {17}},\ \bibinfo {pages} {195} (\bibinfo {year} {1945})}\BibitemShut {NoStop}%
\bibitem [{\citenamefont {Arvidsson-Shukur}\ \emph {et~al.}(2024)\citenamefont {Arvidsson-Shukur}, \citenamefont {Braasch~Jr}, \citenamefont {De~Bièvre}, \citenamefont {Dressel}, \citenamefont {Jordan}, \citenamefont {Langrenez}, \citenamefont {Lostaglio}, \citenamefont {Lundeen},\ and\ \citenamefont {Halpern}}]{arvidssonshukur2024properties}%
  \BibitemOpen
  \bibfield  {author} {\bibinfo {author} {\bibfnamefont {D.~R.~M.}\ \bibnamefont {Arvidsson-Shukur}}, \bibinfo {author} {\bibfnamefont {W.~F.}\ \bibnamefont {Braasch~Jr}}, \bibinfo {author} {\bibfnamefont {S.}~\bibnamefont {De~Bièvre}}, \bibinfo {author} {\bibfnamefont {J.}~\bibnamefont {Dressel}}, \bibinfo {author} {\bibfnamefont {A.~N.}\ \bibnamefont {Jordan}}, \bibinfo {author} {\bibfnamefont {C.}~\bibnamefont {Langrenez}}, \bibinfo {author} {\bibfnamefont {M.}~\bibnamefont {Lostaglio}}, \bibinfo {author} {\bibfnamefont {J.~S.}\ \bibnamefont {Lundeen}},\ and\ \bibinfo {author} {\bibfnamefont {N.~Y.}\ \bibnamefont {Halpern}},\ }\bibfield  {title} {\bibinfo {title} {{Properties and applications of the Kirkwood–Dirac distribution}},\ }\href {https://doi.org/10.1088/1367-2630/ada05d} {\bibfield  {journal} {\bibinfo  {journal} {New Journal of Physics}\ }\textbf {\bibinfo {volume} {26}},\ \bibinfo {pages} {121201} (\bibinfo {year} {2024})}\BibitemShut {NoStop}%
\bibitem [{\citenamefont {Schmid}\ \emph {et~al.}(2024)\citenamefont {Schmid}, \citenamefont {Baldij\~ao}, \citenamefont {Yīng}, \citenamefont {Wagner},\ and\ \citenamefont {Selby}}]{schmid2024kirkwood}%
  \BibitemOpen
  \bibfield  {author} {\bibinfo {author} {\bibfnamefont {D.}~\bibnamefont {Schmid}}, \bibinfo {author} {\bibfnamefont {R.~D.}\ \bibnamefont {Baldij\~ao}}, \bibinfo {author} {\bibfnamefont {Y.}~\bibnamefont {Yīng}}, \bibinfo {author} {\bibfnamefont {R.}~\bibnamefont {Wagner}},\ and\ \bibinfo {author} {\bibfnamefont {J.~H.}\ \bibnamefont {Selby}},\ }\bibfield  {title} {\bibinfo {title} {Kirkwood-dirac representations beyond quantum states and their relation to noncontextuality},\ }\href {https://doi.org/10.1103/PhysRevA.110.052206} {\bibfield  {journal} {\bibinfo  {journal} {Phys. Rev. A}\ }\textbf {\bibinfo {volume} {110}},\ \bibinfo {pages} {052206} (\bibinfo {year} {2024})}\BibitemShut {NoStop}%
\bibitem [{\citenamefont {Liu}\ and\ \citenamefont {Cheng}(2025)}]{liu2025boundarykirkwooddiracquasiprobability}%
  \BibitemOpen
  \bibfield  {author} {\bibinfo {author} {\bibfnamefont {L.}~\bibnamefont {Liu}}\ and\ \bibinfo {author} {\bibfnamefont {S.}~\bibnamefont {Cheng}},\ }\href {https://arxiv.org/abs/2504.09238} {\bibinfo {title} {The boundary of {K}irkwood--{D}irac quasiprobability}} (\bibinfo {year} {2025}),\ \Eprint {https://arxiv.org/abs/2504.09238} {arXiv:2504.09238 [quant-ph]} \BibitemShut {NoStop}%
\bibitem [{\citenamefont {Yunger~Halpern}\ \emph {et~al.}(2018)\citenamefont {Yunger~Halpern}, \citenamefont {Swingle},\ and\ \citenamefont {Dressel}}]{yunger2018quasiprobability}%
  \BibitemOpen
  \bibfield  {author} {\bibinfo {author} {\bibfnamefont {N.}~\bibnamefont {Yunger~Halpern}}, \bibinfo {author} {\bibfnamefont {B.}~\bibnamefont {Swingle}},\ and\ \bibinfo {author} {\bibfnamefont {J.}~\bibnamefont {Dressel}},\ }\bibfield  {title} {\bibinfo {title} {Quasiprobability behind the out-of-time-ordered correlator},\ }\href {https://doi.org/https://doi.org/10.1103/PhysRevA.97.042105} {\bibfield  {journal} {\bibinfo  {journal} {Physical Review A}\ }\textbf {\bibinfo {volume} {97}},\ \bibinfo {pages} {042105} (\bibinfo {year} {2018})}\BibitemShut {NoStop}%
\bibitem [{\citenamefont {Gonz{\'a}lez~Alonso}\ \emph {et~al.}(2019)\citenamefont {Gonz{\'a}lez~Alonso}, \citenamefont {Yunger~Halpern},\ and\ \citenamefont {Dressel}}]{gonzalez2019out}%
  \BibitemOpen
  \bibfield  {author} {\bibinfo {author} {\bibfnamefont {J.~R.}\ \bibnamefont {Gonz{\'a}lez~Alonso}}, \bibinfo {author} {\bibfnamefont {N.}~\bibnamefont {Yunger~Halpern}},\ and\ \bibinfo {author} {\bibfnamefont {J.}~\bibnamefont {Dressel}},\ }\bibfield  {title} {\bibinfo {title} {Out-of-time-ordered-correlator quasiprobabilities robustly witness scrambling},\ }\href {https://doi.org/https://doi.org/10.1103/PhysRevLett.122.040404} {\bibfield  {journal} {\bibinfo  {journal} {Physical Review Letters}\ }\textbf {\bibinfo {volume} {122}},\ \bibinfo {pages} {040404} (\bibinfo {year} {2019})}\BibitemShut {NoStop}%
\bibitem [{\citenamefont {Wagner}\ and\ \citenamefont {Galv\~ao}(2023)}]{wagner2023anomalous}%
  \BibitemOpen
  \bibfield  {author} {\bibinfo {author} {\bibfnamefont {R.}~\bibnamefont {Wagner}}\ and\ \bibinfo {author} {\bibfnamefont {E.~F.}\ \bibnamefont {Galv\~ao}},\ }\bibfield  {title} {\bibinfo {title} {Simple proof that anomalous weak values require coherence},\ }\href {https://doi.org/10.1103/PhysRevA.108.L040202} {\bibfield  {journal} {\bibinfo  {journal} {Phys. Rev. A}\ }\textbf {\bibinfo {volume} {108}},\ \bibinfo {pages} {L040202} (\bibinfo {year} {2023})}\BibitemShut {NoStop}%
\bibitem [{\citenamefont {Hofmann}(2012)}]{hofmann2012complex}%
  \BibitemOpen
  \bibfield  {author} {\bibinfo {author} {\bibfnamefont {H.~F.}\ \bibnamefont {Hofmann}},\ }\bibfield  {title} {\bibinfo {title} {{Complex Joint Probabilities as Expressions of Reversible Transformations in Quantum Mechanics}},\ }\href {https://doi.org/10.1088/1367-2630/14/4/043031} {\bibfield  {journal} {\bibinfo  {journal} {New Journal of Physics}\ }\textbf {\bibinfo {volume} {14}},\ \bibinfo {pages} {043031} (\bibinfo {year} {2012})}\BibitemShut {NoStop}%
\bibitem [{\citenamefont {Silva~Pratapsi}\ \emph {et~al.}(2025)\citenamefont {Silva~Pratapsi}, \citenamefont {Deffner},\ and\ \citenamefont {Gherardini}}]{sagarsilvapratapsi2025quantumspeedlimits}%
  \BibitemOpen
  \bibfield  {author} {\bibinfo {author} {\bibfnamefont {S.}~\bibnamefont {Silva~Pratapsi}}, \bibinfo {author} {\bibfnamefont {S.}~\bibnamefont {Deffner}},\ and\ \bibinfo {author} {\bibfnamefont {S.}~\bibnamefont {Gherardini}},\ }\bibfield  {title} {\bibinfo {title} {Quantum speed limit for kirkwood–dirac quasiprobabilities},\ }\href {https://doi.org/10.1088/2058-9565/add55d} {\bibfield  {journal} {\bibinfo  {journal} {Quantum Science and Technology}\ }\textbf {\bibinfo {volume} {10}},\ \bibinfo {pages} {035019} (\bibinfo {year} {2025})}\BibitemShut {NoStop}%
\bibitem [{\citenamefont {Mukunda}\ \emph {et~al.}(2001)\citenamefont {Mukunda}, \citenamefont {Arvind}, \citenamefont {Chaturvedi},\ and\ \citenamefont {Simon}}]{mukunda2001Bargmann}%
  \BibitemOpen
  \bibfield  {author} {\bibinfo {author} {\bibfnamefont {N.}~\bibnamefont {Mukunda}}, \bibinfo {author} {\bibnamefont {Arvind}}, \bibinfo {author} {\bibfnamefont {S.}~\bibnamefont {Chaturvedi}},\ and\ \bibinfo {author} {\bibfnamefont {R.}~\bibnamefont {Simon}},\ }\bibfield  {title} {\bibinfo {title} {Bargmann invariants and off-diagonal geometric phases for multilevel quantum systems: A unitary-group approach},\ }\href {https://doi.org/10.1103/PhysRevA.65.012102} {\bibfield  {journal} {\bibinfo  {journal} {Phys. Rev. A}\ }\textbf {\bibinfo {volume} {65}},\ \bibinfo {pages} {012102} (\bibinfo {year} {2001})}\BibitemShut {NoStop}%
\bibitem [{\citenamefont {Mukunda}\ \emph {et~al.}(2003{\natexlab{a}})\citenamefont {Mukunda}, \citenamefont {Arvind}, \citenamefont {Ercolessi}, \citenamefont {Marmo}, \citenamefont {Morandi},\ and\ \citenamefont {Simon}}]{mukunda2003Bargmann}%
  \BibitemOpen
  \bibfield  {author} {\bibinfo {author} {\bibfnamefont {N.}~\bibnamefont {Mukunda}}, \bibinfo {author} {\bibnamefont {Arvind}}, \bibinfo {author} {\bibfnamefont {E.}~\bibnamefont {Ercolessi}}, \bibinfo {author} {\bibfnamefont {G.}~\bibnamefont {Marmo}}, \bibinfo {author} {\bibfnamefont {G.}~\bibnamefont {Morandi}},\ and\ \bibinfo {author} {\bibfnamefont {R.}~\bibnamefont {Simon}},\ }\bibfield  {title} {\bibinfo {title} {Bargmann invariants, null phase curves, and a theory of the geometric phase},\ }\href {https://doi.org/10.1103/PhysRevA.67.042114} {\bibfield  {journal} {\bibinfo  {journal} {Phys. Rev. A}\ }\textbf {\bibinfo {volume} {67}},\ \bibinfo {pages} {042114} (\bibinfo {year} {2003}{\natexlab{a}})}\BibitemShut {NoStop}%
\bibitem [{\citenamefont {Mukunda}\ \emph {et~al.}(2003{\natexlab{b}})\citenamefont {Mukunda}, \citenamefont {Aravind},\ and\ \citenamefont {Simon}}]{mukunda2003Wigner}%
  \BibitemOpen
  \bibfield  {author} {\bibinfo {author} {\bibfnamefont {N.}~\bibnamefont {Mukunda}}, \bibinfo {author} {\bibfnamefont {P.~K.}\ \bibnamefont {Aravind}},\ and\ \bibinfo {author} {\bibfnamefont {R.}~\bibnamefont {Simon}},\ }\bibfield  {title} {\bibinfo {title} {{Wigner rotations, Bargmann invariants and geometric phases}},\ }\href {https://doi.org/10.1088/0305-4470/36/9/312} {\bibfield  {journal} {\bibinfo  {journal} {Journal of Physics A: Mathematical and General}\ }\textbf {\bibinfo {volume} {36}},\ \bibinfo {pages} {2347} (\bibinfo {year} {2003}{\natexlab{b}})}\BibitemShut {NoStop}%
\bibitem [{\citenamefont {Avdoshkin}\ and\ \citenamefont {Popov}(2023)}]{avdoshkin2023extrinsic}%
  \BibitemOpen
  \bibfield  {author} {\bibinfo {author} {\bibfnamefont {A.}~\bibnamefont {Avdoshkin}}\ and\ \bibinfo {author} {\bibfnamefont {F.~K.}\ \bibnamefont {Popov}},\ }\bibfield  {title} {\bibinfo {title} {Extrinsic geometry of quantum states},\ }\href {https://doi.org/10.1103/PhysRevB.107.245136} {\bibfield  {journal} {\bibinfo  {journal} {Phys. Rev. B}\ }\textbf {\bibinfo {volume} {107}},\ \bibinfo {pages} {245136} (\bibinfo {year} {2023})}\BibitemShut {NoStop}%
\bibitem [{\citenamefont {Menssen}\ \emph {et~al.}(2017)\citenamefont {Menssen}, \citenamefont {Jones}, \citenamefont {Metcalf}, \citenamefont {Tichy}, \citenamefont {Barz}, \citenamefont {Kolthammer},\ and\ \citenamefont {Walmsley}}]{menssen2017distinguishability}%
  \BibitemOpen
  \bibfield  {author} {\bibinfo {author} {\bibfnamefont {A.~J.}\ \bibnamefont {Menssen}}, \bibinfo {author} {\bibfnamefont {A.~E.}\ \bibnamefont {Jones}}, \bibinfo {author} {\bibfnamefont {B.~J.}\ \bibnamefont {Metcalf}}, \bibinfo {author} {\bibfnamefont {M.~C.}\ \bibnamefont {Tichy}}, \bibinfo {author} {\bibfnamefont {S.}~\bibnamefont {Barz}}, \bibinfo {author} {\bibfnamefont {W.~S.}\ \bibnamefont {Kolthammer}},\ and\ \bibinfo {author} {\bibfnamefont {I.~A.}\ \bibnamefont {Walmsley}},\ }\bibfield  {title} {\bibinfo {title} {Distinguishability and many-particle interference},\ }\href {https://doi.org/10.1103/PhysRevLett.118.153603} {\bibfield  {journal} {\bibinfo  {journal} {Phys. Rev. Lett.}\ }\textbf {\bibinfo {volume} {118}},\ \bibinfo {pages} {153603} (\bibinfo {year} {2017})}\BibitemShut {NoStop}%
\bibitem [{\citenamefont {Minke}\ \emph {et~al.}(2021)\citenamefont {Minke}, \citenamefont {Buchleitner},\ and\ \citenamefont {Dittel}}]{minke2021characterizing}%
  \BibitemOpen
  \bibfield  {author} {\bibinfo {author} {\bibfnamefont {A.~M.}\ \bibnamefont {Minke}}, \bibinfo {author} {\bibfnamefont {A.}~\bibnamefont {Buchleitner}},\ and\ \bibinfo {author} {\bibfnamefont {C.}~\bibnamefont {Dittel}},\ }\bibfield  {title} {\bibinfo {title} {{Characterizing four-body indistinguishability via symmetries}},\ }\href {https://doi.org/10.1088/1367-2630/ac0fb1} {\bibfield  {journal} {\bibinfo  {journal} {New Journal of Physics}\ }\textbf {\bibinfo {volume} {23}},\ \bibinfo {pages} {073028} (\bibinfo {year} {2021})}\BibitemShut {NoStop}%
\bibitem [{\citenamefont {Rodari}\ \emph {et~al.}(2024)\citenamefont {Rodari}, \citenamefont {Fernandes}, \citenamefont {Caruccio}, \citenamefont {Suprano}, \citenamefont {Hoch}, \citenamefont {Giordani}, \citenamefont {Carvacho}, \citenamefont {Albiero}, \citenamefont {Giano}, \citenamefont {Corrielli}, \citenamefont {Ceccarelli}, \citenamefont {Osellame}, \citenamefont {Brod}, \citenamefont {Novo}, \citenamefont {Spagnolo}, \citenamefont {Galvão},\ and\ \citenamefont {Sciarrino}}]{rodari2024experimentalobservationcounterintuitivefeatures}%
  \BibitemOpen
  \bibfield  {author} {\bibinfo {author} {\bibfnamefont {G.}~\bibnamefont {Rodari}}, \bibinfo {author} {\bibfnamefont {C.}~\bibnamefont {Fernandes}}, \bibinfo {author} {\bibfnamefont {E.}~\bibnamefont {Caruccio}}, \bibinfo {author} {\bibfnamefont {A.}~\bibnamefont {Suprano}}, \bibinfo {author} {\bibfnamefont {F.}~\bibnamefont {Hoch}}, \bibinfo {author} {\bibfnamefont {T.}~\bibnamefont {Giordani}}, \bibinfo {author} {\bibfnamefont {G.}~\bibnamefont {Carvacho}}, \bibinfo {author} {\bibfnamefont {R.}~\bibnamefont {Albiero}}, \bibinfo {author} {\bibfnamefont {N.~D.}\ \bibnamefont {Giano}}, \bibinfo {author} {\bibfnamefont {G.}~\bibnamefont {Corrielli}}, \bibinfo {author} {\bibfnamefont {F.}~\bibnamefont {Ceccarelli}}, \bibinfo {author} {\bibfnamefont {R.}~\bibnamefont {Osellame}}, \bibinfo {author} {\bibfnamefont {D.~J.}\ \bibnamefont {Brod}}, \bibinfo {author} {\bibfnamefont {L.}~\bibnamefont {Novo}}, \bibinfo {author} {\bibfnamefont {N.}~\bibnamefont {Spagnolo}}, \bibinfo {author} {\bibfnamefont {E.~F.}\
  \bibnamefont {Galvão}},\ and\ \bibinfo {author} {\bibfnamefont {F.}~\bibnamefont {Sciarrino}},\ }\href {https://doi.org/https://doi.org/10.48550/arXiv.2410.15883} {\bibinfo {title} {Experimental observation of counter-intuitive features of photonic bunching}} (\bibinfo {year} {2024}),\ \Eprint {https://arxiv.org/abs/2410.15883} {arXiv:2410.15883 [quant-ph]} \BibitemShut {NoStop}%
\bibitem [{\citenamefont {Rodari}\ \emph {et~al.}(2025)\citenamefont {Rodari}, \citenamefont {Novo}, \citenamefont {Albiero}, \citenamefont {Suprano}, \citenamefont {Tavares}, \citenamefont {Caruccio}, \citenamefont {Hoch}, \citenamefont {Giordani}, \citenamefont {Carvacho}, \citenamefont {Gardina}, \citenamefont {Di~Giano}, \citenamefont {Di~Giorgio}, \citenamefont {Corrielli}, \citenamefont {Ceccarelli}, \citenamefont {Osellame}, \citenamefont {Spagnolo}, \citenamefont {Galv\~ao},\ and\ \citenamefont {Sciarrino}}]{rodari2024semideviceindependentcharacterizationmultiphoton}%
  \BibitemOpen
  \bibfield  {author} {\bibinfo {author} {\bibfnamefont {G.}~\bibnamefont {Rodari}}, \bibinfo {author} {\bibfnamefont {L.}~\bibnamefont {Novo}}, \bibinfo {author} {\bibfnamefont {R.}~\bibnamefont {Albiero}}, \bibinfo {author} {\bibfnamefont {A.}~\bibnamefont {Suprano}}, \bibinfo {author} {\bibfnamefont {C.~T.}\ \bibnamefont {Tavares}}, \bibinfo {author} {\bibfnamefont {E.}~\bibnamefont {Caruccio}}, \bibinfo {author} {\bibfnamefont {F.}~\bibnamefont {Hoch}}, \bibinfo {author} {\bibfnamefont {T.}~\bibnamefont {Giordani}}, \bibinfo {author} {\bibfnamefont {G.}~\bibnamefont {Carvacho}}, \bibinfo {author} {\bibfnamefont {M.}~\bibnamefont {Gardina}}, \bibinfo {author} {\bibfnamefont {N.}~\bibnamefont {Di~Giano}}, \bibinfo {author} {\bibfnamefont {S.}~\bibnamefont {Di~Giorgio}}, \bibinfo {author} {\bibfnamefont {G.}~\bibnamefont {Corrielli}}, \bibinfo {author} {\bibfnamefont {F.}~\bibnamefont {Ceccarelli}}, \bibinfo {author} {\bibfnamefont {R.}~\bibnamefont {Osellame}}, \bibinfo {author} {\bibfnamefont
  {N.}~\bibnamefont {Spagnolo}}, \bibinfo {author} {\bibfnamefont {E.~F.}\ \bibnamefont {Galv\~ao}},\ and\ \bibinfo {author} {\bibfnamefont {F.}~\bibnamefont {Sciarrino}},\ }\bibfield  {title} {\bibinfo {title} {Semi-device-independent characterization of multiphoton indistinguishability},\ }\href {https://doi.org/10.1103/PRXQuantum.6.020340} {\bibfield  {journal} {\bibinfo  {journal} {PRX Quantum}\ }\textbf {\bibinfo {volume} {6}},\ \bibinfo {pages} {020340} (\bibinfo {year} {2025})}\BibitemShut {NoStop}%
\bibitem [{\citenamefont {Seron}\ \emph {et~al.}(2023)\citenamefont {Seron}, \citenamefont {Novo},\ and\ \citenamefont {Cerf}}]{seron2023boson}%
  \BibitemOpen
  \bibfield  {author} {\bibinfo {author} {\bibfnamefont {B.}~\bibnamefont {Seron}}, \bibinfo {author} {\bibfnamefont {L.}~\bibnamefont {Novo}},\ and\ \bibinfo {author} {\bibfnamefont {N.~J.}\ \bibnamefont {Cerf}},\ }\bibfield  {title} {\bibinfo {title} {{Boson bunching is not maximized by indistinguishable particles}},\ }\href {https://doi.org/10.1038/s41566-023-01213-0} {\bibfield  {journal} {\bibinfo  {journal} {Nature Photonics}\ }\textbf {\bibinfo {volume} {17}},\ \bibinfo {pages} {702} (\bibinfo {year} {2023})}\BibitemShut {NoStop}%
\bibitem [{\citenamefont {Giordani}\ \emph {et~al.}(2021)\citenamefont {Giordani}, \citenamefont {Esposito}, \citenamefont {Hoch}, \citenamefont {Carvacho}, \citenamefont {Brod}, \citenamefont {Galv\~ao}, \citenamefont {Spagnolo},\ and\ \citenamefont {Sciarrino}}]{giordani2021witnesses}%
  \BibitemOpen
  \bibfield  {author} {\bibinfo {author} {\bibfnamefont {T.}~\bibnamefont {Giordani}}, \bibinfo {author} {\bibfnamefont {C.}~\bibnamefont {Esposito}}, \bibinfo {author} {\bibfnamefont {F.}~\bibnamefont {Hoch}}, \bibinfo {author} {\bibfnamefont {G.}~\bibnamefont {Carvacho}}, \bibinfo {author} {\bibfnamefont {D.~J.}\ \bibnamefont {Brod}}, \bibinfo {author} {\bibfnamefont {E.~F.}\ \bibnamefont {Galv\~ao}}, \bibinfo {author} {\bibfnamefont {N.}~\bibnamefont {Spagnolo}},\ and\ \bibinfo {author} {\bibfnamefont {F.}~\bibnamefont {Sciarrino}},\ }\bibfield  {title} {\bibinfo {title} {Witnesses of coherence and dimension from multiphoton indistinguishability tests},\ }\href {https://doi.org/10.1103/PhysRevResearch.3.023031} {\bibfield  {journal} {\bibinfo  {journal} {Phys. Rev. Res.}\ }\textbf {\bibinfo {volume} {3}},\ \bibinfo {pages} {023031} (\bibinfo {year} {2021})}\BibitemShut {NoStop}%
\bibitem [{\citenamefont {Giordani}\ \emph {et~al.}(2020)\citenamefont {Giordani}, \citenamefont {Brod}, \citenamefont {Esposito}, \citenamefont {Viggianiello}, \citenamefont {Romano}, \citenamefont {Flamini}, \citenamefont {Carvacho}, \citenamefont {Spagnolo}, \citenamefont {Galv{\~{a}}o},\ and\ \citenamefont {Sciarrino}}]{giordani2020experimental}%
  \BibitemOpen
  \bibfield  {author} {\bibinfo {author} {\bibfnamefont {T.}~\bibnamefont {Giordani}}, \bibinfo {author} {\bibfnamefont {D.~J.}\ \bibnamefont {Brod}}, \bibinfo {author} {\bibfnamefont {C.}~\bibnamefont {Esposito}}, \bibinfo {author} {\bibfnamefont {N.}~\bibnamefont {Viggianiello}}, \bibinfo {author} {\bibfnamefont {M.}~\bibnamefont {Romano}}, \bibinfo {author} {\bibfnamefont {F.}~\bibnamefont {Flamini}}, \bibinfo {author} {\bibfnamefont {G.}~\bibnamefont {Carvacho}}, \bibinfo {author} {\bibfnamefont {N.}~\bibnamefont {Spagnolo}}, \bibinfo {author} {\bibfnamefont {E.~F.}\ \bibnamefont {Galv{\~{a}}o}},\ and\ \bibinfo {author} {\bibfnamefont {F.}~\bibnamefont {Sciarrino}},\ }\bibfield  {title} {\bibinfo {title} {Experimental quantification of four-photon indistinguishability},\ }\href {https://doi.org/10.1088/1367-2630/ab7a30} {\bibfield  {journal} {\bibinfo  {journal} {New Journal of Physics}\ }\textbf {\bibinfo {volume} {22}},\ \bibinfo {pages} {043001} (\bibinfo {year} {2020})}\BibitemShut {NoStop}%
\bibitem [{\citenamefont {Brod}\ \emph {et~al.}(2019)\citenamefont {Brod}, \citenamefont {Galv\~ao}, \citenamefont {Viggianiello}, \citenamefont {Flamini}, \citenamefont {Spagnolo},\ and\ \citenamefont {Sciarrino}}]{brod2019witnessing}%
  \BibitemOpen
  \bibfield  {author} {\bibinfo {author} {\bibfnamefont {D.~J.}\ \bibnamefont {Brod}}, \bibinfo {author} {\bibfnamefont {E.~F.}\ \bibnamefont {Galv\~ao}}, \bibinfo {author} {\bibfnamefont {N.}~\bibnamefont {Viggianiello}}, \bibinfo {author} {\bibfnamefont {F.}~\bibnamefont {Flamini}}, \bibinfo {author} {\bibfnamefont {N.}~\bibnamefont {Spagnolo}},\ and\ \bibinfo {author} {\bibfnamefont {F.}~\bibnamefont {Sciarrino}},\ }\bibfield  {title} {\bibinfo {title} {{Witnessing Genuine Multiphoton Indistinguishability}},\ }\href {https://doi.org/10.1103/PhysRevLett.122.063602} {\bibfield  {journal} {\bibinfo  {journal} {Phys. Rev. Lett.}\ }\textbf {\bibinfo {volume} {122}},\ \bibinfo {pages} {063602} (\bibinfo {year} {2019})}\BibitemShut {NoStop}%
\bibitem [{\citenamefont {Annoni}\ and\ \citenamefont {Wein}(2025)}]{annoni2025incoherentbehaviorpartiallydistinguishable}%
  \BibitemOpen
  \bibfield  {author} {\bibinfo {author} {\bibfnamefont {E.}~\bibnamefont {Annoni}}\ and\ \bibinfo {author} {\bibfnamefont {S.~C.}\ \bibnamefont {Wein}},\ }\href {https://doi.org/https://doi.org/10.48550/arXiv.2502.05047} {\bibinfo {title} {Incoherent behavior of partially distinguishable photons}},\ \bibinfo {howpublished} {arXiv:2502.05047 [quant-ph]} (\bibinfo {year} {2025})\BibitemShut {NoStop}%
\bibitem [{\citenamefont {Bong}\ \emph {et~al.}(2018)\citenamefont {Bong}, \citenamefont {Tischler}, \citenamefont {Patel}, \citenamefont {Wollmann}, \citenamefont {Pryde},\ and\ \citenamefont {Hall}}]{bong2018strong}%
  \BibitemOpen
  \bibfield  {author} {\bibinfo {author} {\bibfnamefont {K.-W.}\ \bibnamefont {Bong}}, \bibinfo {author} {\bibfnamefont {N.}~\bibnamefont {Tischler}}, \bibinfo {author} {\bibfnamefont {R.~B.}\ \bibnamefont {Patel}}, \bibinfo {author} {\bibfnamefont {S.}~\bibnamefont {Wollmann}}, \bibinfo {author} {\bibfnamefont {G.~J.}\ \bibnamefont {Pryde}},\ and\ \bibinfo {author} {\bibfnamefont {M.~J.~W.}\ \bibnamefont {Hall}},\ }\bibfield  {title} {\bibinfo {title} {{Strong Unitary and Overlap Uncertainty Relations: Theory and Experiment}},\ }\href {https://doi.org/10.1103/PhysRevLett.120.230402} {\bibfield  {journal} {\bibinfo  {journal} {Phys. Rev. Lett.}\ }\textbf {\bibinfo {volume} {120}},\ \bibinfo {pages} {230402} (\bibinfo {year} {2018})}\BibitemShut {NoStop}%
\bibitem [{\citenamefont {Lostaglio}\ \emph {et~al.}(2023)\citenamefont {Lostaglio}, \citenamefont {Belenchia}, \citenamefont {Levy}, \citenamefont {Hern\'{a}ndez-G\'{o}mez}, \citenamefont {Fabbri},\ and\ \citenamefont {Gherardini}}]{lostaglio2022kirkwood}%
  \BibitemOpen
  \bibfield  {author} {\bibinfo {author} {\bibfnamefont {M.}~\bibnamefont {Lostaglio}}, \bibinfo {author} {\bibfnamefont {A.}~\bibnamefont {Belenchia}}, \bibinfo {author} {\bibfnamefont {A.}~\bibnamefont {Levy}}, \bibinfo {author} {\bibfnamefont {S.}~\bibnamefont {Hern\'{a}ndez-G\'{o}mez}}, \bibinfo {author} {\bibfnamefont {N.}~\bibnamefont {Fabbri}},\ and\ \bibinfo {author} {\bibfnamefont {S.}~\bibnamefont {Gherardini}},\ }\bibfield  {title} {\bibinfo {title} {Kirkwood-{D}irac quasiprobability approach to the statistics of incompatible observables},\ }\href {https://doi.org/10.22331/q-2023-10-09-1128} {\bibfield  {journal} {\bibinfo  {journal} {Quantum}\ }\textbf {\bibinfo {volume} {7}},\ \bibinfo {pages} {1128} (\bibinfo {year} {2023})}\BibitemShut {NoStop}%
\bibitem [{\citenamefont {Levy}\ and\ \citenamefont {Lostaglio}(2020)}]{levy2020quasiprobability}%
  \BibitemOpen
  \bibfield  {author} {\bibinfo {author} {\bibfnamefont {A.}~\bibnamefont {Levy}}\ and\ \bibinfo {author} {\bibfnamefont {M.}~\bibnamefont {Lostaglio}},\ }\bibfield  {title} {\bibinfo {title} {Quasiprobability distribution for heat fluctuations in the quantum regime},\ }\href {https://doi.org/10.1103/PRXQuantum.1.010309} {\bibfield  {journal} {\bibinfo  {journal} {PRX Quantum}\ }\textbf {\bibinfo {volume} {1}},\ \bibinfo {pages} {010309} (\bibinfo {year} {2020})}\BibitemShut {NoStop}%
\bibitem [{\citenamefont {Hern\'andez-G\'omez}\ \emph {et~al.}(2024)\citenamefont {Hern\'andez-G\'omez}, \citenamefont {Gherardini}, \citenamefont {Belenchia}, \citenamefont {Lostaglio}, \citenamefont {Levy},\ and\ \citenamefont {Fabbri}}]{hernandez2024projective}%
  \BibitemOpen
  \bibfield  {author} {\bibinfo {author} {\bibfnamefont {S.}~\bibnamefont {Hern\'andez-G\'omez}}, \bibinfo {author} {\bibfnamefont {S.}~\bibnamefont {Gherardini}}, \bibinfo {author} {\bibfnamefont {A.}~\bibnamefont {Belenchia}}, \bibinfo {author} {\bibfnamefont {M.}~\bibnamefont {Lostaglio}}, \bibinfo {author} {\bibfnamefont {A.}~\bibnamefont {Levy}},\ and\ \bibinfo {author} {\bibfnamefont {N.}~\bibnamefont {Fabbri}},\ }\bibfield  {title} {\bibinfo {title} {Projective measurements can probe nonclassical work extraction and time correlations},\ }\href {https://doi.org/10.1103/PhysRevResearch.6.023280} {\bibfield  {journal} {\bibinfo  {journal} {Phys. Rev. Res.}\ }\textbf {\bibinfo {volume} {6}},\ \bibinfo {pages} {023280} (\bibinfo {year} {2024})}\BibitemShut {NoStop}%
\bibitem [{\citenamefont {Santini}\ \emph {et~al.}(2023)\citenamefont {Santini}, \citenamefont {Solfanelli}, \citenamefont {Gherardini},\ and\ \citenamefont {Collura}}]{santini2023work}%
  \BibitemOpen
  \bibfield  {author} {\bibinfo {author} {\bibfnamefont {A.}~\bibnamefont {Santini}}, \bibinfo {author} {\bibfnamefont {A.}~\bibnamefont {Solfanelli}}, \bibinfo {author} {\bibfnamefont {S.}~\bibnamefont {Gherardini}},\ and\ \bibinfo {author} {\bibfnamefont {M.}~\bibnamefont {Collura}},\ }\bibfield  {title} {\bibinfo {title} {Work statistics, quantum signatures, and enhanced work extraction in quadratic fermionic models},\ }\href {https://doi.org/10.1103/PhysRevB.108.104308} {\bibfield  {journal} {\bibinfo  {journal} {Phys. Rev. B}\ }\textbf {\bibinfo {volume} {108}},\ \bibinfo {pages} {104308} (\bibinfo {year} {2023})}\BibitemShut {NoStop}%
\bibitem [{\citenamefont {Donati}\ \emph {et~al.}(2024)\citenamefont {Donati}, \citenamefont {Cataliotti},\ and\ \citenamefont {Gherardini}}]{donati2024energetics}%
  \BibitemOpen
  \bibfield  {author} {\bibinfo {author} {\bibfnamefont {L.}~\bibnamefont {Donati}}, \bibinfo {author} {\bibfnamefont {F.~S.}\ \bibnamefont {Cataliotti}},\ and\ \bibinfo {author} {\bibfnamefont {S.}~\bibnamefont {Gherardini}},\ }\bibfield  {title} {\bibinfo {title} {Energetics and quantumness of fano coherence generation},\ }\href {https://doi.org/10.1038/s41598-024-67037-2} {\bibfield  {journal} {\bibinfo  {journal} {Scientific Reports}\ }\textbf {\bibinfo {volume} {14}},\ \bibinfo {pages} {20145} (\bibinfo {year} {2024})}\BibitemShut {NoStop}%
\bibitem [{\citenamefont {Fernandes}\ \emph {et~al.}(2024)\citenamefont {Fernandes}, \citenamefont {Wagner}, \citenamefont {Novo},\ and\ \citenamefont {Galv\~ao}}]{fernandes2024unitary}%
  \BibitemOpen
  \bibfield  {author} {\bibinfo {author} {\bibfnamefont {C.}~\bibnamefont {Fernandes}}, \bibinfo {author} {\bibfnamefont {R.}~\bibnamefont {Wagner}}, \bibinfo {author} {\bibfnamefont {L.}~\bibnamefont {Novo}},\ and\ \bibinfo {author} {\bibfnamefont {E.~F.}\ \bibnamefont {Galv\~ao}},\ }\bibfield  {title} {\bibinfo {title} {{Unitary-Invariant Witnesses of Quantum Imaginarity}},\ }\href {https://doi.org/10.1103/PhysRevLett.133.190201} {\bibfield  {journal} {\bibinfo  {journal} {Phys. Rev. Lett.}\ }\textbf {\bibinfo {volume} {133}},\ \bibinfo {pages} {190201} (\bibinfo {year} {2024})}\BibitemShut {NoStop}%
\bibitem [{\citenamefont {Wagner}\ \emph {et~al.}(2025)\citenamefont {Wagner}, \citenamefont {Peres}, \citenamefont {Zambrini~Cruzeiro},\ and\ \citenamefont {Galvão}}]{wagner2025unitary}%
  \BibitemOpen
  \bibfield  {author} {\bibinfo {author} {\bibfnamefont {R.}~\bibnamefont {Wagner}}, \bibinfo {author} {\bibfnamefont {F.~C.~R.}\ \bibnamefont {Peres}}, \bibinfo {author} {\bibfnamefont {E.}~\bibnamefont {Zambrini~Cruzeiro}},\ and\ \bibinfo {author} {\bibfnamefont {E.~F.}\ \bibnamefont {Galvão}},\ }\bibfield  {title} {\bibinfo {title} {Unitary-invariant method for witnessing nonstabilizerness in quantum processors},\ }\href {https://doi.org/10.1088/1751-8121/ade9ff} {\bibfield  {journal} {\bibinfo  {journal} {Journal of Physics A: Mathematical and Theoretical}\ }\textbf {\bibinfo {volume} {58}},\ \bibinfo {pages} {285302} (\bibinfo {year} {2025})}\BibitemShut {NoStop}%
\bibitem [{\citenamefont {Wagner}\ \emph {et~al.}(2024{\natexlab{c}})\citenamefont {Wagner}, \citenamefont {Camillini},\ and\ \citenamefont {Galv{\~{a}}o}}]{wagner2024coherence}%
  \BibitemOpen
  \bibfield  {author} {\bibinfo {author} {\bibfnamefont {R.}~\bibnamefont {Wagner}}, \bibinfo {author} {\bibfnamefont {A.}~\bibnamefont {Camillini}},\ and\ \bibinfo {author} {\bibfnamefont {E.~F.}\ \bibnamefont {Galv{\~{a}}o}},\ }\bibfield  {title} {\bibinfo {title} {Coherence and contextuality in a {M}ach-{Z}ehnder interferometer},\ }\href {https://doi.org/10.22331/q-2024-02-05-1240} {\bibfield  {journal} {\bibinfo  {journal} {{Quantum}}\ }\textbf {\bibinfo {volume} {8}},\ \bibinfo {pages} {1240} (\bibinfo {year} {2024}{\natexlab{c}})}\BibitemShut {NoStop}%
\bibitem [{\citenamefont {Zhang}\ \emph {et~al.}(2024)\citenamefont {Zhang}, \citenamefont {Xie},\ and\ \citenamefont {Tao}}]{zhang2024local}%
  \BibitemOpen
  \bibfield  {author} {\bibinfo {author} {\bibfnamefont {L.}~\bibnamefont {Zhang}}, \bibinfo {author} {\bibfnamefont {B.}~\bibnamefont {Xie}},\ and\ \bibinfo {author} {\bibfnamefont {Y.}~\bibnamefont {Tao}},\ }\href {https://arxiv.org/abs/2412.17237} {\bibinfo {title} {{Local unitarity equivalence and entanglement by Bargmann invariants}}} (\bibinfo {year} {2024}),\ \Eprint {https://arxiv.org/abs/2412.17237} {arXiv:2412.17237 [quant-ph]} \BibitemShut {NoStop}%
\bibitem [{\citenamefont {Giordani}\ \emph {et~al.}(2023)\citenamefont {Giordani}, \citenamefont {Wagner}, \citenamefont {Esposito}, \citenamefont {Camillini}, \citenamefont {Hoch}, \citenamefont {Carvacho}, \citenamefont {Pentangelo}, \citenamefont {Ceccarelli}, \citenamefont {Piacentini}, \citenamefont {Crespi}, \citenamefont {Spagnolo}, \citenamefont {Osellame}, \citenamefont {Galv{\~{a}}o},\ and\ \citenamefont {Sciarrino}}]{giordani2023experimental}%
  \BibitemOpen
  \bibfield  {author} {\bibinfo {author} {\bibfnamefont {T.}~\bibnamefont {Giordani}}, \bibinfo {author} {\bibfnamefont {R.}~\bibnamefont {Wagner}}, \bibinfo {author} {\bibfnamefont {C.}~\bibnamefont {Esposito}}, \bibinfo {author} {\bibfnamefont {A.}~\bibnamefont {Camillini}}, \bibinfo {author} {\bibfnamefont {F.}~\bibnamefont {Hoch}}, \bibinfo {author} {\bibfnamefont {G.}~\bibnamefont {Carvacho}}, \bibinfo {author} {\bibfnamefont {C.}~\bibnamefont {Pentangelo}}, \bibinfo {author} {\bibfnamefont {F.}~\bibnamefont {Ceccarelli}}, \bibinfo {author} {\bibfnamefont {S.}~\bibnamefont {Piacentini}}, \bibinfo {author} {\bibfnamefont {A.}~\bibnamefont {Crespi}}, \bibinfo {author} {\bibfnamefont {N.}~\bibnamefont {Spagnolo}}, \bibinfo {author} {\bibfnamefont {R.}~\bibnamefont {Osellame}}, \bibinfo {author} {\bibfnamefont {E.~F.}\ \bibnamefont {Galv{\~{a}}o}},\ and\ \bibinfo {author} {\bibfnamefont {F.}~\bibnamefont {Sciarrino}},\ }\bibfield  {title} {\bibinfo {title} {Experimental certification of contextuality,
  coherence, and dimension in a programmable universal photonic processor},\ }\href {https://doi.org/10.1126/sciadv.adj4249} {\bibfield  {journal} {\bibinfo  {journal} {Science Advances}\ }\textbf {\bibinfo {volume} {9}},\ \bibinfo {pages} {eadj4249} (\bibinfo {year} {2023})}\BibitemShut {NoStop}%
\bibitem [{\citenamefont {Zhang}\ \emph {et~al.}(2025{\natexlab{c}})\citenamefont {Zhang}, \citenamefont {Xie},\ and\ \citenamefont {Li}}]{zhang2025geometrysets}%
  \BibitemOpen
  \bibfield  {author} {\bibinfo {author} {\bibfnamefont {L.}~\bibnamefont {Zhang}}, \bibinfo {author} {\bibfnamefont {B.}~\bibnamefont {Xie}},\ and\ \bibinfo {author} {\bibfnamefont {B.}~\bibnamefont {Li}},\ }\bibfield  {title} {\bibinfo {title} {{Geometry of sets of Bargmann invariants}},\ }\href {https://doi.org/10.1103/PhysRevA.111.042417} {\bibfield  {journal} {\bibinfo  {journal} {Phys. Rev. A}\ }\textbf {\bibinfo {volume} {111}},\ \bibinfo {pages} {042417} (\bibinfo {year} {2025}{\natexlab{c}})}\BibitemShut {NoStop}%
\bibitem [{\citenamefont {Selby}\ \emph {et~al.}(2023{\natexlab{b}})\citenamefont {Selby}, \citenamefont {Schmid}, \citenamefont {Wolfe}, \citenamefont {Sainz}, \citenamefont {Kunjwal},\ and\ \citenamefont {Spekkens}}]{selby2023accessible}%
  \BibitemOpen
  \bibfield  {author} {\bibinfo {author} {\bibfnamefont {J.~H.}\ \bibnamefont {Selby}}, \bibinfo {author} {\bibfnamefont {D.}~\bibnamefont {Schmid}}, \bibinfo {author} {\bibfnamefont {E.}~\bibnamefont {Wolfe}}, \bibinfo {author} {\bibfnamefont {A.~B.}\ \bibnamefont {Sainz}}, \bibinfo {author} {\bibfnamefont {R.}~\bibnamefont {Kunjwal}},\ and\ \bibinfo {author} {\bibfnamefont {R.~W.}\ \bibnamefont {Spekkens}},\ }\bibfield  {title} {\bibinfo {title} {Accessible fragments of generalized probabilistic theories, cone equivalence, and applications to witnessing nonclassicality},\ }\href {https://doi.org/10.1103/PhysRevA.107.062203} {\bibfield  {journal} {\bibinfo  {journal} {Phys. Rev. A}\ }\textbf {\bibinfo {volume} {107}},\ \bibinfo {pages} {062203} (\bibinfo {year} {2023}{\natexlab{b}})}\BibitemShut {NoStop}%
\bibitem [{\citenamefont {Zamora}\ \emph {et~al.}(2025)\citenamefont {Zamora}, \citenamefont {Macedo}, \citenamefont {Sarubi}, \citenamefont {Alves}, \citenamefont {Poderini},\ and\ \citenamefont {Chaves}}]{zamora2025prepareandmagicsemideviceindependentmagic}%
  \BibitemOpen
  \bibfield  {author} {\bibinfo {author} {\bibfnamefont {S.}~\bibnamefont {Zamora}}, \bibinfo {author} {\bibfnamefont {R.~A.}\ \bibnamefont {Macedo}}, \bibinfo {author} {\bibfnamefont {T.~S.}\ \bibnamefont {Sarubi}}, \bibinfo {author} {\bibfnamefont {M.}~\bibnamefont {Alves}}, \bibinfo {author} {\bibfnamefont {D.}~\bibnamefont {Poderini}},\ and\ \bibinfo {author} {\bibfnamefont {R.}~\bibnamefont {Chaves}},\ }\href {https://arxiv.org/abs/2506.02226} {\bibinfo {title} {Prepare-and-magic: Semi-device independent magic certification in the prepare-and-measure scenario}} (\bibinfo {year} {2025}),\ \Eprint {https://arxiv.org/abs/2506.02226} {arXiv:2506.02226 [quant-ph]} \BibitemShut {NoStop}%
\bibitem [{\citenamefont {Fraser}(2023)}]{fraser2023estimationtheoreticapproachquantum}%
  \BibitemOpen
  \bibfield  {author} {\bibinfo {author} {\bibfnamefont {T.~C.}\ \bibnamefont {Fraser}},\ }\href {https://doi.org/https://doi.org/10.48550/arXiv.2402.10902} {\bibinfo {title} {An estimation theoretic approach to quantum realizability problems}},\ \bibinfo {howpublished} {arXiv:2402.10902 [quant-ph]} (\bibinfo {year} {2023})\BibitemShut {NoStop}%
\bibitem [{\citenamefont {Xu}(2025)}]{xu2025numericalrangesbargmanninvariants}%
  \BibitemOpen
  \bibfield  {author} {\bibinfo {author} {\bibfnamefont {J.}~\bibnamefont {Xu}},\ }\href {https://arxiv.org/abs/2506.13266} {\bibinfo {title} {Numerical ranges of {B}argmann invariants}} (\bibinfo {year} {2025}),\ \Eprint {https://arxiv.org/abs/2506.13266} {arXiv:2506.13266 [quant-ph]} \BibitemShut {NoStop}%
\bibitem [{\citenamefont {Li}\ and\ \citenamefont {Tan}(2025)}]{mao2025bargmann}%
  \BibitemOpen
  \bibfield  {author} {\bibinfo {author} {\bibfnamefont {M.-S.}\ \bibnamefont {Li}}\ and\ \bibinfo {author} {\bibfnamefont {Y.-X.}\ \bibnamefont {Tan}},\ }\bibfield  {title} {\bibinfo {title} {Bargmann invariants for quantum imaginarity},\ }\href {https://doi.org/10.1103/PhysRevA.111.022409} {\bibfield  {journal} {\bibinfo  {journal} {Phys. Rev. A}\ }\textbf {\bibinfo {volume} {111}},\ \bibinfo {pages} {022409} (\bibinfo {year} {2025})}\BibitemShut {NoStop}%
\bibitem [{\citenamefont {Pratapsi}\ \emph {et~al.}(2025)\citenamefont {Pratapsi}, \citenamefont {Gouveia}, \citenamefont {Novo},\ and\ \citenamefont {Galvão}}]{pratapsi2025elementarycharacterizationbargmanninvariants}%
  \BibitemOpen
  \bibfield  {author} {\bibinfo {author} {\bibfnamefont {S.~S.}\ \bibnamefont {Pratapsi}}, \bibinfo {author} {\bibfnamefont {J.}~\bibnamefont {Gouveia}}, \bibinfo {author} {\bibfnamefont {L.}~\bibnamefont {Novo}},\ and\ \bibinfo {author} {\bibfnamefont {E.~F.}\ \bibnamefont {Galvão}},\ }\href {https://arxiv.org/abs/2506.17132} {\bibinfo {title} {{An Elementary Characterization of Bargmann Invariants}}} (\bibinfo {year} {2025}),\ \Eprint {https://arxiv.org/abs/2506.17132} {arXiv:2506.17132 [quant-ph]} \BibitemShut {NoStop}%
\bibitem [{\citenamefont {Budroni}\ \emph {et~al.}(2022)\citenamefont {Budroni}, \citenamefont {Cabello}, \citenamefont {G{\"u}hne}, \citenamefont {Kleinmann},\ and\ \citenamefont {Larsson}}]{budroni2021kochenspeckerreview}%
  \BibitemOpen
  \bibfield  {author} {\bibinfo {author} {\bibfnamefont {C.}~\bibnamefont {Budroni}}, \bibinfo {author} {\bibfnamefont {A.}~\bibnamefont {Cabello}}, \bibinfo {author} {\bibfnamefont {O.}~\bibnamefont {G{\"u}hne}}, \bibinfo {author} {\bibfnamefont {M.}~\bibnamefont {Kleinmann}},\ and\ \bibinfo {author} {\bibfnamefont {J.-{\AA}.}\ \bibnamefont {Larsson}},\ }\bibfield  {title} {\bibinfo {title} {Kochen-specker contextuality},\ }\href {https://doi.org/10.1103/RevModPhys.94.045007} {\bibfield  {journal} {\bibinfo  {journal} {Reviews of Modern Physics}\ }\textbf {\bibinfo {volume} {94}},\ \bibinfo {pages} {045007} (\bibinfo {year} {2022})}\BibitemShut {NoStop}%
\bibitem [{\citenamefont {Kimura}(2003)}]{kimura2003bloch}%
  \BibitemOpen
  \bibfield  {author} {\bibinfo {author} {\bibfnamefont {G.}~\bibnamefont {Kimura}},\ }\bibfield  {title} {\bibinfo {title} {The bloch vector for n-level systems},\ }\href {https://doi.org/10.1016/s0375-9601(03)00941-1} {\bibfield  {journal} {\bibinfo  {journal} {Physics Letters A}\ }\textbf {\bibinfo {volume} {314}},\ \bibinfo {pages} {339–349} (\bibinfo {year} {2003})}\BibitemShut {NoStop}%
\bibitem [{\citenamefont {Bertlmann}\ and\ \citenamefont {Krammer}(2008)}]{bertlmann2008blochvectors}%
  \BibitemOpen
  \bibfield  {author} {\bibinfo {author} {\bibfnamefont {R.~A.}\ \bibnamefont {Bertlmann}}\ and\ \bibinfo {author} {\bibfnamefont {P.}~\bibnamefont {Krammer}},\ }\bibfield  {title} {\bibinfo {title} {Bloch vectors for qudits},\ }\href {https://doi.org/10.1088/1751-8113/41/23/235303} {\bibfield  {journal} {\bibinfo  {journal} {Journal of Physics A: Mathematical and Theoretical}\ }\textbf {\bibinfo {volume} {41}},\ \bibinfo {pages} {235303} (\bibinfo {year} {2008})}\BibitemShut {NoStop}%
\bibitem [{\citenamefont {Kosmann-Schwarzbach}(2022)}]{Kosmann-Schwarzbach:2022}%
  \BibitemOpen
  \bibfield  {author} {\bibinfo {author} {\bibfnamefont {Y.}~\bibnamefont {Kosmann-Schwarzbach}},\ }\href@noop {} {\emph {\bibinfo {title} {Groups and Symmetries From Finite Groups to Lie Groups}}},\ \bibinfo {edition} {2nd}\ ed.\ (\bibinfo  {publisher} {Springer Cham},\ \bibinfo {year} {2022})\BibitemShut {NoStop}%
\bibitem [{\citenamefont {Gell-Mann}(1962)}]{gellmann1962symmetries}%
  \BibitemOpen
  \bibfield  {author} {\bibinfo {author} {\bibfnamefont {M.}~\bibnamefont {Gell-Mann}},\ }\bibfield  {title} {\bibinfo {title} {{Symmetries of Baryons and Mesons}},\ }\href {https://doi.org/10.1103/PhysRev.125.1067} {\bibfield  {journal} {\bibinfo  {journal} {Phys. Rev.}\ }\textbf {\bibinfo {volume} {125}},\ \bibinfo {pages} {1067} (\bibinfo {year} {1962})}\BibitemShut {NoStop}%
\bibitem [{\citenamefont {Tavakoli}\ \emph {et~al.}(2024)\citenamefont {Tavakoli}, \citenamefont {Pozas-Kerstjens}, \citenamefont {Brown},\ and\ \citenamefont {Ara\'ujo}}]{tavakoli2023semidefinite}%
  \BibitemOpen
  \bibfield  {author} {\bibinfo {author} {\bibfnamefont {A.}~\bibnamefont {Tavakoli}}, \bibinfo {author} {\bibfnamefont {A.}~\bibnamefont {Pozas-Kerstjens}}, \bibinfo {author} {\bibfnamefont {P.}~\bibnamefont {Brown}},\ and\ \bibinfo {author} {\bibfnamefont {M.}~\bibnamefont {Ara\'ujo}},\ }\bibfield  {title} {\bibinfo {title} {Semidefinite programming relaxations for quantum correlations},\ }\href {https://doi.org/10.1103/RevModPhys.96.045006} {\bibfield  {journal} {\bibinfo  {journal} {Rev. Mod. Phys.}\ }\textbf {\bibinfo {volume} {96}},\ \bibinfo {pages} {045006} (\bibinfo {year} {2024})}\BibitemShut {NoStop}%
\bibitem [{\citenamefont {Matsumoto}(2002)}]{matsumoto2002newapproach}%
  \BibitemOpen
  \bibfield  {author} {\bibinfo {author} {\bibfnamefont {K.}~\bibnamefont {Matsumoto}},\ }\bibfield  {title} {\bibinfo {title} {{A new approach to the Cramér-Rao-type bound of the pure-state model}},\ }\href {https://doi.org/10.1088/0305-4470/35/13/307} {\bibfield  {journal} {\bibinfo  {journal} {Journal of Physics A: Mathematical and General}\ }\textbf {\bibinfo {volume} {35}},\ \bibinfo {pages} {3111–3123} (\bibinfo {year} {2002})}\BibitemShut {NoStop}%
\bibitem [{\citenamefont {Reascos}\ \emph {et~al.}(2023)\citenamefont {Reascos}, \citenamefont {Murta}, \citenamefont {Galv\~ao},\ and\ \citenamefont {Fern\'andez-Rossier}}]{reascos2023quantum}%
  \BibitemOpen
  \bibfield  {author} {\bibinfo {author} {\bibfnamefont {L.~I.}\ \bibnamefont {Reascos}}, \bibinfo {author} {\bibfnamefont {B.}~\bibnamefont {Murta}}, \bibinfo {author} {\bibfnamefont {E.~F.}\ \bibnamefont {Galv\~ao}},\ and\ \bibinfo {author} {\bibfnamefont {J.}~\bibnamefont {Fern\'andez-Rossier}},\ }\bibfield  {title} {\bibinfo {title} {Quantum circuits to measure scalar spin chirality},\ }\href {https://doi.org/10.1103/PhysRevResearch.5.043087} {\bibfield  {journal} {\bibinfo  {journal} {Phys. Rev. Res.}\ }\textbf {\bibinfo {volume} {5}},\ \bibinfo {pages} {043087} (\bibinfo {year} {2023})}\BibitemShut {NoStop}%
\bibitem [{\citenamefont {Kimbell}\ \emph {et~al.}(2022)\citenamefont {Kimbell}, \citenamefont {Kim}, \citenamefont {Wu}, \citenamefont {Cuoco},\ and\ \citenamefont {Robinson}}]{kimbell2022challenges}%
  \BibitemOpen
  \bibfield  {author} {\bibinfo {author} {\bibfnamefont {G.}~\bibnamefont {Kimbell}}, \bibinfo {author} {\bibfnamefont {C.}~\bibnamefont {Kim}}, \bibinfo {author} {\bibfnamefont {W.}~\bibnamefont {Wu}}, \bibinfo {author} {\bibfnamefont {M.}~\bibnamefont {Cuoco}},\ and\ \bibinfo {author} {\bibfnamefont {J.~W.~A.}\ \bibnamefont {Robinson}},\ }\bibfield  {title} {\bibinfo {title} {Challenges in identifying chiral spin textures via the topological {H}all effect},\ }\href {https://doi.org/10.1038/s43246-022-00238-2} {\bibfield  {journal} {\bibinfo  {journal} {Communications Materials}\ }\textbf {\bibinfo {volume} {3}},\ \bibinfo {pages} {19} (\bibinfo {year} {2022})}\BibitemShut {NoStop}%
\bibitem [{\citenamefont {Berry}(1984)}]{berry1984quantal}%
  \BibitemOpen
  \bibfield  {author} {\bibinfo {author} {\bibfnamefont {M.~V.}\ \bibnamefont {Berry}},\ }\bibfield  {title} {\bibinfo {title} {Quantal phase factors accompanying adiabatic changes},\ }\href {https://doi.org/https://doi.org/10.1098/rspa.1984.0023} {\bibfield  {journal} {\bibinfo  {journal} {Proceedings of the Royal Society of London. A. Mathematical and Physical Sciences}\ }\textbf {\bibinfo {volume} {392}},\ \bibinfo {pages} {45} (\bibinfo {year} {1984})}\BibitemShut {NoStop}%
\bibitem [{\citenamefont {Ye}\ \emph {et~al.}(1999)\citenamefont {Ye}, \citenamefont {Kim}, \citenamefont {Millis}, \citenamefont {Shraiman}, \citenamefont {Majumdar},\ and\ \citenamefont {Te\ifmmode \check{s}\else \v{s}\fi{}anovi\ifmmode~\acute{c}\else \'{c}\fi{}}}]{ye1999berry}%
  \BibitemOpen
  \bibfield  {author} {\bibinfo {author} {\bibfnamefont {J.}~\bibnamefont {Ye}}, \bibinfo {author} {\bibfnamefont {Y.~B.}\ \bibnamefont {Kim}}, \bibinfo {author} {\bibfnamefont {A.~J.}\ \bibnamefont {Millis}}, \bibinfo {author} {\bibfnamefont {B.~I.}\ \bibnamefont {Shraiman}}, \bibinfo {author} {\bibfnamefont {P.}~\bibnamefont {Majumdar}},\ and\ \bibinfo {author} {\bibfnamefont {Z.}~\bibnamefont {Te\ifmmode \check{s}\else \v{s}\fi{}anovi\ifmmode~\acute{c}\else \'{c}\fi{}}},\ }\bibfield  {title} {\bibinfo {title} {{Berry Phase Theory of the Anomalous Hall Effect: Application to Colossal Magnetoresistance Manganites}},\ }\href {https://doi.org/10.1103/PhysRevLett.83.3737} {\bibfield  {journal} {\bibinfo  {journal} {Phys. Rev. Lett.}\ }\textbf {\bibinfo {volume} {83}},\ \bibinfo {pages} {3737} (\bibinfo {year} {1999})}\BibitemShut {NoStop}%
\bibitem [{\citenamefont {Simon}\ and\ \citenamefont {Mukunda}(1993)}]{simon1993bargmann}%
  \BibitemOpen
  \bibfield  {author} {\bibinfo {author} {\bibfnamefont {R.}~\bibnamefont {Simon}}\ and\ \bibinfo {author} {\bibfnamefont {N.}~\bibnamefont {Mukunda}},\ }\bibfield  {title} {\bibinfo {title} {{Bargmann invariant and the geometry of the G\"uoy effect}},\ }\href {https://doi.org/10.1103/PhysRevLett.70.880} {\bibfield  {journal} {\bibinfo  {journal} {Phys. Rev. Lett.}\ }\textbf {\bibinfo {volume} {70}},\ \bibinfo {pages} {880} (\bibinfo {year} {1993})}\BibitemShut {NoStop}%
\bibitem [{\citenamefont {Kuroiwa}\ \emph {et~al.}(2024{\natexlab{a}})\citenamefont {Kuroiwa}, \citenamefont {Takagi}, \citenamefont {Adesso},\ and\ \citenamefont {Yamasaki}}]{kuroiwa2024every}%
  \BibitemOpen
  \bibfield  {author} {\bibinfo {author} {\bibfnamefont {K.}~\bibnamefont {Kuroiwa}}, \bibinfo {author} {\bibfnamefont {R.}~\bibnamefont {Takagi}}, \bibinfo {author} {\bibfnamefont {G.}~\bibnamefont {Adesso}},\ and\ \bibinfo {author} {\bibfnamefont {H.}~\bibnamefont {Yamasaki}},\ }\bibfield  {title} {\bibinfo {title} {{Every Quantum Helps: Operational Advantage of Quantum Resources beyond Convexity}},\ }\href {https://doi.org/10.1103/PhysRevLett.132.150201} {\bibfield  {journal} {\bibinfo  {journal} {Phys. Rev. Lett.}\ }\textbf {\bibinfo {volume} {132}},\ \bibinfo {pages} {150201} (\bibinfo {year} {2024}{\natexlab{a}})}\BibitemShut {NoStop}%
\bibitem [{\citenamefont {Kuroiwa}\ \emph {et~al.}(2024{\natexlab{b}})\citenamefont {Kuroiwa}, \citenamefont {Takagi}, \citenamefont {Adesso},\ and\ \citenamefont {Yamasaki}}]{kuroiwa2024robustness}%
  \BibitemOpen
  \bibfield  {author} {\bibinfo {author} {\bibfnamefont {K.}~\bibnamefont {Kuroiwa}}, \bibinfo {author} {\bibfnamefont {R.}~\bibnamefont {Takagi}}, \bibinfo {author} {\bibfnamefont {G.}~\bibnamefont {Adesso}},\ and\ \bibinfo {author} {\bibfnamefont {H.}~\bibnamefont {Yamasaki}},\ }\bibfield  {title} {\bibinfo {title} {Robustness- and weight-based resource measures without convexity restriction: Multicopy witness and operational advantage in static and dynamical quantum resource theories},\ }\href {https://doi.org/10.1103/PhysRevA.109.042403} {\bibfield  {journal} {\bibinfo  {journal} {Phys. Rev. A}\ }\textbf {\bibinfo {volume} {109}},\ \bibinfo {pages} {042403} (\bibinfo {year} {2024}{\natexlab{b}})}\BibitemShut {NoStop}%
\bibitem [{\citenamefont {Takagi}\ \emph {et~al.}(2019)\citenamefont {Takagi}, \citenamefont {Regula}, \citenamefont {Bu}, \citenamefont {Liu},\ and\ \citenamefont {Adesso}}]{takagi2019operational}%
  \BibitemOpen
  \bibfield  {author} {\bibinfo {author} {\bibfnamefont {R.}~\bibnamefont {Takagi}}, \bibinfo {author} {\bibfnamefont {B.}~\bibnamefont {Regula}}, \bibinfo {author} {\bibfnamefont {K.}~\bibnamefont {Bu}}, \bibinfo {author} {\bibfnamefont {Z.-W.}\ \bibnamefont {Liu}},\ and\ \bibinfo {author} {\bibfnamefont {G.}~\bibnamefont {Adesso}},\ }\bibfield  {title} {\bibinfo {title} {Operational advantage of quantum resources in subchannel discrimination},\ }\href {https://doi.org/10.1103/PhysRevLett.122.140402} {\bibfield  {journal} {\bibinfo  {journal} {Phys. Rev. Lett.}\ }\textbf {\bibinfo {volume} {122}},\ \bibinfo {pages} {140402} (\bibinfo {year} {2019})}\BibitemShut {NoStop}%
\bibitem [{\citenamefont {Napoli}\ \emph {et~al.}(2016)\citenamefont {Napoli}, \citenamefont {Bromley}, \citenamefont {Cianciaruso}, \citenamefont {Piani}, \citenamefont {Johnston},\ and\ \citenamefont {Adesso}}]{napoli2016robustness}%
  \BibitemOpen
  \bibfield  {author} {\bibinfo {author} {\bibfnamefont {C.}~\bibnamefont {Napoli}}, \bibinfo {author} {\bibfnamefont {T.~R.}\ \bibnamefont {Bromley}}, \bibinfo {author} {\bibfnamefont {M.}~\bibnamefont {Cianciaruso}}, \bibinfo {author} {\bibfnamefont {M.}~\bibnamefont {Piani}}, \bibinfo {author} {\bibfnamefont {N.}~\bibnamefont {Johnston}},\ and\ \bibinfo {author} {\bibfnamefont {G.}~\bibnamefont {Adesso}},\ }\bibfield  {title} {\bibinfo {title} {Robustness of coherence: An operational and observable measure of quantum coherence},\ }\href {https://doi.org/10.1103/PhysRevLett.116.150502} {\bibfield  {journal} {\bibinfo  {journal} {Phys. Rev. Lett.}\ }\textbf {\bibinfo {volume} {116}},\ \bibinfo {pages} {150502} (\bibinfo {year} {2016})}\BibitemShut {NoStop}%
\bibitem [{\citenamefont {Piani}\ \emph {et~al.}(2016)\citenamefont {Piani}, \citenamefont {Cianciaruso}, \citenamefont {Bromley}, \citenamefont {Napoli}, \citenamefont {Johnston},\ and\ \citenamefont {Adesso}}]{piani2016robustness}%
  \BibitemOpen
  \bibfield  {author} {\bibinfo {author} {\bibfnamefont {M.}~\bibnamefont {Piani}}, \bibinfo {author} {\bibfnamefont {M.}~\bibnamefont {Cianciaruso}}, \bibinfo {author} {\bibfnamefont {T.~R.}\ \bibnamefont {Bromley}}, \bibinfo {author} {\bibfnamefont {C.}~\bibnamefont {Napoli}}, \bibinfo {author} {\bibfnamefont {N.}~\bibnamefont {Johnston}},\ and\ \bibinfo {author} {\bibfnamefont {G.}~\bibnamefont {Adesso}},\ }\bibfield  {title} {\bibinfo {title} {Robustness of asymmetry and coherence of quantum states},\ }\href {https://doi.org/10.1103/PhysRevA.93.042107} {\bibfield  {journal} {\bibinfo  {journal} {Phys. Rev. A}\ }\textbf {\bibinfo {volume} {93}},\ \bibinfo {pages} {042107} (\bibinfo {year} {2016})}\BibitemShut {NoStop}%
\bibitem [{\citenamefont {Skrzypczyk}\ \emph {et~al.}(2019)\citenamefont {Skrzypczyk}, \citenamefont {\ifmmode \check{S}\else \v{S}\fi{}upi\ifmmode~\acute{c}\else \'{c}\fi{}},\ and\ \citenamefont {Cavalcanti}}]{skrzypczyk2019all}%
  \BibitemOpen
  \bibfield  {author} {\bibinfo {author} {\bibfnamefont {P.}~\bibnamefont {Skrzypczyk}}, \bibinfo {author} {\bibfnamefont {I.}~\bibnamefont {\ifmmode \check{S}\else \v{S}\fi{}upi\ifmmode~\acute{c}\else \'{c}\fi{}}},\ and\ \bibinfo {author} {\bibfnamefont {D.}~\bibnamefont {Cavalcanti}},\ }\bibfield  {title} {\bibinfo {title} {All sets of incompatible measurements give an advantage in quantum state discrimination},\ }\href {https://doi.org/10.1103/PhysRevLett.122.130403} {\bibfield  {journal} {\bibinfo  {journal} {Phys. Rev. Lett.}\ }\textbf {\bibinfo {volume} {122}},\ \bibinfo {pages} {130403} (\bibinfo {year} {2019})}\BibitemShut {NoStop}%
\bibitem [{\citenamefont {Oszmaniec}\ and\ \citenamefont {Biswas}(2019)}]{oszmaniec2019operational}%
  \BibitemOpen
  \bibfield  {author} {\bibinfo {author} {\bibfnamefont {M.}~\bibnamefont {Oszmaniec}}\ and\ \bibinfo {author} {\bibfnamefont {T.}~\bibnamefont {Biswas}},\ }\bibfield  {title} {\bibinfo {title} {Operational relevance of resource theories of quantum measurements},\ }\href {https://doi.org/10.22331/q-2019-04-26-133} {\bibfield  {journal} {\bibinfo  {journal} {{Quantum}}\ }\textbf {\bibinfo {volume} {3}},\ \bibinfo {pages} {133} (\bibinfo {year} {2019})}\BibitemShut {NoStop}%
\bibitem [{\citenamefont {Takagi}\ and\ \citenamefont {Regula}(2019)}]{takagi2019general}%
  \BibitemOpen
  \bibfield  {author} {\bibinfo {author} {\bibfnamefont {R.}~\bibnamefont {Takagi}}\ and\ \bibinfo {author} {\bibfnamefont {B.}~\bibnamefont {Regula}},\ }\bibfield  {title} {\bibinfo {title} {General resource theories in quantum mechanics and beyond: Operational characterization via discrimination tasks},\ }\href {https://doi.org/10.1103/PhysRevX.9.031053} {\bibfield  {journal} {\bibinfo  {journal} {Phys. Rev. X}\ }\textbf {\bibinfo {volume} {9}},\ \bibinfo {pages} {031053} (\bibinfo {year} {2019})}\BibitemShut {NoStop}%
\end{thebibliography}%

\onecolumngrid
\begin{appendix}

\section{Proof of Theorem~\ref{theorem:Bargmanns_dependence_overlaps}}\label{app:proof_1}

\begin{proof}
We proceed similarly to Appendix B from Ref.~\cite{zhang2025geometrysets}. For any single-qubit state $\rho \in \mathcal{D}\left(\mathbb{C}^2\right)$ we write $\rho=\rho(\mathbf{r})=\frac{1}{2}\left(\mathbb{1}_2+\mathbf{r} \cdot \boldsymbol{\sigma}\right)$ using the Bloch representation. Define $$\boldsymbol{\mathrm{Prod}}_n=\rho_1 \,\rho_2 \,\cdots \,\rho_n=2^{-n}\left(p_0^{(n)} \mathbb{1}_2+\boldsymbol{p}^{(n)}\boldsymbol{\sigma}\right),$$ where we set $p_0^{(1)}=1$  and $\boldsymbol{p}^{(1)}=\mathbf{r}_1$ the Bloch vector of state $\rho_1$. Recursively, $\boldsymbol{\mathrm{Prod}}_{n+1}=\boldsymbol{\mathrm{Prod}}_n\,\rho_{n+1}$, which implies the following update rules:
$$
\begin{aligned}
	& p_0^{(n+1)}=p_0^{(n)}+\left\langle\boldsymbol{p}^{(n)}, \mathbf{r}_{n+1}\right\rangle, \\
	& \boldsymbol{p}^{(n+1)}=p_0^{(n)} \mathbf{r}_{n+1}+\boldsymbol{p}^{(n)}+\mathrm{i} \boldsymbol{p}^{(n)} \times \mathbf{r}_{n+1}.
\end{aligned}
$$
If we write $ p_0^{(n)}=a_0^{(n)}+\mathrm{i}b_0^{(n)}$  and  $\boldsymbol{p}^{(n)}=\boldsymbol{a}^{(n)}+\mathrm{i}
\boldsymbol{b}^{(n)},$ then we have the following recursive relations:
\begin{equation}
\begin{aligned}
	& a_0^{(n+1)}=a_0^{(n)}+\left\langle\boldsymbol{a}^{(n)}, \mathbf{r}_{n+1}\right\rangle, \\
&	b_0^{(n+1)}=b_0^{(n)}+\left\langle\boldsymbol{b}^{(n)}, \mathbf{r}_{n+1}\right\rangle, \\
	& \boldsymbol{a}^{(n+1)}=a_0^{(n)} \mathbf{r}_{n+1}+\boldsymbol{a}^{(n)}- \boldsymbol{b}^{(n)} \times \mathbf{r}_{n+1},\\
		& \boldsymbol{b}^{(n+1)}=b_0^{(n)} \mathbf{r}_{n+1}+\boldsymbol{b}^{(n)}+ \boldsymbol{a}^{(n)} \times \mathbf{r}_{n+1},
\end{aligned}
\end{equation}
with initial conditions $a_0^{(1)}=1,b_0^{(1)}=0,$   $\boldsymbol{a}^{(1)}=\mathbf{r}_1, \boldsymbol{b}^{(1)}=\boldsymbol{0}.$ 

Let us now consider the following algebraic constructions: 
\begin{equation}
\begin{aligned}
	& R_n:=\mathbb{Z}\left[\left\langle\mathbf{r}_i, \mathbf{r}_j\right\rangle, 1\leq i,j\leq n\right],\\	
	 &S_n:=\sum_{1\leq i<j<k\leq n}  R_n \det[\mathbf{r}_i,\mathbf{r}_j,\mathbf{r}_k]  ,\\
		& T_n:=\sum_{k=1}^n R_n  \mathbf{r}_k+\sum_{1\leq i<j\leq n}  S_n\mathbf{r}_{i}\times  \mathbf{r}_{j}, \\
	& U_n:=\sum_{k=1}^n S_n\mathbf{r}_{k}  +\sum_{1\leq i<j\leq n} R_n  \mathbf{r}_i\times \mathbf{r}_j. 
\end{aligned}
\end{equation}
{Here, \( R_n \) denotes the ring of polynomials over \( \mathbb{Z} \) generated by the inner products \( \langle \mathbf{r}_i, \mathbf{r}_j \rangle \) for \( 1 \leq i,j \leq n \), denoted as $\mathbb{Z}[\langle \mathbf{r}_i,\mathbf{r}_j\rangle,1\leq i,j \leq n]$. The space \( S_n \) is linearly generated by \(\det[\mathbf{r}_i, \mathbf{r}_j, \mathbf{r}_k]\) for \(1 \leq i \leq j \leq k \leq n\), with coefficients in \( R_n \). Similarly for \( T_n \) and \( U_n \).} In the following, we will prove by induction that for all $n\in \mathbb{N}$
$$ 
\begin{aligned}
	  a_0^{(n)}\in R_n,   
 	b_0^{(n)}\in S_n,   
	  \boldsymbol{a}^{(n)}\in T_n, 
   \boldsymbol{b}^{(n)}\in U_n.
\end{aligned}
$$
Before the proof, we recall the following properties:
\begin{enumerate}[(A)]	\item The product $\det[\mathbf{r}_{i'},\mathbf{r}_{j'},\mathbf{r}_{k'}] \det[\mathbf{r}_{i },\mathbf{r}_{j },\mathbf{r}_{k }]$ equals to
 $$ \left| \begin{array}{ccc}
		\langle \mathbf{r}_{i'}, \mathbf{r}_{i }\rangle& \langle\mathbf{r}_{i'}, \mathbf{r}_{j }\rangle&  	\langle\mathbf{r}_{i'}, \mathbf{r}_{k }\rangle\\
		\langle \mathbf{r}_{j'}, \mathbf{r}_{i }\rangle& \langle\mathbf{r}_{j'}, \mathbf{r}_{j }\rangle&  	\langle\mathbf{r}_{j'}, \mathbf{r}_{k }\rangle\\
		\langle \mathbf{r}_{k'}, \mathbf{r}_{i }\rangle& \langle\mathbf{r}_{k'}, \mathbf{r}_{j }\rangle&  	\langle\mathbf{r}_{k'}, \mathbf{r}_{k }\rangle\\
		\end{array}\right|.$$
	To show this, it suffices to use the simple properties of determinants $\det[A]=\det[A^T]$ and $\det[AB]=\det[A]\det[B].$
	\item  $\langle \mathbf{r}_{i}\times \mathbf{r}_{j }, \mathbf{r}_{k} \rangle=\det[\mathbf{r}_{i},\mathbf{r}_{j},\mathbf{r}_{k}].$
	\item $(\mathbf{r}_{i}\times \mathbf{r}_{j }) \times  \mathbf{r}_{k }= \langle \mathbf{r}_{i},\mathbf{r}_{k} \rangle \mathbf{r}_{j} -\langle \mathbf{r}_{j},\mathbf{r}_{k} \rangle \mathbf{r}_{i}.$
	\end{enumerate}
Clearly, $a_0^{(1)}\in R_1, b_0^{(1)}\in S_1, \boldsymbol{a}^{(1)}\in T_1, \boldsymbol{b}^{(1)}\in U_1$ holds (which provides our base step for $n=1$). Assume this holds for $n=l$, now we need to prove it also holds for $n=l+1$. 
\begin{enumerate}
	\item (Inductive step for $R_n$)   As  $a_0^{(l+1)}=a_0^{(l)}+\left\langle\boldsymbol{a}^{(l)}, \mathbf{r}_{l+1}\right\rangle,$  and $\boldsymbol{a}^{(l)}\in T_l,$ so 
	$$\boldsymbol{a}^{(l)}=\sum_{k=1}^l r_k\mathbf{r}_k+\sum_{1\leq i<j\leq l} s_{ij} \mathbf{r}_i \times   \mathbf{r}_j.$$ Clearly, for each $1\leq k\leq l$, $\left\langle\mathbf{r}_k, \mathbf{r}_{l+1}\right\rangle\in R_{l+1}$  and for each $s_{ij}=r_{i'j'k'} \det[\mathbf{r}_{i'},\mathbf{r}_{j'},\mathbf{r}_{k'}]$ where $1\leq i'<j'<k'\leq l$ and $r_{i'j'k'}\in R_l,$ 
	 by using properties (A) and (B), we can show that
	 $$ \langle s_{ij}\mathbf{r}_i\times \mathbf{r}_j , \mathbf{r}_{l+1} \rangle \in R_{l+1}.$$
	 So we get $a_0^{(l+1)}\in R_{l+1}.$
	 
	 \item (Inductive step for $S_n$) As $b_0^{(l+1)}=b_0^{(l)}+\left\langle\boldsymbol{b}^{(l)}, \mathbf{r}_{l+1}\right\rangle $ and $\boldsymbol{b}^{(l)}\in T_l$, so 
	 $$\boldsymbol{b}^{(l)}=\sum_{k=1}^l s_k\mathbf{r}_k+\sum_{1\leq i<j\leq l} r_{ij} \mathbf{r}_i \times   \mathbf{r}_j.$$ By definition $\langle s_k\mathbf{r}_k, \mathbf{r}_{l+1}\rangle=\langle \mathbf{r}_k, \mathbf{r}_{l+1}\rangle s_k \in S_{l+1}.$ And the property (B) implies that 
	 $$ \langle r_{ij} \mathbf{r}_i \times   \mathbf{r}_j,\mathbf{r}_{l+1}\rangle \in S_{l+1}.$$ 
	 Therefore, $b_0^{(l+1)}\in S_{l+1}.$

	  \item (Inductive step for $T_n$) As $\boldsymbol{a}^{(l+1)}=a_0^{(l)} \mathbf{r}_{l+1}+\boldsymbol{a}^{(l)}- \boldsymbol{b}^{(l)} \times \mathbf{r}_{l+1},$  and $a_0^{(l)}\in R_l,\boldsymbol{a}^{(l)}\in T_l, \boldsymbol{b}^{(l)}\in U_l$, we have $a_0^{(l)} \mathbf{r}_{l+1}+\boldsymbol{a}^{(l)}\in T_{l+1} $
	 and  
	 $$\boldsymbol{b}^{(l)}=\sum_{k=1}^l s_k\mathbf{r}_k+\sum_{1\leq i<j\leq l} r_{ij} \mathbf{r}_i \times   \mathbf{r}_j.$$
	 So by the definition of $T_{l+1}$ and Property (C), we get
	 $$  \boldsymbol{b}^{(l)} \times   \mathbf{r}_{l+1} \in T_{l+1}.$$
	 Hence $\boldsymbol{a}^{(l+1)}\in T_{l+1}.$
	  
	 \item (Inductive step for $U_n$) As $\boldsymbol{b}^{(l+1)}=b_0^{(l)} \mathbf{r}_{l+1}+\boldsymbol{b}^{(l)}+ \boldsymbol{a}^{(l)} \times \mathbf{r}_{l+1},$  and $b_0^{(l)}\in S_l,\boldsymbol{a}^{(l)}\in T_l, \boldsymbol{b}^{(l)}\in U_l$, we have $b_0^{(l)} \mathbf{r}_{l+1}+\boldsymbol{b}^{(l)}\in U_{l+1} $
	 and  
	 $$\boldsymbol{a}^{(l)}=\sum_{k=1}^l r_k\mathbf{r}_k+\sum_{1\leq i<j\leq l} s_{ij} \mathbf{r}_i \times   \mathbf{r}_j.$$
	 So by the definition of $U_{l+1}$ and Property (C), we get
	 $$  \boldsymbol{a}^{(l)} \times   \mathbf{r}_{l+1} \in \textsc{U}_{l+1}.$$
	 Hence $\boldsymbol{b}^{(l+1)}\in U_{l+1}.$
	  
\end{enumerate}

The Bargmann invariant for the tuple   $\varrho = (\rho_1, \ldots, \rho_n)$, where $\rho_k=\rho\left(\mathbf{r}_k\right)$, is thus given by $\Delta(\varrho):=$ $\operatorname{Tr}\left(\boldsymbol{\mathrm{Prod}}_n\right)=2^{1-n} p_0^{(n)}=\frac{a_0^{(n)}+\mathrm{i}b_0^{(n)} }{2^{n-1}}$. Note that 
$$a_0^{(n)}\in R_n, \text{ and } b_0^{(n)}\in S_n.$$
Using property (A), one obtains $S_n*S_n\in R_n$. In particular, 
$$\left(a_0^{(n)}\right)^2+\left(b_0^{(n)}\right)^2 \in R_n.$$
That is, the real part of $\mathrm{Tr}[\rho_1\cdots \rho_n]$ and the power of two of the imaginarity of  $\mathrm{Tr}[\rho_1\cdots \rho_n]$ are determined by 
$$ \langle \mathbf{r}_i,\mathbf{r}_j\rangle= 2\mathrm{Tr}[\rho_i\rho_j]-1, 1\leq i,j\leq n.$$
 The corresponding polynomial $P_n$  can be got by replacing the  $\langle \mathbf{r}_i,\mathbf{r}_j\rangle$   in $\frac{1}{2^{n-1}}a_0^{(n)}$ with $2\Delta_{ij}-1$.  The corresponding polynomial $Q_n$  can be got by replacing the  $\langle \mathbf{r}_i,\mathbf{r}_j\rangle$   in $\frac{\left(a_0^{(n)}\right)^2+\left(b_0^{(n)}\right)^2}{2^{2n-2}}$ with $2\Delta_{ij}-1$.
\end{proof}

As a concrete example, the first three steps $n=3,4,$ and $5$ have been determined in Ref.~\cite{zhang2025geometrysets} and are given by: 
$$ 
\left\{\begin{array}{l}
a_0^{(3)}=  1+\sum_{1 \leqslant i<j \leqslant 3}\left\langle\mathbf{r}_i, \mathbf{r}_j\right\rangle, \\
b_0^{(3)}= \operatorname{det}\left(\mathbf{r}_1, \mathbf{r}_2, \mathbf{r}_3\right),
\end{array}\right.
$$ 
$$
\left\{\begin{array}{l}
a_0^{(4)}= \left(1+\left\langle\mathbf{r}_1, \mathbf{r}_2\right\rangle\right)\left(1+\left\langle\mathbf{r}_3, \mathbf{r}_4\right\rangle\right)-\left(1-\left\langle\mathbf{r}_1, \mathbf{r}_3\right\rangle\right)\left(1-\left\langle\mathbf{r}_2, \mathbf{r}_4\right\rangle\right)+\left(1+\left\langle\mathbf{r}_1, \mathbf{r}_4\right\rangle\right)\left(1+\left\langle\mathbf{r}_2, \mathbf{r}_3\right\rangle\right), \\
b_0^{(4)}=  \operatorname{det}\left(\mathbf{r}_1+\mathbf{r}_2, \mathbf{r}_2+\mathbf{r}_3, \mathbf{r}_3+\mathbf{r}_4\right) ,
\end{array}\right.
$$
and 
$$
\left\{\begin{aligned}
a_0^{(5)}=  &   1+\sum_{1 \leqslant i<j \leqslant 5}\left\langle\mathbf{r}_i, \mathbf{r}_j\right\rangle+\left\langle\mathbf{r}_1, \mathbf{r}_2\right\rangle\left\langle\mathbf{r}_3, \mathbf{r}_4\right\rangle-\left\langle\mathbf{r}_1, \mathbf{r}_3\right\rangle\left\langle\mathbf{r}_2, \mathbf{r}_4\right\rangle \\
& +\left\langle\mathbf{r}_1, \mathbf{r}_4\right\rangle\left\langle\mathbf{r}_2, \mathbf{r}_3\right\rangle+\left(\left\langle\mathbf{r}_2, \mathbf{r}_3\right\rangle+\left\langle\mathbf{r}_2, \mathbf{r}_4\right\rangle+\left\langle\mathbf{r}_3, \mathbf{r}_4\right\rangle\right)\left\langle\mathbf{r}_1, \mathbf{r}_5\right\rangle \\
& +\left(-\left\langle\mathbf{r}_1, \mathbf{r}_3\right\rangle-\left\langle\mathbf{r}_1, \mathbf{r}_4\right\rangle+\left\langle\mathbf{r}_3, \mathbf{r}_4\right\rangle\right)\left\langle\mathbf{r}_2, \mathbf{r}_5\right\rangle \\
& +\left(\left\langle\mathbf{r}_1, \mathbf{r}_2\right\rangle-\left\langle\mathbf{r}_1, \mathbf{r}_4\right\rangle-\left\langle\mathbf{r}_2, \mathbf{r}_4\right\rangle\right)\left\langle\mathbf{r}_3, \mathbf{r}_5\right\rangle \\
&  +\left(\left\langle\mathbf{r}_1, \mathbf{r}_2\right\rangle+\left\langle\mathbf{r}_1, \mathbf{r}_3\right\rangle+\left\langle\mathbf{r}_2, \mathbf{r}_3\right\rangle\right)\left\langle\mathbf{r}_4, \mathbf{r}_5\right\rangle , \\
b_0^{(5)}= &   \sum_{1 \leqslant i<j<k \leqslant 5}\left\langle\mathbf{r}_i \times \mathbf{r}_j, \mathbf{r}_k\right\rangle+\left\langle\mathbf{r}_2, \mathbf{r}_3\right\rangle\left\langle\mathbf{r}_1 \times \mathbf{r}_4, \mathbf{r}_5\right\rangle  \\
& -\left\langle\mathbf{r}_1, \mathbf{r}_3\right\rangle\left\langle\mathbf{r}_2 \times \mathbf{r}_4, \mathbf{r}_5\right\rangle+\left\langle\mathbf{r}_1, \mathbf{r}_2\right\rangle\left\langle\mathbf{r}_3 \times \mathbf{r}_4, \mathbf{r}_5\right\rangle  +\left\langle\mathbf{r}_4, \mathbf{r}_5\right\rangle\left\langle\mathbf{r}_1 \times \mathbf{r}_2, \mathbf{r}_3\right\rangle  .
\end{aligned}\right.
$$

\end{appendix}
 
\end{document}